\documentclass[11pt,onecolumn]{article}
\usepackage{amsmath,amsthm,amssymb,amsfonts,epsfig,graphicx}
\usepackage{color}
\usepackage{enumitem}
\usepackage{url}
\usepackage{array,multirow}
\usepackage{hyperref}

\usepackage{stmaryrd}

\usepackage[english]{rccol}
\usepackage{arydshln}

\usepackage{enumitem}
\usepackage[vlined,linesnumbered,ruled,resetcount]{algorithm2e}

\usepackage[left=1.0in,top=0.8in,right=1.05in,bottom=0.5in,nohead,textheight=10in,footskip=0.3in]{geometry}

\newtheorem{theorem}{Theorem}
\newtheorem{corollary}[theorem]{Corollary}

\newtheorem{assumption}{Assumption}
\newtheorem{lemma}[theorem]{Lemma}
\newtheorem{definition}[theorem]{Definition}
\newtheorem{proposition}[theorem]{Proposition}
\newtheorem{remark}{Remark}

\newtheorem{lemmaRESTATED}{Lemma}

\SetKwInput{kwMyRepeat}{Repeat}

\begin{document}

\def\myparagraph#1{\vspace{2pt}\noindent{\bf #1~~}}



\def\Q{{Q}}

\long\def\ignore#1{}
\def\myps[#1]#2{\includegraphics[#1]{#2}}
\def\etal{{\em et al.}}
\def\Bar#1{{\bar #1}}
\def\br(#1,#2){{\langle #1,#2 \rangle}}
\def\setZ[#1,#2]{{[ #1 .. #2 ]}}
\def\Pr#1{\mbox{\tt Pr}\left[{#1}\right]}
\def\REACHED{\mbox{\tt REACHED}}
\def\AdjustFlow{\mbox{\tt AdjustFlow}}
\def\GetNeighbors{\mbox{\tt GetNeighbors}}
\def\true{\mbox{\tt true}}
\def\false{\mbox{\tt false}}
\def\Process{\mbox{\tt Process}}
\def\ProcessLeft{\mbox{\tt ProcessLeft}}
\def\ProcessRight{\mbox{\tt ProcessRight}}
\def\Add{\mbox{\tt Add}}

\newcommand{\eqdef}{{\stackrel{\mbox{\tiny \tt ~def~}}{=}}}

\def\setof#1{{\left\{#1\right\}}}
\def\suchthat#1#2{\setof{\,#1\mid#2\,}} 
\def\event#1{\setof{#1}}
\def\q={\quad=\quad}
\def\qq={\qquad=\qquad}
\def\calA{{\cal A}}
\def\calB{{\cal B}}
\def\calC{{\cal C}}
\def\calD{{\cal D}}
\def\calE{{\cal E}}
\def\calF{{\cal F}}
\def\calG{{\cal G}}
\def\calI{{\cal I}}
\def\calH{{H}}
\def\calL{{\cal L}}
\def\calN{{\cal N}}
\def\calP{{\cal P}}
\def\calR{{\cal R}}
\def\calS{{\cal S}}
\def\calT{{\cal T}}
\def\calU{{\cal U}}
\def\calV{{\cal V}}
\def\calO{{\cal O}}
\def\calX{{\cal X}}

\def\X{{\Omega}}
\def\E{{\calE}}
\def\T{{\calF}}
\def\U{{\Phi}}
\def\XX(#1){{{#1}^\downarrow}}

\def\calY{{\cal Y}}
\def\s{\footnotesize}
\def\calNG{{\cal N_G}}
\def\psfile[#1]#2{}
\def\psfilehere[#1]#2{}
\def\epsfw#1#2{\includegraphics[width=#1\hsize]{#2}}
\def\assign(#1,#2){\langle#1,#2\rangle}
\def\edge(#1,#2){(#1,#2)}
\def\VS{\calV^s}
\def\VT{\calV^t}
\def\slack(#1){\texttt{slack}({#1})}
\def\barslack(#1){\overline{\texttt{slack}}({#1})}
\def\NULL{\texttt{NULL}}
\def\PARENT{\texttt{PARENT}}
\def\GRANDPARENT{\texttt{GRANDPARENT}}
\def\TAIL{\texttt{TAIL}}
\def\HEADORIG{\texttt{HEAD$\_\:$ORIG}}
\def\TAILORIG{\texttt{TAIL$\_\:$ORIG}}
\def\HEAD{\texttt{HEAD}}
\def\CURRENTEDGE{\texttt{CURRENT$\!\_\:$EDGE}}

\def\unitvec(#1){{{\bf u}_{#1}}}
\def\uvec{{\bf u}}
\def\vvec{{\bf v}}
\def\Nvec{{\bf N}}

\newcommand{\bg}{\mbox{$\bf g$}}
\newcommand{\bh}{\mbox{$\bf h$}}

\newcommand{\bx}{\mbox{$x$}}
\newcommand{\by}{\mbox{\boldmath $y$}}
\newcommand{\bz}{\mbox{\boldmath $z$}}
\newcommand{\bu}{\mbox{\boldmath $u$}}
\newcommand{\bv}{\mbox{\boldmath $v$}}
\newcommand{\bw}{\mbox{\boldmath $w$}}
\newcommand{\bvarphi}{\mbox{\boldmath $\varphi$}}
\newcommand{\balpha}{\mbox{\boldmath $\alpha$}}

\newcommand\myqed{{}}

\newcommand{\IBFSFS}{{IBFS$^{\mbox{~\!\tiny FS}}$}}
\newcommand{\IBFSAL}{{IBFS$^{\mbox{~\!\tiny AL}}$}}
\newcommand{\BKFS}{{BK$^{\mbox{~\!\tiny FS}}$}}
\newcommand{\BKAL}{{BK$^{\mbox{~\!\tiny AL}}$}}

\newcommand{\XL}{{X_{\le L}}}

\def\br#1{{\llbracket #1 \rrbracket}}


\title{\Large\bf  }
\author{}
\title{\Large\bf  \vspace{0pt} OrderedCuts: A new approach for computing Gomory-Hu tree}
\author{Vladimir Kolmogorov \\ \normalsize Institute of Science and Technology Austria (ISTA) \\ {\normalsize\tt vnk@ist.ac.at}}
\date{}
\maketitle

\begin{abstract}
The Gomory-Hu tree, or a cut tree, is a classic data structure 
that stores minimum $s$-$t$ cuts of an undirected weighted graph for all pairs of nodes $(s,t)$.
We propose a new approach for computing the cut tree based on a reduction
to the problem that we call {\tt OrderedCuts}. 
Given a sequence of nodes $s,v_1,\ldots,v_\ell$,
its goal is to compute minimum $\{s,v_1,\ldots,v_{i-1}\}$-$v_i$ cuts for all $i\in[\ell]$.
We show that the cut tree  can be computed by $\tilde O(1)$ calls to  {\tt OrderedCuts}.
We also establish new results for {\tt OrderedCuts} that may be of independent interest.
First, we prove that all $\ell$ cuts can be stored compactly with $O(n)$ space in a data structure that we call an {\em {\tt OC} tree}.
Second, we prove results that allow divide-and-conquer algorithms for computing OC tree.
Finally, we describe a practical implementation based on {\tt OrderedCuts},
and compare it experimentally with two existing implementations of the classical Gomory-Hu tree
algorithm as well as with our implementations.
The results suggest that the  {\tt OrderedCuts}-based approach is the most robust:
on many family of problems it outperforms other algorithms by 1-2 orders of magnitude,
and is never slower by more than a small factor.
Our implementation is publicly available at \url{https://pub.ist.ac.at/~vnk/software.html}.
\end{abstract}

\section{Introduction}\label{sec:intro}

We study the problem of computing the {\em Gomory-Hu (GH) tree}~\cite{GomoryHu:61} (aka {\em cut tree}) in an undirected weighted graph $G=(V,E,w)$ with non-negative edge weights.
This is a weighted tree $\calT$ on nodes $V$ with the property that for any distinct $s,t\in V$,
a minimum $s$-$t$ cut $S\subseteq V$ in $\calT$ is also a minimum $s$-$t$ cut in $G$.
It has numerous applications in various domains (see e.g. Section 1.4 in~\cite{Abboud:STOC21}).

More than half a century ago Gomory and Hu~\cite{GomoryHu:61} showed that such tree exists and can be computed using $n-1$ maximum flow computations and graph contractions. 
(As usual, we denote $n=|V|$ and $m=|E|$).
This complexity has been improved upon only relatively recently in several breakthrough results:
first to $\tilde O(n^2)$ 
\cite{GH:subcubic} 
and then to $O(m^{1+o(1)})$~\cite{GH:linear,GHdeterminstic:FOCS25}.
(The methods in~\cite{GH:subcubic,GH:linear} are randomized Monte-Carlo,
while the method in~\cite{GHdeterminstic:FOCS25} is deterministic).

Unfortunately, the algorithms in~\cite{GH:subcubic,GH:linear,GHdeterminstic:FOCS25}
 do not appear very practical at present.
These algorithms quite are involved and rely on several complex subroutines.
The authors of~\cite{GH:linear} write that their reduction to $n^{o(1)}$ maxflow computations is
``{\em is highly technical and requires several heavy hammers and over 50 pages to describe}''.
The determinsitic approach in~\cite{GHdeterminstic:FOCS25} is even more extensive and introduces further sophisticated ingredients such as expander decompositions.
Note that several key steps in~\cite{GH:subcubic,GH:linear} are Monte-Carlo (correct only with high probability).
This means that certain parameters such as the number of iterations 
need to be set to their worst-case bounds,
as the success of an individual iteration is not directly observable.
These bounds often involve large constants that likely offset the asymptotic gains on practical datasets.

To our knowledge, the algorithms in~\cite{GH:subcubic,GH:linear,GHdeterminstic:FOCS25}
have not been implemented. All reported implementations that we aware of~\cite{Goldberg:01,MassiveGraphs,GH:parallel:11,GH:parallel:12,GH:parallel:17,GH:parallel:20} use 
either variants of the original algorithm by Gomory and Hu~\cite{GomoryHu:61}
or its simplification due to Gusfield~\cite{Gusfield:90}.

In this paper we propose a conceptually new approach for  constructing a GH tree.
It is based on a reduction to the computational problem that we call {\tt OrderedCuts}.
The latter is formulated (algorithmically) below;
 we use notation $f(U,v)$ to denote the cost of a minimum $U$-$v$ cut in $G=(V,E,w)$.
 To our knowledge, it has not been studied in the literature before.

\begin{algorithm}[H]
  \DontPrintSemicolon
\SetNoFillComment
  for each $i\in[\ell]$ compute $s$-$v_i$ cut $S_i$ with $v_i\in S_i$ s.t.\ ${\tt cost}(S_i)\le f(\{s,v_1,\ldots,v_{i-1}\},v_i)$\!\!\!\!\!\!\!\!\!\!\!\!\!\! \\
  for each $v\in V$ compute  $\lambda(v)=\min_{i\in [\ell]:v\in S_i} {\tt cost}(S_i)$,
  return $\lambda$ 
      {\em\small \hspace{250pt}$\slash\ast$~ 
      convention: minimum over the empty set is $+\infty$
      ~$\ast\slash$\!\!\!\!\!\!}
      
      \caption{${\tt OrderedCuts}(s,v_1,\ldots,v_\ell;G)$ {\em\small \hspace{20pt}$\slash\ast$~ $s,v_1,\ldots,v_\ell$ are distinct nodes in $G$ ~$\ast\slash$\!\!\!} 
      }
\end{algorithm}

Note that the definition is rather flexible, and allows several implementations.
One could for example choose $S_i$ as a minimum $\{s,v_1,\ldots,v_{i-1}\}$-$v_i$ cut for each $i\in[\ell]$;
we will refer to such version as {\em strong} ${\tt OrderedCuts}$ (and accordingly to the general version as {\em weak} {\tt OrderedCuts}).
Another option is to choose $S_i$ as a minimum $s$-$v_i$ cut; then the procedure would return vector~$\lambda$ with $\lambda(v_i)=f(s,v_i)$.
Below we discuss 
(i) how to use it to compute the GH tree, and
(ii) how to implement (strong) ${\tt OrderedCuts}$.

\myparagraph{Computing GH tree via {\tt OrderedCuts}}
Using techniques from~\cite{Abboud:FOCS20} and~\cite{LiPanigrahi:SICOMP24}, we establish the following result.
\begin{theorem}\label{th:main}
There exists a randomized (Las-Vegas) algorithm for computing GH tree for graph~$G$ with
expected complexity $\tilde O(1)\cdot (t_{\tt OC}(n,m) + t_{\tt MC}(n,m))$,
where $t_{\tt OC}(n,m)$ and $t_{\tt MC}(n,m)$ are the complexities of ${\tt OrderedCuts}(\cdot)$ and minimum $s$-$t$ cut computations respectively on a sequence
of graphs $H_1,\ldots,H_k$ that have $O(n)$ nodes and $O(m)$ edges in total.
\footnote{
When stating complexity bounds, we assume that $t_{\tt MC}(n,m)=\Omega(n+m)$ and $t_{\tt OC}(n,m)=\Omega(n+m)$.
}
\end{theorem}
The algorithm in Theorem~\ref{th:main} calls {\tt OrderedCuts} for subsets of nodes that are sampled according
to certain probabilities (and then ordered according to the current estimates of $f(s,v)$).
We found that setting these probabilities to 1 is very effective in practice:
in our experiments the number of calls to  {\tt OrderedCuts} was a very small constant.
Accordingly, in our implementation we use a deterministic algorithm
that does not have guarantees of Theorem~\ref{th:main}, but appears to work well in practice.

\myparagraph{Implementing {\tt OrderedCuts}} 
We prove that strong  {\tt OrderedCuts} have several interesting properties.
First, we show that all cuts $S_1,\ldots,S_\ell$ can be stored compactly with $O(n)$ space using a data structure
that we call an {\em {\tt OC} tree}. It is given by a rooted tree on nodes $U=\{s,v_1,\ldots,v_\ell\}$
satisfying a certain monotonicity property together with a partition of $V_G$ such that each partition contains
exactly one node in $U$.


Second, we present lemmas that allow the use of divide-and-conquer strategies
for computing {\tt OC} tree.
One of the strategies can be formulated as follows (ignoring base cases):
(1)~split input sequence as $(s,v_1,\ldots,v_\ell)=s\alpha\beta$ with $|\alpha|\approx |\beta|$;
(2)~solve recursively the problem for sequence $s\alpha$ in $G$;
(3)~compute minimum $s\alpha$-$\beta$ cut $(S,T)$ in $G$;
(4)~solve recursively the problem for sequence $s\beta$ in the graph obtained from $G$ by contracting $S$ to $s$;
(5)~combine the results. 
If, for example, the order of $(v_1,\ldots,v_\ell)$ is a random
permutation then we will have $|T|\le \tfrac 12 |V_G|$ with probability at least $\tfrac 12$
(since orderings $\alpha\beta$ and $\beta\alpha$ are equally likely),
which leads to an improved complexity over a naive approach.~\footnote{
The actual algorithm formulated in Section~\ref{sec:implementation}
is slightly different: 
instead of computing minimum $s\alpha$-$\beta$ cut, 
we split $G$ into components using the information obtained in step (2),
and then compute an appropriate minimum cut in each component. The guarantees for random permutations are still preserved.
}
\begin{theorem}\label{th:main2}
There exists an algorithm for solving strong {\tt OrderedCuts}
for sequence $(s,v_1,\ldots,v_\ell)$ that
has complexity of  maximum flow computations on $O(n)$ graphs of size $(n,m)$ each.
If the order of $v_1,\ldots,v_\ell$ is a uniformly random permutation
then the expected total number of nodes and edges in these graphs is $(O(n^{1+\gamma}),O(n^\gamma m))$
where $\gamma=\log_2 1.5=0.584...$.
\end{theorem}

Our practical implementation uses a different divide-and-conquer strategy for computing {\tt OC} tree.
On the high-level, the method has the same form as the computations performed by the classical
Gomory-Hu algorithm:
given a problem on the current graph $G$, run maxflow algorithm
to obtain a certain cut $(S,T)$ and then recurse on the problems defined by $S$ and by $T$
(contracting the other set to a single vertex). 
There is, however, an important distinction between the two:
cuts computed in the GH approach are minimum $s$-$t$ cuts in the original graph
for some terminals $s,t$, whereas cuts in the OC approach
are minimum $S$-$T$ cuts in the original graph
where the sizes of $S,T$ may grow at deeper levels of the recursion, and furthermore
the size of $T$ can be controlled to a certain degree. 
Thus, one may expect the OC approach to give more balanced cuts in practice.
Note, if all cuts $(S,T)$ were balanced
(meaning $\min\{|S|,|T|\}=\Omega(|V_G|)$) then the maximum depth of the recursion would be logarithmic in $n$
and the approaches would take $\tilde O(1)$ calls to the maxflow algorithm on graphs of size $(O(n),O(m))$.

Our experiments confirm this intuition.
On all of our test instances (except for one family of graphs, namely cycle graphs)
the total size of graphs on which the maxflow algorithm is run was significantly smaller
in the OC approach compared to the classical GH approach. For the majority of test instances this translated
to smaller runtimes (and outperformed previous implementations of Goldberg and Tsioutsiouliklis~\cite{Goldberg:01} and of Akiba et al.~\cite{MassiveGraphs}).

\myparagraph{Related work}
GH tree construction algorithms have been extensively studied in the literature.
A significant progress has been made for unweighted simple graphs~\cite{Bhalgat:07,Abboud:STOC21,Abboud:FOCS21,Li:FOCS21,Abboud:SODA22}
and for the approximate version of the problem~\cite{Abboud:FOCS20,LiPanigrahi:STOC21,FairCuts}.

Several authors gave reductions from some other problems.
For example,
Abboud, Krauthgamer and Trabelsi showed in~\cite{Abboud:FOCS20} that the GH tree can
be constructed via $\tilde O(1)$ calls to an oracle for the ``Single Source Min Cut problem'',
whose goal is to compute minimum $s$-$v$ cuts for a fixed source node $s$ and all
other nodes $v$.
Other results in this direction are described in Appendix~\ref{sec:related}.

\myparagraph{Connection to the Hao-Orlin algorithm}
The strong version of {\tt OrderedCuts} is reminiscent of the Hao-Orlin algorithm~\cite{HaoOrlin:94} for computing a minimum cut in a graph.
The latter also computes minimum $\{s,v_1,\ldots,v_{i-1}\}$-$v_i$ cut for each $i\in[\ell]$,
for some sequence of nodes $s,v_1,\ldots,v_\ell$. It has the same complexity $O(nm\log (n^2/m))$ as a single call to a push-relabel maxflow algorithm.
The crucial difference is that the Hao-Orlin algorithm selects the order of nodes $v_1,\ldots,v_\ell$ itself;
it cannot be used for a given ordering of nodes. Note that the Hao-Orlin algorithm works for directed graphs,
whereas our results for {\tt OrderedCuts} are restricted to undirected weighted graphs.

\vspace{5pt}

The rest of the paper is organized as follows. Section~\ref{sec:background} gives background on GH tree construction
algorithms. 
Section~\ref{sec:implementation} studied properties of strong {\tt OrderedCuts}
and discusses its computation.
Section~\ref{sec:impl} describes details of our implementations,
and Section~\ref{sec:experiments} presents experimental results.
Due to space limitations some parts are moved to appendices.
Section~\ref{sec:related} elaborates some of related work on the theoretical side.
Section~\ref{sec:depth1} shows how to construct a GH tree using $\tilde O(1)$
calls to an oracle for {\em depth-1 {\tt OC} tree}.
(It is a weaker version of the {\tt OC} tree obtained from the latter by repeatedly removing leaves at depth two or larger).
Section~\ref{sec:alg} shows how to construct a GH tree using $\tilde O(1)$ oracle calls
for weak {\tt OrderedCuts}.
Section~\ref{sec:implementation:proofs} contains proofs omitted in Sections~\ref{sec:implementation} and~\ref{sec:impl}.

\section{Background and notation}\label{sec:background}
Consider an undirected weighted graph $G=(V_G,E_G,w_G)$.
A {\em cut} of $G$ is a set $U$ with $\varnothing\subsetneq U\subsetneq V_G$.
It is an {\em $S$-$T$ cut} for disjoint subsets $S,T$ of $V_G$ if $T\subseteq U\subseteq V_G-S$.
The cost of $U$ is defined as ${\tt cost}_G(U)=\sum_{uv\in E_G:|\{u,v\}\cap U|=1}w_G(uv)$.
The cost of a minimum $S$-$T$ cut in $G$ is denoted as $f_G(S,T)$. 
If one of the sets $S,T$ is singleton, e.g. $T=\{t\}$,
then we say ``$S$-$t$ cut'' and write $f_G(S,t)$ for brevity.

When graph $G$ is clear from the context
we omit subscript $G$ and write $V$, $E$, ${\tt cost}(U)$, $f(S,T)$.

\myparagraph{Gomory-Hu algorithm}
Below we review the classical Gomory-Hu algorithm~\cite{GomoryHu:61} for computing the cut tree.
It works with 
a {\em partition tree} for $G$ which is a tree $\calT=(V_\calT,E_\calT)$
such that $V_\calT$ is a partition of $V$. An element $X\in V_\calT$ is called a {\em supernode of $\calT$}.
Each edge $XY\in E_\calT$ defines a cut in $G$ in a natural way; it will be denoted as $C_{XY}$ where we assume that $Y\subseteq C_{XY}$.
We view $\calT$ as a {\bf weighted} tree where the weight of $XY$ (equivalently, $f_\calT(X,Y)$) equals ${\tt cost}_G(C_{XY})$.
Tree $\calT$ is {\em complete} if all supernodes are singleton subsets of the form $\{v\}$;
such $\calT$ can be identified with a tree on nodes $V$ in a natural way.

We use the following notation for a graph $G$, partition tree $\calT$ on $V$ and supernode $X\in V_\calT$:
\begin{itemize}
\item
 $H=G[\calT,X]$ is the auxiliary graph obtained from $G$ as follows:
 (i) initialize $H:=G$, let $\calF$ be the forest  on $V_\calT-\{X\}$ obtained from tree $\calT$ by removing node $X$;
 (ii) for each edge $XY\in E_\calT$ find the connected component $\calC_Y$ of $\calF$ containing $Y$,
 and then modify $H$ by contracting nodes in $\bigcup_{C\in\calC_Y}C$ to a single node called $v_Y$.
 Note that $V_H=X\cup\{v_Y\::\:XY\in E_\calT\}$.
\end{itemize}

\begin{algorithm}[H]
\setcounter{AlgoLine}{0}
  \DontPrintSemicolon
  set $\calT=(\{V\},\varnothing)$ \\
  \While{exists $X\in V_\calT$ with $|X|\ge 2$}
  {
 pick supernode $X\in V_\calT$ with $|X|\ge 2$ and distinct $s,t\in X$
 \\
 form auxiliary graph $H=G[\calT,X]$
 \\
 compute minimum $s$-$t$ cut $S$ in $H$ \\ 
 let $(A,B)=(X-S,X\cap S)$, 
	update $V_\calT:=(V_\calT-\{X\})\cup\{A,B\}$  and  $E_\calT:=E_\calT\cup\{AB\}$\!\!\!\!\!\!\!\!\!\!\!\!\!\!\!\!\!
	\\ for each  $XY\in E_\calT$ update $E_\calT:=(E_\calT-\{XY\})\cup\{CY\}$
	where $C=\begin{cases}A & \mbox{if }v_{Y}\notin S \\ B & \mbox{if }v_{Y}\in S \end{cases}$\!\!\!\!\!\!\!\!\!\!\!\!\!\!\!\!\!\!\!
}
      \caption{Gomory-Hu algorithm for graph $G$.
      }\label{alg:GH}
\end{algorithm}

Note that at every step Algorithm~\ref{alg:GH} splits some supernode $X$ into two smaller supernodes, $A$~and~$B$.
We will use a generalized version proposed by Abboud et al.~\cite{Abboud:FOCS20} in which one iteration may split $X$ into more than two supernodes.

\begin{algorithm}[H]
\setcounter{AlgoLine}{0}
  \DontPrintSemicolon
  set $\calT=(\{V\},\varnothing)$ \\
  \While{exists $X\in V_\calT$ with $|X|\ge 2$}
  {
 pick supernode $X\in V_\calT$ with $|X|\ge 2$, 
 form auxiliary graph $H=G[\calT,X]$
 \\
 find node $s\in X$ and non-empty laminar family $\Pi$ of subsets of $V_H-\{s\}$ such that each~$S\in\Pi$ is a minimum $s$-$t$ cut in $H$ for some $t\in X$ 
 \\
 \While{$\Pi\ne\varnothing$}
 {
 	pick minimal set $S\in\Pi$ \\
	update $\calT$ as in lines 6-7 of Alg.~\ref{alg:GH} for given $S$ (with $(A,B)=(X-S,X\cap S)$) \\
	update $X:=A$, update $H$ accordingly to restore $H=G[\calT,X]$~~ \tcp{\footnotesize \em\!\!\!now $v_B\in V_H\!\!\!\!\!\!\!~$}
	remove $S$ from $\Pi$, for each $S'\in\Pi$ with $S\subseteq S'$ replace $S'$ with $(S'-S)\cup v_B$ in $\Pi$
 }
}
      \caption{Generalized Gomory-Hu algorithm for graph $G$.
      }\label{alg:GH'}
\end{algorithm}
Note that graph $H$ at line 8 is updated by contracting set $S$ to a single vertex named $v_B$.
Clearly, such transformation and  the update at line 9
preserve the property of set $\Pi$: it is still a laminar family of subsets of $V_H-\{s\}$
such that each~$S\in\Pi$ is a minimum $s$-$t$ cut in $H$ for some $t\in X$.
This means Algorithm~\ref{alg:GH'} updates tree $\calT$ in the same way as Algorithm~\ref{alg:GH}
(with certain choices of triplets $(X,s,t)$). This fact implies the correctness of Algorithm~\ref{alg:GH'}.

To simplify the analysis, in the theoretical part we will use family $\Pi$ in which all sets are pairwise-disjoint;
clearly, such family is laminar. Our practical implementation in Section~\ref{sec:impl}
uses general laminar families.

We will need another result that appeared in \cite{Goldberg:01} in the context of the Hao-Orlin algorithm.
Using our notation, it can be stated as follows.
\begin{lemma}[{\cite{Goldberg:01}}]\label{lemma:Goldberg}
Consider distinct nodes $v_0,v_1,\ldots,v_k$ in an undirected graph $G$ with $k\ge 1$.
If $f(\{v_0,\ldots,v_{k-1}\},v_k)\le \min_{i\in[k]} f(\{v_0,\ldots,v_{i-1}\},v_i)$
then $f(\{v_0,\ldots,v_{k-1}\},v_k)=f(v_0,v_k)$.
\end{lemma}
\begin{proof}
Denote $\lambda_i=f(\{v_0,\ldots,v_{i-1}\},v_i)$.
Suppose the claim is false, then there exists $v_0$-$v_k$ cut $S$ with $v_k\in S$ and ${\tt cost}(S)<\lambda_k$.
Let $i$ be the largest index in $[k]$ such that $\{s,v_1,\ldots,v_{i-1}\}\subseteq V-S$, then $v_i\in S$
and so ${\tt cost}(S)\ge \lambda_i$, contradicting condition $\lambda_k \le \min_{i\in[k]} \lambda_i$.
\end{proof}

\myparagraph{GH tree via {\tt OrderedCuts}}
We can now give a high-level idea of how we use the {\tt OrderedCuts} subroutine to implement line 4 of Algorithm~\ref{alg:GH'}
for a given source node $s$. For simplicity, let us assume that ${\tt OrderedCuts}(s,v_1,\ldots,v_\ell;H)$ explicitly
outputs $s$-$v_i$ cuts $S_i$ with ${\tt cost}(S_i)\le f(\{s,v_1,\ldots,v_{i-1}\},v_i)$.
We maintain upper bounds
$\lambda(v)$ on minimum cut costs $f(s,v)$ for nodes $v\in X-\{s\}$. In each iteration
we choose a subset $Y\subseteq X-\{s\}$, sort nodes in $Y$ in the non-increasing order of $\lambda(\cdot)$,
and call ${\tt OrderedCuts}$ for this sequence of nodes to obtain cuts $S_1,\ldots,S_\ell$.
Lemma~\ref{lemma:Goldberg} allows to certify that some of these cuts are ``true'' minimum $s$-$v_i$ cuts.
We add such cuts to $\Pi$,  update bounds $\lambda(\cdot)$ based on cuts $S_1,\ldots,S_\ell$, and repeat. For further details we refer to later sections;
they differ for practical and theoretical implementations, and also depend on the version of {\tt OrderedCuts} being used.

Throughout the paper we will use notation $C_{uv}$ to denote the minimal minimum $u$-$v$ cut in a given graph.
(The graph will always be clear from the context). Note that $v\in C_{uv}$. Set $C_{uv}$ is also known as the {\em latest $u$-$v$ cut}~\cite{Gabow:FOCS91}.
It is well-known that the family $\{C_{sv}\::\:v\ne s\}$ for a fixed node $s$ is laminar (this can be easily deduced from the submodularity inequality for cuts).

We  denote $\langle\Pi\rangle=\bigcup_{S\in \Pi} S$
 for a set of subsets~$\Pi$.


\section{Properties and computation of strong {\tt OrderedCuts}}\label{sec:implementation}

%

Consider undirected weighted graph $G=(V,E,w)$.
We will use letters $\varphi,\alpha,\beta,\ldots$ to denote sequences of distinct nodes in $V$.
For a sequence $\varphi$ and a set $X\subseteq V$ we denote $\varphi\cap X$
to be the subsequence of $\varphi$ containing only nodes in~$X$. 
With some abuse of notation we will often view sequence $\varphi=v_0 \ldots v_\ell$ as a set $\{v_0,\ldots,v_\ell\}$;
such places should always be clear from the context. 
Sequence $\varphi$ imposes a total order on nodes in $\varphi$;
we will denote this order as $\sqsubseteq$ (so that $v_0\sqsubset v_1\sqsubset\ldots\sqsubset v_\ell$), when $\varphi$ is clear from the context.

\begin{definition}\label{def:representing-tree}
Consider  sequence $\varphi=s\ldots$.
An {\em {\tt OC} tree for $\varphi$} is a pair $(\X,\calE)$ where
\begin{itemize}
\item $(\varphi,\calE)$ is a rooted tree with the root $s$ and edges oriented towards the root such that for any $uv\in\calE$ we have $v\sqsubset u$.
We write $u\preceq v$ for nodes $u,v\in\varphi$  if this tree has a directed path from $u$ to $v$.
\item $\X$ is a partition of some set $V$ with $|\X|=|\varphi|$ such that $|A\cap \varphi|=1$ for all $A\in\X$.
For node $v\in V$ we denote $[v]$ to be the unique component in $\X$ containing $v$
(then $[u]\ne[v]$ for distinct $u,v\in \varphi$).
We also denote $\XX([v])=\bigcup_{u\in\varphi\::\:u\preceq v} [u]$ for $v\in\varphi$.
\end{itemize}
We say that $(\X,\calE)$ is an {\em {\tt OC} tree for $(\varphi,G)$} 
(or {\em is valid for $G$}, when $\varphi$ is clear)
if the following holds: if $\varphi=\alpha v\ldots$ with $|\alpha|\ge 1$ then $\XX([v])$ is a minimum $\alpha$-$v$ cut in $G$.
\end{definition}

\begin{figure}
  \begin{minipage}[c]{0.4\textwidth}
			\includegraphics[trim=180 180 200 210, clip, scale=0.42]{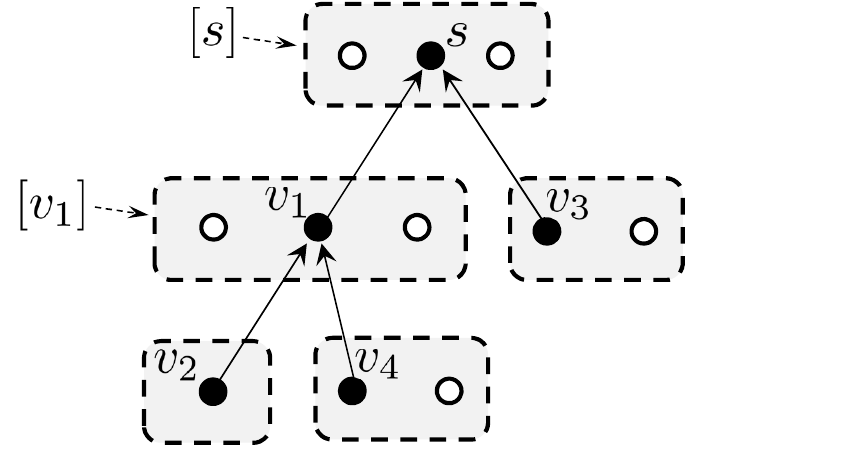} 
  \end{minipage}\hfill
  \begin{minipage}[c]{0.6\textwidth}
    \caption{Example {\tt OC} tree for sequence $\varphi=sv_1v_2v_3v_4$.
    Black and white circles are nodes in $\varphi$ and in $V-\varphi$, respectively.
    A valid {\tt OC} tree encodes all  cuts in  the definition of the strong {\tt OrderedCuts} problem;
    for example, set $[v_1]^\downarrow=[v_1]\cup [v_2]\cup [v_4]$ is a minimum $s$-$v_1$ cut,
    and set $[v_3]^\downarrow=[v_3]$ is a minimum $\{s,v_1,v_2\}$-$v_3$ cut.
    } \label{fig:03-03}
  \end{minipage}
\end{figure}

Note that operation $[\cdot]$ depends on $\X$.
Below we will need to work with multiple {\tt OC} trees. Unless noted otherwise, we will use the following convention:
different {\tt OC} trees will be denoted with an additional symbol, e.g.\ $(\X,\calE)$, $(\X',\calE')$, ...,
and the corresponding operations will be denoted accordingly as  $[\cdot]$, $[\cdot]'$,  ....
Operations $\preceq$ and $(\cdot)^\downarrow$ also depend on the {\tt OC} tree, but we will always use the same notation
for them; the corresponding {\tt OC} tree should always be clear from the context
(or these operations would be the same for all considered {\tt OC} trees).

Consider {\tt OC} tree  $(\Omega,\calE)$ for sequence $\varphi$ and a leaf node $u\in\varphi$.
Let $v$ be the parent of $u$ in $(\varphi,\calE)$ (i.e.\ $uv\in\calE$).
We define  $\varphi^{-u}$ to be the sequence obtained from $\varphi$ by removing node $u$,
and  $(\Omega^{-u},\calE^{-u})$ to be the {\tt OC} tree for $\varphi^{-u}$ obtained
 from $(\Omega,\calE)$ by removing edge $uv$ and merging components $[u]$ and $[v]$ of $\Omega$ into a single component $[v]^{-u}=[v]\cup[u]$.
 Equivalently, $(\Omega^{-u},\calE^{-u})$ can be defined as the unique {\tt OC} tree for $\varphi^{-u}$
 with the property that ${[w]^{-u}}^\downarrow = [w]^\downarrow$ for all $w\in\varphi^{-u}$.

It will be convenient to introduce the following notation for graph $G=(V,E,w)$,
subset $X\subseteq V$ and node $s\in X$: $G[X;s]$ is the graph obtained from $G$ by contracting set $(V-X)\cup\{s\}$ to $s$.
Note that the set of nodes of $G[X;s]$ equals $X$.
We will often use an $s$-$t$ cut in $G[X;s]$ where $s,t\in X$;
clearly, this is the same as an $s$-$t$ cut $U$ in $G$ satisfying $U\subseteq X$.

\begin{lemma}\label{lemma:OC-add-node}
Suppose that $(\X,\calE)$ is an {\tt OC} tree
for $\varphi u$ and $v\in\varphi$ is the parent of $u$ in $(\varphi u,\calE)$.
Suppose further that $(\X^{-u},\calE^{-u})$ is an {\tt OC} tree for $(\varphi,G)$.
Then $(\X,\calE)$ is valid for~$G$ if and only if set $[u]$ is a minimum $v$-$u$ cut in $G[[v]^{-u};v]$.
\end{lemma}

Lemma~\ref{lemma:OC-add-node} gives a constructive proof of the existence of {\tt OC} tree for a given $(\varphi,G)$.
\begin{corollary}
For every undirected graph $G$ and non-empty sequence $\varphi$ of nodes in $G$ there exists {\tt OC} tree $(\X,\calE)$ for $(\varphi,G)$.
\end{corollary}
\begin{proof}
It suffices to show that if there exists {\tt OC} tree $(\Omega,\calE)$ for  $(\varphi,G)$
and $u$ is a node in $V-\varphi$
then there also exists {\tt OC} tree $(\tilde\Omega,\tilde\calE)$ for  $(\varphi u,G)$; the claim will then follow by induction.
Let $v\in\varphi$ be the unique node in $\varphi$ with $u\in[v]$.
Let $T$ be a minimum $v$-$u$ cut in $G[[v];v]$. Define 
 $\tilde\X=(\X-[v])\cup\{[v]-T,T\}$ and $\tilde\calE=\calE\cup\{uv\}$.
 Note that $(\X,\calE)=({\tilde\X}^{-u},{\tilde\calE}^{-u})$.
 By Lemma~\ref{lemma:OC-add-node}, $(\tilde\X,\tilde\calE)$ is an {\tt OC} tree for $(\varphi u,G)$.
\end{proof}

\myparagraph{Extracting minimum $s$-$t$ cuts from an {\tt OC} tree}
Let $(\X,\calE)$ be an {\tt OC} tree for $(\varphi,G)$ with $\varphi=s\ldots u\ldots$,
and let $v$ be the parent of $u$ in $(\varphi,\calE)$.
Define $\pi(u)$ as the maximal node $w\in\varphi$ w.r.t. $\sqsubseteq$ 
satisfying the following conditions:
(i)  $w\sqsubset u$; 
(ii) $w\in\{v\}\cup\{w\::\:wv\in\calE\}$.
Clearly, by following pointers $\pi(\cdot)$ we will eventually arrive at the root node $s$ (visiting $v$ on the way).
Let $\pi^\ast(u)=s\ldots v\ldots =\ldots\pi(u)$ be the reversed sequence of nodes traversed during this process. 
We will prove the following result.

\begin{lemma}\label{lemma:pistar}
{\rm (a)}
Set $[u]^\downarrow$ is a minimum $\pi^\ast(u)$-$u$ cut. \\
{\rm (b)} If  ${\tt cost}([w]^\downarrow)\ge {\tt cost}([u]^\downarrow)$ for all $w\in \pi^\ast(u)-\{s\}$
then $[u]^\downarrow$ is a minimum $s$-$u$ cut.
\end{lemma}


\subsection{Divide-and-conquer algorithms} Next, we give two results
that can be used to design a divide-and-conquer algorithm for constructing {\tt OC} tree
(see Figs.~\ref{fig:A} and~\ref{fig:B} for their illustration).

\begin{figure}
    \hspace{25pt}
    \includegraphics[scale=0.60]{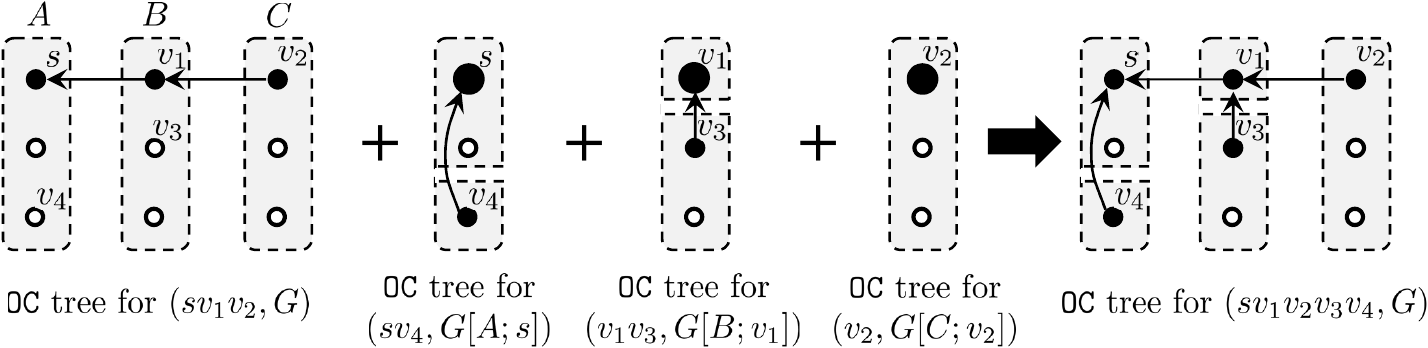}
    \caption{Illustration of Lemma~\ref{lemma:divide-and-conquer:one} with $\alpha=sv_1v_2$ and $\varphi=sv_1v_2 v_3v_4$. Thick solid circles in three graphs in the middle  represent contracted nodes, e.g.\ $v_1$ in $G[B;v_1]$ corresponds to subset $\{v_1\}\cup A\cup C$ in~$G$.
    Here we denoted $A=[s]^\circ$, $B=[v_1]^\circ$, $C=[v_2]^\circ$.
    } \label{fig:A}
\end{figure}

\begin{lemma}\label{lemma:divide-and-conquer:one}
Consider sequence $\varphi=\alpha\ldots$ in graph $G$ with $|\alpha|\ge 1$.
Suppose that $(\X^\circ,\calE^\circ)$ is an ${\tt OC}$ tree for $(\alpha,G)$,
and for each $v\in\alpha$ pair 
$(\X^v,\calE^v)$ is an {\tt OC} tree for $(\varphi\cap [v]^\circ,G[[v]^\circ;v])$.
Then $(\X,\calE)$ is an {\tt OC} tree for $(\varphi,G)$
where 
$\X=\bigcup_{v\in\alpha}\Omega^v$
and $\calE=\calE^\circ\cup\bigcup_{v\in\alpha}\calE^v$.
\end{lemma}

\begin{figure}
    \hspace{20pt}
    \includegraphics[scale=0.60]{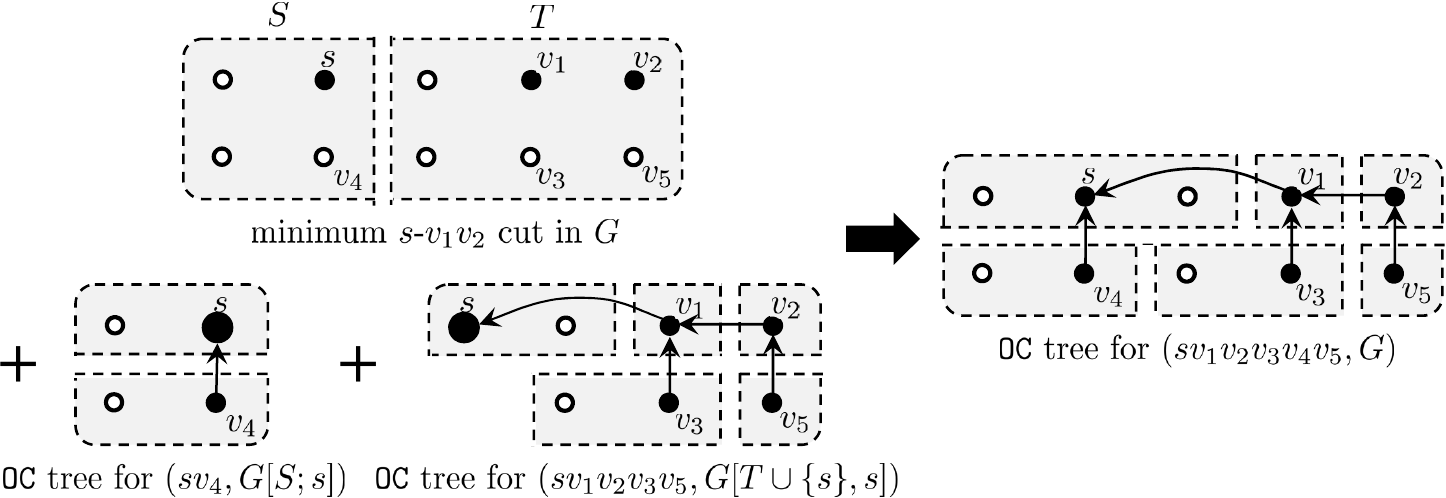}
    \caption{Illustration of Lemma~\ref{lemma:divide-and-conquer:two} with $\alpha=v_1v_2$ and $\varphi=sv_1v_2 v_3v_4v_5$. 
        } \label{fig:B}
\end{figure}

\begin{lemma}\label{lemma:divide-and-conquer:two}
Consider sequence $\varphi=s\alpha\ldots$ in graph $G$ with $|\alpha|\ge 1$.
Let $(S,T)$ be a minimum $s$-$\alpha$ cut in $G$, and let $T_s=T\cup\{s\}$.
Suppose that $(\Omega',\calE')$ is an {\tt OC} tree for $(\varphi\cap S,G[S;s])$
and $(\Omega'',\calE'')$ is an {\tt OC} tree for $(\varphi\cap T_s,G[T_s;s])$.
Then $(\X,\calE)$ is an {\tt OC} tree for $(\varphi,G)$
where $\Omega=\{[s]'\cup[s]''\}\cup (\Omega'-\{[s]'\}) \cup (\Omega''-\{[s]''\})$
and $\calE=\calE'\cup\calE''$.

Note, if $\varphi\cap S=s$ then $(\Omega',\calE')=(\{S\},\varnothing)$
and hence
 $\Omega=\{S\cup[s]''\}\cup (\Omega''-\{[s]''\})$
and $\calE=\calE''$.
\end{lemma}

Using these results, in Section~\ref{sec:permutations} we present a divide-and-conquer algorithm
with the following properties.

\begin{theorem}\label{th:random-permutation}
Procedure ${\tt OrderedCuts}(s,v_1,\ldots,v_\ell,G)$
can be implemented with the complexity of  maximum flow computations on $O(\ell)$ graphs of size $(n,m)$ each.
If the order of $v_1,\ldots,v_\ell$ is a uniformly random permutation
then the expected total number of nodes and edges in these graphs is $(O(n^{1+\gamma}),O(n^\gamma m))$
where $\gamma=\log_2 1.5=0.584...$.
\end{theorem}

For our experiments we used a different divide-and-conquer strategy given in Algorithm~\ref{alg:practical:OC}.
Essentially, it directly translates Lemma~\ref{lemma:divide-and-conquer:two} into an algorithm.
(Note that in Algorithm~\ref{alg:practical:OC} we only used a special case given in the last line of Lemma~\ref{lemma:divide-and-conquer:two}).

\begin{algorithm}[H]
\setcounter{AlgoLine}{0}

  \DontPrintSemicolon
    if $\ell=0$ then return {\tt OC} tree $(\{V\},\varnothing)$ \\
     choose $k\in[1,\ell]$,    compute minimum $s$-$v_1\ldots v_k$ cut $(S,T)$ in $G$ \\
    call $(\Omega',\calE')\leftarrow {\tt OrderedCuts}(\varphi\cap S,G[S;s])$ \\
    \If{$k=1$}
    {
	    call $(\Omega'',\calE'')\leftarrow {\tt OrderedCuts}(\varphi\cap T,G[T;v_1])$ \\
	    return $(\;\;\Omega' \cup \Omega''\;\;,\;\;\calE'\cup\calE''\cup\{sv_1\}\;\;)$
    }
    \Else
    {
	    call $(\Omega'',\calE'')\leftarrow {\tt OrderedCuts}(\varphi\cap T_s,G[T_s;s])$ where $T_s=T\cup \{s\}$ \\
	    return $(\;\;\{[s]'\cup[s]''\}\cup (\Omega'-\{[s]'\}) \cup (\Omega''-\{[s]''\})\;\;,\;\;\calE'\cup\calE''\;\;)$
    }
      \caption{${\tt OrderedCuts}(\varphi;G)$ with $\varphi=sv_1\ldots v_\ell$ and $G=(V,E,w)$. 
      }\label{alg:practical:OC}
\end{algorithm}

Next, we discuss how to choose index $k$ at line 2.
Ideally, we would like to choose $k$ so that the cut $(S,T)$ is balanced, i.e.\ $|S|\approx |T|$.
Note that $(S,T)$ depends monotonically on $k$: as $k$ grows, set $S$ shrinks and set $T$ expands
(see e.g.~\cite{Gallo:SICOMP89}).
Thus, one possibility would be to do a binary search on $k$ to find a value that maximizes $\min\{|S|,|T|\}$.
This would mean throwing away intermediate cuts, and keeping only the last one.
We used an alternative technique where we use every cut (even if it is unbalanced), but adjust the value of $k$ for recursive calls.
Specifically, we pass integer parameter $\bar k$ to every call of ${\tt OrderedCuts}(\cdot)$.
For the initial call we set $\bar k=1$.
At line 2 we set $k=\max\{1,\min\{\lfloor \ell/2\rfloor, \bar k\}\}$.
Then after computing cut $(S,T)$ we update $\bar k={\tt round}(k2^{1-2|T|/(|V|-1)})$,
and pass this value to recursive calls at lines 3,5,8.
Thus, $\bar k\in [\lfloor k/2\rfloor,2k]$, and if $S=\{s\}$ then $\bar k=\lfloor k/2\rfloor$.
We refer to Section-\ref{sec:experiments} for an experimental evaluation of this scheme.


\begin{remark}
We have also implemented a different procedure: (i) compute $s$-$\{v_1,\ldots,v_i\}$ for all $i\in[\ell]$;
the result is a nested family of cuts $T_1\subset T_2 \subset \ldots \subset T_k$ where 
$T_1$ is a minimum $s$-$v_1$ cut and $T_k$ is a minimum $s$-$\{v_1,\ldots,v_\ell\}$ cut;
(ii) solve recursively problems on subsets $T_1-\{v_1\}$, $T_2-T_1$, $\ldots$, $T_k-T_{k-1}$
(contracting other nodes to a single vertex); (iii) merge the results.
Note that the problem (i) is an instance of the parametric maxflow problem (cf.~\cite{Gallo:SICOMP89});
to solve it, we used a natural divide-and-conquer strategy with a binary search 
(except that the initial index $i$ was taken as $\min\{2^d,\lceil\ell /2\rceil\}$ where $d$ is depth of the recursion).

The procedure described earlier performed better than this parametric scheme, so we settled for the former.
\end{remark}


\section{Summary of implementations}\label{sec:impl}
In this section we describe our implementations of an {\tt OrderedCuts}-based approach
(``OC algorithm'') and of the classical Gomory-Hu tree technique
(''GH algorithm''). When reporting results, we denote ${\tt size}(G)=(n,m)$ to be the size of the original graph
and ${\tt size}(MF)$ to be the total size of graphs on which the maxflow problem is solved (excluding terminals and their incident edges).
For the OC approach we let ${\tt size}(OC)$ be the total size of graphs on which the {\tt OrderedCuts} procedure is called.
We will always have ${\tt size}(OC)\le {\tt size}(MF)$.

\subsection{OC algorithm}
Our implementation follows Algorithm~\ref{alg:GH'}. Its main computational subroutine is line 4.
Details of this subroutine are described below.


As stated in Theorem~\ref{th:main}, the GH tree can be constructed via $\tilde O(1)$ expected number of calls
to ${\tt OrderedCuts}$, as we prove in Section~\ref{sec:alg}.
Our proof  makes an essential use
of randomization: we select the source node $s$ uniformly at random, and then randomly subsample a subset of nodes to be passed to {\tt OrderedCuts}
(with varying probabilities). 
This achieves ${\tt size}(OC)=O(n\log^4 n,m\log^4 n)$.

We implemented a simpler (deterministic) scheme that does not use subsampling (see Algorithm~\ref{alg:practical:SSPQ}).
We take {\em all} nodes in $X$, sort them according to the current estimates on the min-cut costs, and 
compute {\tt OC} tree for the obtained sequence. 
We can only show a trivial bound ${\tt size}(OC)=(O(n^2), O(nm))$ for such scheme
(see Appendix~\ref{sec:practical:proof:complexity}). Nonetheless, experimental results indicate that this approach is very
effective: on all of our test instances ${\tt size}(OC)$ was at most $(cn,cm)$ for a small constant $c<10$  (in fact, $c<4$ for most of the families).
Based on these results, we did not investigate the approach based on subsampling (which we conjecture would be slower in practice).

\begin{algorithm}[H]
\setcounter{AlgoLine}{0}

  \DontPrintSemicolon
    choose node $s\in X$ and vector $\mu:X-\{s\}\rightarrow\mathbb R$ so that $\mu(v)\ge f(s,v)$ for each $v\in X-\{s\}$ \\
  		sort nodes in $X-\{s\}$ as $v_1,\ldots,v_\ell$ so that $\mu(v_1)\ge \ldots \ge \mu(v_\ell)$, $\ell=|X-\{s\}|$ \\
		compute {\tt OC} tree $(\Omega,\calE)$ for sequence $s v_1 \ldots v_\ell$ and graph $H$ using Alg.~\ref{alg:practical:OC} \\
		use Lemma~\ref{lemma:pistar} to find subset $U\subseteq X-\{s\}$ s.t.\ $[v]^\downarrow$ is a minimum $s$-$v$ cut for each $v\in U$ \\
		let $\Pi=\{[v]^\downarrow\::\:v\in U\}$,   return $(s,\Pi)$

      \caption{Line 4 of Alg.~\ref{alg:GH'}. 
      }\label{alg:practical:SSPQ}
\end{algorithm}
We need to specify line 1.
For the initial supernode $X=V$ we let source $s$ to be a node that maximizes ${\tt cost}(\{s\})$,
and set $\mu(v)={\tt cost}(\{v\})$ for all $v\in X-\{s\}$.
Now consider lines 6-9 of Alg.~\ref{alg:GH'} for set $S=[v]^\downarrow$.
Recall that it splits $X$ into supernodes $A$ and $B$ and then replaces $X:= A$.
We let $v$ to be the source for supernode $B$,
and set upper bounds $\mu$ for $B$ based on the cuts stored in the {\tt OC} tree $(\Omega,\calE)$:
$\mu(u)=\min_{w:u\preceq w\prec v}{\tt cost}([w]^\downarrow)$.
When $\Pi$ becomes empty, we update upper bounds $\mu$ for nodes $u\in X-\{s\}$ in a similar way
(except that we also take into account current upper bounds):
$\mu(u):=\min\{\mu(u),\min_{w:u\preceq w\prec s}{\tt cost}([w]^\downarrow)\}$.
Note that the source node $s$ is not changed for $X$.

It may easily happen that some nodes $v\in X-\{s\}$ have identical upper bounds $\mu(v)$.
To break the ties, we use an additional rule when sorting nodes at line 2:
if $\mu(u)=\mu(v)$ then we choose the same order for $u,v$ that was used in the previous iteration.

\subsection{GH algorithm}
Versions of the classical Gomory-Hu tree method (Alg.~\ref{alg:GH}) have been implemented in~\cite{Goldberg:01} (for general graphs)
and in~\cite{MassiveGraphs} (for unweighted simple graphs). To make the comparison with the OC approach fairer,
we have additionally implemented our own version of the GH method
to make sure that OC and GH approaches use the same maxflow algorithms (see Section~\ref{sec:impl:MF}).
Details are described below.

The main question in implementing Algorithm~\ref{alg:GH} is 
how to select nodes $s,t$ at line 3.
There are several guiding principles proposed in the literature.
\cite{Goldberg:01} proposes two heuristics that try to find a balanced cut:
(i) choose two heaviest nodes (i.e.\ nodes with the highest total capacity of incident edges);
(ii) after choosing $t$, choose $s$ that is furthest away from $t$ (with respect to unit edge lengths).
\cite{MassiveGraphs} uses a somewhat opposite strategy:
pick $s,t$ which are adjacent in $G$, if exist (otherwise pick an arbitrary pair).
They argue that such choice leads to a smaller time for  computing maxflow.

We implemented the following heuristic. For each component $X$ we maintain (1) nodes $s,t\in X$ and (2)
integers $D_s(i),D_t(i)$ for all $i\in X$ which approximate the distance from $s$ and $t$, respectively.
In the beginning we compute $s,t$ as two heaviest nodes, and then compute distances $D_s(\cdot)$ and $D_s(\cdot)$
(with respect to the unit edge lengths)
exactly using breadth-first searches from $s$ and from $t$. Now suppose that $X$ is split into supernodes $A,B$ after computing a minimum $s$-$t$ cut.
The new terminals $s_A,t_A$ for component $A$ are computed as follows.
We set $s_A=s$. We then do breadth-first search in $A$ from (contracted) $t$, and 
and for each layer check whether it has nodes eligible to become new sink $t$.
If yes, then we choose a node with the the smallest $D_s(i)$.
If there are multiple such nodes then we choose the node with the largest degree.
Intuitively, this strategy attempts to find new sink $t$ which is close both to $s$ and to the old sink $t$.
A similar strategy is used for $B$. The integers $D_s(\cdot)$, $D_t(\cdot)$ are not modified during this procedure.

We also implemented an alternative scheme in which we always choose $s,t$ as two heaviest nodes in the current component.
However, the scheme above was usually substantially faster, probably because warm-starting the maxflow algorithm
was much more effective. Consequently, we report results only for the former scheme.

\subsection{Maxflow algorithm}\label{sec:impl:MF}
We reimplemented the IBFS maxflow algorithm~\cite{IBFS:11},
with dynamic updates as in~\cite{IBFS:15}. As in~\cite{IBFS:15},
we do not store arcs $(s,i)$, $(i,t)$ explicitly, but only the difference of their residual capacities.
All other arcs are stored using the forward star representation.
We augmented our implementation with the following operation:
\begin{itemize}
\item After computing a maximum $s$-$t$ flow and corresponding cut $(S,T)$,
remove arcs between $S$ and $T$. Optionally, add new node $\hat s$ (corresponding to contracted $S$)
together with contracted arcs to $T$, and
new node $\hat t$ (corresponding to contracted $T$)
together with contracted arcs from $S$.
This may require allocating new memory for contracted arcs, since arcs incident to each node
are stored in a contiguous array.
Note that new nodes $\hat s,\hat t$ are needed for the {\tt GH} algorithm, but not for the {\tt OC} algorithm \footnote{In the {\tt OC} algorithm (Alg.~\ref{alg:practical:OC}),
node $s$ is present in both graphs $G[S;s]$ and $G[T\cup \{s\},s]$ (see Fig.~\ref{fig:B}).
However, in the latter graph $s$ will always be connected to the source with an infinite capacity in subsequent maxflow computations.
Hence, $s$ can be immediately eliminated from  $G[T\cup \{s\},s]$, and so no new nodes need to be allocated.
}.
\end{itemize}
We maintain a list of ``components'' of the current graph; computing maxflow in a given component splits this component into two.
Each component stores a pointer to the a singly-link list of its nodes, as well as other information needed
for warm-starts
(e.g.\ maximum levels of source and sink search trees in the IBFS algorithm).
We order components in the decreasing order of their size (which is the number of nodes), and always process the smallest component.
It can be seen that the size of each component is at most half the size of the previous component, and
hence at each moment there are at most $O(\log n)$ components.

In both {\tt OC} and {\tt GH} algorithm computed cuts $(S,T)$ can be very unbalanced.
For efficiency, it is important not to traverse the larger component in such cases.
To achieve this, we maintain the sizes of components during IBFS operations.
After IBFS finishes, we find nodes of the smaller component using breadth-first search from
the appropriate terminal. During this operation we also detect ``boundary'' edges between $S$ and $T$ and remove them from the graph.

Note that in {\tt OC} algorithm we also need to maintain a doubly-linked list of ordered nodes $(v_1,\ldots,v_\ell)$.
We do it in $O(L\log L)$ time where $L$ is the number of nodes in $\{v_1,\ldots,v_\ell\}$ that end up in the smaller component.
(The order of these nodes in the larger component comes from the original order, while nodes in the smaller component are sorted from scratch).

\section{Experimental results}\label{sec:experiments}

We used a machine with 12th Gen Intel(R) Core(TM) i7-1255U CPU with 10 cores, 12 threads @4.7GHz max and 15Gb RAM.
All codes were written in C / C++. For each instance we reordered nodes and edges
by performing a breadth-first search (in order to improve locality). Next, we describe baseline methods,
test instances, and discuss the results.

\subsection{Baseline methods}
Our implementations are termed as {\tt OC} and {\tt GH}. Other implementations, described below,
are always denoted with lower-case letters.

The study~\cite{Goldberg:01} reported 4 algorithms named {\tt gus}, {\tt gh}, {\tt ghs} and {\tt ghg}.
The first one is an implementation of the Gusfield's algorithm~\cite{Gusfield:90} that works on the original graph,
while the other three perform graph contractions. {\tt gh} and {\tt ghs} differ in the source-sink
selection strategy: the former picks an $s$-$t$ pair at random while the latter chooses $s$
which is furthest away from the current $t$.  {\tt ghg} uses the Hao-Orlin algorithm instead of the standard preflow algorithm.
\cite{Goldberg:01} shows that a single Hao-Orlin run can identify several valid cuts (1 or 2 in their implementation).
In our tests  {\tt ghs} crashed with a segmentation fault on almost all TSP instances, and was never faster than  {\tt ghg} by more than a factor 2.
For these reasons we do not report {\tt ghs} results.

We also evaluated the algorithm from~\cite{MassiveGraphs} that we term {\tt mg}
(which stands for ``Massive Graphs'' from the title of the paper). Their code supports only unweighted simple graphs.
It uses a bidirectional Dinitz algorithm with a ``goal-oriented search'':
it precomputes a shortest path tree from a fixed sink $t$ and then
uses this tree for multiple sources $s$. It also uses heuristics such as augmenting non-shortest paths
with ``detour edges'' (see~\cite{MassiveGraphs} for more details).
According to~\cite{MassiveGraphs}, the code supports several preprocessing heuristics for graph reduction, namely tree packing for identifying singleton valid cuts,
removing bridge edges, and contracting degree-2 nodes. 
We ran the code with the flag {\tt cut\_tree\_enable\_greedy\_tree\_packing} set to {\tt true} ({\tt mg+}) and to {\tt false} ({\tt mg-}).

Finally, we experimented with the Gomory-Hu algorithm from the {\tt Lemon} graph library~\cite{Lemon},
but found it not to be competetive.

As discussed in the introduction, the theoretically-fast algorithms in~\cite{GH:subcubic,GH:linear,GHdeterminstic:FOCS25}
are unlikely to be competetive as described.
A direct comparison is currently infeasible: these algorithms lack publicly available implementations, 
and the papers do not provide pseudocode at a level of detail that would allow for a faithful reproduction of the results.

\subsection{Test instances}
\myparagraph{Weighted graphs}
First, we tested the algorithms on synthetic classes of graphs used in~\cite{Goldberg:01} and in earlier studies
(only of bigger sizes). We used the generators that come with the code in~\cite{Goldberg:01}.
In tables~\ref{table:BIKEWHE}-\ref{table:WHE} we specify the exact commands that were used to generate the instances.
The generators are randomized (except for CYC1 and DBLCYC1); as in~\cite{Goldberg:01},
we report the numbers averaged over 5 runs with different random seeds.

Second, we tested TSP instances from the TSPLIB dataset~\cite{TSPLIB}. Each instance
is given by a set of points and a function that allows to compute a distance
between two points. To obtain a sparse instance, we took the $kn$ edges with the smallest weight for $k=2,4,8$.

Results for weighted graphs are presented in tables \ref{table:TSP} (TSP instances) and~\ref{table:BIKEWHE}-\ref{table:WHE} (synthetic instances).
Table~\ref{table:ALL} summarizes the last entries in tables \ref{table:BIKEWHE}-\ref{table:WHE}.
(In general, these entries are the hardest instances in the corresponding family, except for NOI5 and NOI6).

\myparagraph{Unweighted simple graphs} Recall that the {\tt mg} implementation allows only unweighted
simple graphs. To generate such graphs, we first took the last instance in each of the tables~\ref{table:BIKEWHE}-\ref{table:WHE}
and changed the weight of all edges to one. Results are given in Table~\ref{table:unit:ALL}.
Table~\ref{table:unit:TSP} shows the results on TSP instances from Table~\ref{table:TSP}
converted to unit weights.

Finally, we took large social and web graphs from~\cite{snapnets} used in~\cite{MassiveGraphs}
(Table~\ref{table:unit:SOC}). Some of them are directed; we converted them to undirected.
We also removed bridge edges and took the largest remaining connected component
(motivated by the preprocessing technique in~\cite{MassiveGraphs}).

To summarize, we tested the algorithms on the following instances:
\begin{itemize}
\item Synthetic classes of graphs used in~\cite{Goldberg:01} and in earlier studies (only of bigger sizes):
BIKEWHE, CYC1, DBLCYC1, NO1, ..., NO6, PATH, PR1, PR5, ..., PR8, REG1, TREE, WHE.
\item TSP instances from the TSPLIB dataset~\cite{TSPLIB} with different sparsity ($m/n\in\{2,4,8\}$).
\item Unweighted versions of the graphs above (with all edge weights set to 1).
\item Large (unweighted) social and web graphs from~\cite{snapnets} used in~\cite{MassiveGraphs}
\end{itemize}

\subsection{Performance measures}
For each family of graphs we present two tables: left and right (see e.g.\ Table~\ref{table:ALL}).
The left table contains the following information (see Section~\ref{sec:impl} for
the description of ${\tt size}(G)$, ${\tt size}(OC)$, ${\tt size}(MF)$):
\begin{itemize}
\item column 2: ${\tt size}(G)=(n,m)$
\item column 3: GH tree diameter (or a range of diameters, if different algorithms produced different trees)
\item column 4: $\frac{{\tt size}(OC)}{{\tt size}(G)}$ for the {\tt OC} method. (The ratios are computed component-wise, i.e.\ separately for nodes and for edges).
This measures the effectiveness of our reduction technique to the {\tt OrderedCuts} problems.
\item column 5: $\frac{{\tt size}(MF)}{{\tt size}(OC)}$ for the {\tt OC} method. This can be viewed as a measure
of effectiveness of our implementation of the {\tt OrderedCuts} procedure.
\item columns 6-7:  $\frac{{\tt size}(MF)}{{\tt size}(G)}$ for {\tt OC} and {\tt GH} methods.
\end{itemize}
The right table contains runtimes in seconds. We report two numbers in the format $t_{\tt total} / t_{\tt MF}$
where $t_{\tt total}$ is the overall runtime and $t_{\tt MF}\le t_{\tt total}$ is the time spent in maxflow computations.
(If $t_{\tt MF}$ is not available then we report only $t_{\tt total}$).

The caption in each table (except for Table~\ref{table:ALL}) also gives the time of the {\tt Lemon}'s implementation for the first instance in the corresponding table.


\subsection{Discussion of results}
\myparagraph{Performance of the reduction to {\tt OrderedCuts}}
On all test instances the ratios $\frac{{\tt size}(OC)}{{\tt size}(G)}$ is  small constants
($<10$ on CYC1, $<6$ on PATH and $<4$ on all other instances).
This indicates that our reduction to {\tt OrderedCuts} is very effective, despite the lack the theoretical guarantees.

\myparagraph{Sizes of maxflow graphs}
Let us now discuss the ratio $\frac{{\tt size}(MF)}{{\tt size}(G)}$. For most of the problems
this ratio was dramatically smaller for the OC approach compared to the GH approach.
The latter had ${\tt size}(MF)\approx n\cdot {\tt size}(G)$ for many classes of problems,
indicating that most $s$-$t$ cuts computed during the algorithm were very unbalanced.
However, there was one single exception: for CYC1 problems (which are cycle graphs on $n$ nodes and $n$ edges)
${\tt size}(MF)$ was much larger for the OC approach.


\myparagraph{{\tt GH} vs. {\tt gus}, {\tt gh}, {\tt ghg}}
The maxflow time in GH is always smaller than that in the codes of~\cite{Goldberg:01}.
This could be attributed to several factors:
(1) we use the more recent IBFS maxflow algorithm;
(2) our warm-starts may be more efficient; in particular, we avoid traversing the entire graphs after unbalanced splits,
and reuse graphs stored in memory instead of allocating new ones (see Sec.~\ref{sec:impl:MF});
(3) we use the forward star graph representation instead of adjacent lists.
On top of it, on some families the codes of~\cite{Goldberg:01} have a very large overhead over maxflow computations (e.g.\ for TSP instances),
whereas the overhead in our codes is usually below 50\%.
As a result, our {\tt GH} implementation strictly dominates the codes of~\cite{Goldberg:01}.

\myparagraph{{\tt OC} vs. {\tt GH}} On almost all instances {\tt OC} is faster than {\tt GH}, and often significantly.
In particular, out of 18 families of weighted synthetic instances in Table~\ref{table:ALL}, {\tt OC} is faster on 16 families
(out of which by 10 times or more on 14 families). On weighted TSP instances in Table~\ref{table:TSP} {\tt OC} is usually faster than {\tt GH} by a factor of 1.5 - 2.
Unweighted instances (Tables \ref{table:unit:ALL} and \ref{table:unit:TSP}) exhibit a similar pattern.
{\tt OC} is slower than {\tt GH} only on two families: by a factor of 6 on CYC1 (cycles with random weights) and by a factor of 3 on DBLCYC (rows 2,3 in Table \ref{table:ALL}).
Note that on unweighted cycles {\tt OC} becomes faster than {\tt GH} (row 2 in Table \ref{table:unit:ALL}).

\myparagraph{{\tt OC} vs. {\tt mg} on unweighted simple graphs} This comparison is given Tables \ref{table:unit:ALL}, \ref{table:unit:TSP}, \ref{table:unit:SOC}.
In the first two tables (synthetic and TSP instances) {\tt OC} is slower than {\tt mg} in the first two rows of Table \ref{table:unit:ALL}
(by a factor of 2-5) and faster for all other instances, often by 1-2 ordered of magnitude.
On social and web graphs (Table \ref{table:unit:SOC}) {\tt OC} is slower than {\tt mg} on most instances, but by a small factor (usually less than 2).

\myparagraph{Summary} Overall, {\tt OC} is the most robust algorithm in the tests above:
on the majority of families it outperforms {\tt GH} and {\tt mg} (when the latter is applicable) by 1-2 orders of
magnitude, and is never slower by more than a factor of 6.
The biggest challendge for {\tt OC} appears to be the CYC1 family (weighted cycles).
We remark that it should be possible to design a customized algorithm for such families
(e.g. by contracting paths containing degree-2 nodes).

The results also suggests that {\tt GH} is currently the fastest implementation
of the classical Gomory-Hu tree algorithm on weighted graphs.

\def\ABIKEWHE{
\P{1024.0}{2045.0} & 2 & \P{1.0}{1.0} & \P{27.2754}{18.7042} & \P{27.2754}{18.7042} & \P{1023.0}{1023.0} \\ 
\P{2048.0}{4093.0} & 2 & \P{1.0}{1.0} & \P{36.3496}{24.0994} & \P{36.3496}{24.0994} & \P{2047.0}{2047.0} \\ 
\P{4196.0}{8389.0} & 2 & \P{1.0}{1.0} & \P{43.2896}{28.9809} & \P{43.2896}{28.9809} & \P{4195.0}{4195.0} \\ 
}
\def\BBIKEWHE{
\Q{0.009851794}{0.0087076} & \Q{0.05465714000000001}{0.05400594000000001} & \Q{2.3493736}{2.2835308000000003} & \Q{2.1473164000000002}{2.0611821999999997} & \Q{0.5630128}{0.4855838}  \\
\Q{0.023825740000000005}{0.022104359999999997} & \Q{0.2213604}{0.21997219999999995} & \Q{17.2121044}{16.9133654} & \Q{16.2657002}{15.860764399999999} & \Q{3.644588}{3.3037749999999995}  \\
\Q{0.098969}{0.09543172} & \Q{1.068994}{1.065664} & \Q{177.8264038}{176.33547059999998} & \Q{154.74469919999999}{152.45422979999998} & \Q{22.2316876}{20.7858218}  \\
}

\def\ACYCONE{
\P{4196.0}{4196.0} & 4195 & \P{6.98213}{6.98213} & \P{17.69101978909015}{16.60438863212229} & \P{123.521}{115.934} & \P{20.8799}{20.8799} \\ 
\P{8392.0}{8392.0} & 6875-7493 & \P{7.72081}{7.72081} & \P{18.99178972154476}{17.8779169543092} & \P{146.632}{138.032} & \P{25.7324}{25.7324} \\ 
\P{16784.0}{16784.0} & 15267-15885 & \P{7.39377}{7.39377} & \P{21.891403167802086}{20.7876360774003} & \P{161.86}{153.699} & \P{26.7837}{26.7837} \\ 
}
\def\BCYCONE{
\Q{0.0415024}{0.0229392} & \Q{0.0105462}{0.0065812} & \Q{2.168719}{0.711024} & \Q{1.138642}{0.005533} & \Q{0.85358}{0.005706}  \\
\Q{0.0884025}{0.050836} & \Q{0.0176617}{0.0111302} & \Q{9.796519}{3.57604} & \Q{4.513091}{0.017993} & \Q{3.635111}{0.020176}  \\
\Q{0.173987}{0.0994704} & \Q{0.0334107}{0.02139} & \Q{53.268906}{16.128618} & \Q{27.557215}{0.040244} & \Q{21.189104}{0.039382}  \\
}

\def\ADBLCYC{
\P{2048.0}{4096.0} & 5 & \P{2.00098}{2.00098} & \P{37.74540475167168}{34.4484702495777} & \P{75.5278}{68.9307} & \P{1021.99}{1021.82} \\ 
\P{4196.0}{8392.0} & 5 & \P{1.99976}{1.99964} & \P{66.64549745969516}{63.25638614950691} & \P{133.275}{126.49} & \P{2096.99}{2096.83} \\ 
\P{8392.0}{16784.0} & 5 & \P{2.00024}{2.00024} & \P{53.9960204775427}{50.61742590889094} & \P{108.005}{101.247} & \P{4194.51}{4194.35} \\ 
}
\def\BDBLCYC{
\Q{0.0158835}{0.0107259} & \Q{0.00869429}{0.00643989} & \Q{24.938303}{24.568422} & \Q{1.68083}{1.289032} & \Q{0.39465}{0.063931}  \\
\Q{0.0319196}{0.0215965} & \Q{0.0139821}{0.010779} & \Q{180.307007}{178.713165} & \Q{11.104429}{9.307926} & \Q{1.91377}{0.565612}  \\
\Q{0.0565502}{0.0390912} & \Q{0.0190084}{0.0138528} & \Q{1733.210693}{1725.780884} & \Q{60.329689}{51.719501} & \Q{6.63112}{1.432986}  \\
}

\def\ANOIONE{
\P{400.0}{39900.0} & 2 & \P{1.0}{1.0} & \P{14.39}{7.392026000000001} & \P{14.39}{7.392026000000001} & \P{399.0}{341.52919999999995} \\ 
\P{600.0}{89850.0} & 2 & \P{1.0}{1.0} & \P{15.6067}{8.025079999999999} & \P{15.6067}{8.025079999999999} & \P{599.0}{512.3467999999999} \\ 
\P{800.0}{159800.0} & 2 & \P{1.0}{1.0} & \P{16.4175}{8.462890000000002} & \P{16.4175}{8.462890000000002} & \P{799.0}{683.6982} \\ 
\P{1000.0}{249750.0} & 2 & \P{1.0}{1.0} & \P{17.106}{8.818312} & \P{17.106}{8.818312} & \P{999.0}{855.083} \\ 
}
\def\BNOIONE{
\Q{0.007644072}{0.003337562} & \Q{0.22913360000000002}{0.2251932} & \Q{0.3004898}{0.1159966} & \Q{0.44448119999999997}{0.1338848} & \Q{0.4983792}{0.2745202}  \\
\Q{0.01780304}{0.007003384} & \Q{0.7321875999999999}{0.7233642} & \Q{1.012618}{0.3969242} & \Q{1.5524706}{0.4601128} & \Q{1.7148062}{0.9176286000000001}  \\
\Q{0.034791819999999994}{0.012772779999999997} & \Q{1.6888020000000001}{1.671496} & \Q{3.453572}{1.3238532} & \Q{5.4443684}{1.5275945999999998} & \Q{5.8766532}{3.1077369999999997}  \\
\Q{0.05597426}{0.01929914} & \Q{3.231376}{3.203306} & \Q{7.481244200000001}{2.874071} & \Q{11.7368366}{3.3123036} & \Q{13.0124244}{7.0519663999999995}  \\
}

\def\ANOITWO{
\P{400.0}{39900.0} & 3 & \P{2.005}{1.505364} & \P{14.625935162094763}{6.719622629476991} & \P{29.325}{10.115478} & \P{207.38500000000005}{97.45808} \\ 
\P{600.0}{89850.0} & 3 & \P{2.00333}{1.5010819999999998} & \P{15.998362726061107}{7.2676762495320055} & \P{32.05}{10.909378} & \P{305.6142}{137.23579999999998} \\ 
\P{800.0}{159800.0} & 3 & \P{2.0025}{1.500526} & \P{17.01085642946317}{7.955450288765407} & \P{34.06424}{11.937360000000002} & \P{404.5502}{178.6096} \\ 
\P{1000.0}{249750.0} & 3 & \P{2.002}{1.502422} & \P{17.091808191808195}{7.412832080467405} & \P{34.217800000000004}{11.137201999999998} & \P{505.6704}{223.12239999999997} \\ 
}
\def\BNOITWO{
\Q{0.01175556}{0.005926672} & \Q{0.06904438}{0.0660309} & \Q{0.2830398}{0.1179668} & \Q{0.18924760000000002}{0.0536466} & \Q{0.22900900000000002}{0.1240604}  \\
\Q{0.02369174}{0.01125394} & \Q{0.1984132}{0.19219499999999998} & \Q{1.0472808}{0.4296171999999999} & \Q{0.6299524}{0.17475759999999999} & \Q{0.7557484000000001}{0.41996199999999995}  \\
\Q{0.04508564}{0.02028034} & \Q{0.4661156}{0.454543} & \Q{3.5087316}{1.394961} & \Q{1.5073394}{0.4133368} & \Q{1.8878426000000001}{1.0313226}  \\
\Q{0.07140118}{0.0294554} & \Q{0.8887740000000001}{0.8690198} & \Q{7.595719000000001}{3.0078178} & \Q{3.8159959999999997}{0.9948924000000001} & \Q{4.4666652}{2.3928225999999997}  \\
}

\def\ANOITHREE{
\P{1000.0}{24975.0} & 2 & \P{1.0}{1.0} & \P{17.106}{8.652444} & \P{17.106}{8.652444} & \P{999.0}{982.1568} \\ 
\P{1000.0}{49950.0} & 2 & \P{1.0}{1.0} & \P{17.106}{8.722933999999999} & \P{17.106}{8.722933999999999} & \P{999.0}{965.7714} \\ 
\P{1000.0}{124875.0} & 2 & \P{1.0}{1.0} & \P{17.106}{8.784332000000001} & \P{17.106}{8.784332000000001} & \P{999.0}{920.324} \\ 
\P{1000.0}{249750.0} & 2 & \P{1.0}{1.0} & \P{17.106}{8.818312} & \P{17.106}{8.818312} & \P{999.0}{855.083} \\ 
\P{1000.0}{374625.0} & 2 & \P{1.0}{1.0} & \P{17.106}{8.831526} & \P{17.106}{8.831526} & \P{999.0}{798.9264000000001} \\ 
\P{1000.0}{499500.0} & 2 & \P{1.0}{1.0} & \P{17.106}{8.84085} & \P{17.106}{8.84085} & \P{999.0}{750.9464} \\ 
}
\def\BNOITHREE{
\Q{0.007653986}{0.0045384859999999996} & \Q{0.36697040000000003}{0.3627614} & \Q{0.5917178000000001}{0.251402} & \Q{0.9825208}{0.3562278} & \Q{1.1636484}{0.675597}  \\
\Q{0.01313482}{0.006659474} & \Q{0.7209248}{0.713557} & \Q{1.0196011999999999}{0.42462} & \Q{1.5608064000000001}{0.5359944000000001} & \Q{1.8562965999999999}{1.0926788}  \\
\Q{0.027939360000000003}{0.011452399999999998} & \Q{1.7815360000000002}{1.7670280000000003} & \Q{2.8487532}{1.1157982} & \Q{4.53506}{1.3340109999999998} & \Q{5.1592215999999995}{2.8560478000000002}  \\
\Q{0.05573376000000001}{0.01917668} & \Q{3.239444}{3.210866} & \Q{7.431343000000001}{2.8528046000000002} & \Q{11.7257996}{3.3059563999999995} & \Q{13.027487400000002}{7.0607112}  \\
\Q{0.09072526}{0.02793658} & \Q{4.696834}{4.6552500000000006} & \Q{11.7177866}{4.299871400000001} & \Q{18.233640599999998}{4.9047491999999995} & \Q{20.645401}{11.0815444}  \\
\Q{0.131793}{0.0374498} & \Q{6.396582}{6.340027999999999} & \Q{15.8665428}{5.6599488000000004} & \Q{24.6179486}{6.3898852} & \Q{25.937971999999995}{12.971258599999999}  \\
}

\def\ANOIFOUR{
\P{1000.0}{24975.0} & 3 & \P{2.002}{1.540196} & \P{16.289310689310692}{6.4818893179829065} & \P{32.611200000000004}{9.98338} & \P{502.14399999999995}{267.79499999999996} \\ 
\P{1000.0}{49950.0} & 3 & \P{2.002}{1.518196} & \P{17.343356643356643}{7.4968317661224235} & \P{34.721399999999996}{11.38166} & \P{503.28780000000006}{254.52079999999995} \\ 
\P{1000.0}{124875.0} & 3 & \P{2.002}{1.506342} & \P{16.48061938061938}{6.794029509898813} & \P{32.9942}{10.234131999999999} & \P{502.4192}{236.34179999999998} \\ 
\P{1000.0}{249750.0} & 3 & \P{2.002}{1.502422} & \P{17.091808191808195}{7.412832080467405} & \P{34.217800000000004}{11.137201999999998} & \P{505.6704}{223.12239999999997} \\ 
\P{1000.0}{374625.0} & 3 & \P{2.002}{1.4997159999999998} & \P{16.929070929070928}{7.379873256003137} & \P{33.891999999999996}{11.067713999999999} & \P{506.246}{209.33759999999998} \\ 
\P{1000.0}{499500.0} & 3 & \P{2.002}{1.4957699999999998} & \P{17.303196803196805}{7.760796111701667} & \P{34.641}{11.608366} & \P{506.2496}{196.9808} \\ 
}
\def\BNOIFOUR{
\Q{0.0125245}{0.007726012} & \Q{0.12073459999999998}{0.11761379999999999} & \Q{0.6704269999999999}{0.30622679999999997} & \Q{0.47012339999999997}{0.1690944} & \Q{0.5595752}{0.3380116}  \\
\Q{0.01721714}{0.009716293999999999} & \Q{0.22875040000000002}{0.2239638} & \Q{1.0099886000000002}{0.4517112} & \Q{0.7584681999999999}{0.24704259999999997} & \Q{0.9324062}{0.5547834}  \\
\Q{0.03614154000000001}{0.01766278} & \Q{0.4439266}{0.4342466} & \Q{2.9656244000000003}{1.2090252} & \Q{1.5911272}{0.4552184} & \Q{1.8588692000000002}{1.0258379999999998}  \\
\Q{0.06973732000000002}{0.02890128} & \Q{0.8842008}{0.865024} & \Q{7.6067431999999995}{3.0145998} & \Q{3.8187188}{0.9946545999999999} & \Q{4.4468608000000005}{2.3791556}  \\
\Q{0.1103}{0.0414418} & \Q{1.267928}{1.23986} & \Q{11.925386399999999}{4.5139924} & \Q{7.685649}{1.9593192} & \Q{8.6996662}{4.520223}  \\
\Q{0.15832759999999999}{0.05505462000000001} & \Q{1.663372}{1.6259059999999999} & \Q{16.130343600000003}{5.913485200000001} & \Q{11.1345242}{2.7733228} & \Q{12.4172796}{6.2421884}  \\
}

\def\ANOIFIVE{
\P{1000.0}{249750.0} & 2 & \P{1.0}{1.0} & \P{17.106}{8.818312} & \P{17.106}{8.818312} & \P{999.0}{855.083} \\ 
\P{1000.0}{249750.0} & 6.2 & \P{3.6124}{2.6222439999999994} & \P{12.26481009854944}{9.650924932996322} & \P{44.3054}{25.307080000000003} & \P{562.3804}{474.6188} \\ 
\P{1000.0}{249750.0} & 3 & \P{2.002}{1.502422} & \P{17.091808191808195}{7.412832080467405} & \P{34.217800000000004}{11.137201999999998} & \P{505.6704}{223.12239999999997} \\ 
\P{1000.0}{249750.0} & 4 & \P{2.0467999999999997}{1.23099} & \P{17.57367598202072}{7.300725432375568} & \P{35.969800000000006}{8.98712} & \P{227.0022}{57.77038} \\ 
\P{1000.0}{249750.0} & 4 & \P{2.3826}{1.32823} & \P{17.94762024678922}{7.351740285944452} & \P{42.762}{9.764802} & \P{122.44940000000001}{25.282559999999997} \\ 
\P{1000.0}{249750.0} & 5.6 & \P{2.7748}{1.579266} & \P{17.862476574888284}{9.494132084145418} & \P{49.564800000000005}{14.99376} & \P{62.56}{21.039119999999997} \\ 
\P{1000.0}{249750.0} & 6.4 & \P{2.9362}{1.712162} & \P{16.749199645800694}{9.68020549457353} & \P{49.178999999999995}{16.574080000000002} & \P{55.736799999999995}{24.013560000000002} \\ 
\P{1000.0}{249750.0} & 8 & \P{3.2312}{1.9432639999999999} & \P{14.131468185194356}{8.460764980980453} & \P{45.6616}{16.441499999999998} & \P{99.0848}{55.7013} \\ 
\P{1000.0}{249750.0} & 8 & \P{3.4198}{2.230466} & \P{14.110006433124743}{10.499976238149339} & \P{48.2534}{23.41984} & \P{272.03639999999996}{197.72420000000002} \\ 
\P{1000.0}{249750.0} & 6.4 & \P{3.532}{2.496152} & \P{13.576443941109853}{10.987744336082097} & \P{47.952}{27.427079999999997} & \P{494.6412}{409.8754} \\ 
}
\def\BNOIFIVE{
\Q{0.057001500000000004}{0.01960152} & \Q{3.2420299999999997}{3.2127460000000005} & \Q{7.4488684}{2.8596678} & \Q{11.8356362}{3.3412978000000004} & \Q{12.970643599999999}{7.035105400000001}  \\
\Q{0.09590067999999999}{0.038454880000000004} & \Q{1.375261}{1.3520244} & \Q{8.668163400000001}{4.0823278} & \Q{13.679091}{5.6879982} & \Q{22.284239200000002}{15.421907199999998}  \\
\Q{0.06919478}{0.028813} & \Q{0.8842612000000001}{0.8648087999999999} & \Q{7.6321118}{3.0268116} & \Q{3.8141962}{0.9951048} & \Q{4.4792787999999994}{2.4003242}  \\
\Q{0.06646752}{0.034080719999999995} & \Q{0.2644934}{0.2503838} & \Q{7.8500546}{3.2486726000000004} & \Q{0.8544271999999999}{0.23644379999999998} & \Q{1.3167558}{0.7000766}  \\
\Q{0.06984148000000001}{0.03416442} & \Q{0.1324798}{0.1196512} & \Q{8.144613}{3.5613116000000007} & \Q{0.40759239999999997}{0.167792} & \Q{0.7212409999999999}{0.4116676}  \\
\Q{0.07955024}{0.038257780000000005} & \Q{0.10121374000000001}{0.0879003} & \Q{9.214341600000001}{4.6396526} & \Q{0.8620664}{0.4742108} & \Q{0.8576790000000001}{0.5732134}  \\
\Q{0.08738074}{0.04276744} & \Q{0.10091704000000001}{0.08689842} & \Q{10.154416600000001}{5.5821048} & \Q{1.8062770000000001}{1.0504658} & \Q{1.3112458}{0.9423594}  \\
\Q{0.09273918}{0.04500556} & \Q{0.18011960000000002}{0.1649312} & \Q{10.8150626}{6.257327200000001} & \Q{5.0773436}{2.893979} & \Q{2.7273488}{2.0284862}  \\
\Q{0.0960635}{0.0450327} & \Q{0.6261444}{0.6072732000000001} & \Q{9.8041084}{5.2348404} & \Q{10.590464800000001}{5.335715} & \Q{9.324684600000001}{6.4784048}  \\
\Q{0.0948601}{0.041569140000000004} & \Q{0.8728111999999999}{0.850487} & \Q{8.771822799999999}{4.1851956} & \Q{13.3781146}{5.745402400000001} & \Q{19.3348556}{13.0546248}  \\
}

\def\ANOISIX{
\P{1000.0}{249750.0} & 2 & \P{1.0}{1.0} & \P{17.106}{8.815684} & \P{17.106}{8.815684} & \P{999.0}{855.0436000000002} \\ 
\P{1000.0}{249750.0} & 3 & \P{2.002}{1.49961} & \P{18.193106893106897}{8.53136482152026} & \P{36.4226}{12.793719999999999} & \P{506.0072}{223.6072} \\ 
\P{1000.0}{249750.0} & 2 & \P{1.0}{1.0} & \P{17.106}{8.840876} & \P{17.106}{8.840876} & \P{999.0}{855.608} \\ 
\P{1000.0}{249750.0} & 2 & \P{1.0}{1.0} & \P{17.450400000000002}{8.836794000000001} & \P{17.450400000000002}{8.836794000000001} & \P{999.0}{855.6677999999999} \\ 
\P{1000.0}{249750.0} & 2 & \P{1.0}{1.0} & \P{17.9932}{8.864355999999999} & \P{17.9932}{8.864355999999999} & \P{999.0}{855.7819999999999} \\ 
\P{1000.0}{249750.0} & 2 & \P{1.0}{1.0} & \P{18.7376}{8.444410000000001} & \P{18.7376}{8.444410000000001} & \P{999.0}{855.6964} \\ 
\P{1000.0}{249750.0} & 2 & \P{1.0}{1.0} & \P{18.752799999999997}{8.861825999999999} & \P{18.752799999999997}{8.861825999999999} & \P{999.0}{855.6458} \\ 
\P{1000.0}{249750.0} & 3 & \P{2.002}{1.423774} & \P{16.286813186813188}{7.731777655723449} & \P{32.6062}{11.008303999999999} & \P{505.27920000000006}{222.5462} \\ 
\P{1000.0}{249750.0} & 3 & \P{2.002}{1.502422} & \P{17.091808191808195}{7.412832080467405} & \P{34.217800000000004}{11.137201999999998} & \P{505.6704}{223.12239999999997} \\ 
}
\def\BNOISIX{
\Q{0.0586849}{0.019875280000000002} & \Q{3.1861939999999995}{3.1581960000000002} & \Q{7.4710408}{2.8839812} & \Q{11.7325374}{3.3156708000000004} & \Q{12.979194199999998}{7.012371399999999}  \\
\Q{0.07220422000000001}{0.03185632} & \Q{0.8713092}{0.852392} & \Q{7.6602184}{3.0474357999999997} & \Q{3.940063}{1.1133278} & \Q{4.6385844}{2.5013844}  \\
\Q{0.05950742}{0.02307764} & \Q{3.8470459999999997}{3.8190639999999996} & \Q{7.5793534000000005}{2.9859851999999996} & \Q{11.9371302}{3.5183405999999997} & \Q{14.8719778}{8.9012828}  \\
\Q{0.06694454000000001}{0.02979212} & \Q{4.3429199999999994}{4.315338} & \Q{7.792582}{3.1970378} & \Q{12.3444152}{3.935846} & \Q{15.091567800000002}{9.1038606}  \\
\Q{0.07231678}{0.03491154} & \Q{3.883896}{3.85625} & \Q{8.076376999999999}{3.4860892} & \Q{12.703514}{4.3028652} & \Q{15.917205000000001}{9.893867}  \\
\Q{0.06770624}{0.02945716} & \Q{3.961884}{3.933732} & \Q{8.4226066}{3.8338612000000003} & \Q{13.0866346}{4.6856212} & \Q{16.001866}{9.9631162}  \\
\Q{0.07372086}{0.03712912} & \Q{4.067241999999999}{4.038221999999999} & \Q{8.739317999999999}{4.1512286} & \Q{13.376771600000001}{4.964092600000001} & \Q{16.620753999999998}{10.5851044}  \\
\Q{0.06683356}{0.0261244} & \Q{0.8828811999999999}{0.863634} & \Q{7.636826600000001}{3.0368876} & \Q{3.8929258000000004}{1.0238482} & \Q{4.4061006}{2.3214124}  \\
\Q{0.07051347999999999}{0.029121400000000002} & \Q{0.8876924}{0.8681547999999999} & \Q{7.6106646}{3.0171794000000003} & \Q{3.8056748000000007}{0.9941566} & \Q{4.470398}{2.3964635999999997}  \\
}

\def\APATH{
\P{2000.0}{21990.0} & 2 & \P{1.0}{1.0} & \P{19.0845}{8.970023999999999} & \P{19.0845}{8.970023999999999} & \P{1999.0}{1992.2559999999999} \\ 
\P{2000.0}{21990.0} & 4 & \P{1.7401}{1.5879759999999998} & \P{10.53261306821447}{2.694280014307521} & \P{18.3278}{4.278452} & \P{519.2232}{242.90699999999998} \\ 
\P{2000.0}{21990.0} & 4 & \P{2.0082}{2.027068} & \P{11.646150781794642}{4.004524761872814} & \P{23.3878}{8.117444} & \P{168.0494}{78.87748} \\ 
\P{2000.0}{21990.0} & 6.6 & \P{2.5321}{2.394188} & \P{13.689190790253154}{5.273311870245778} & \P{34.662400000000005}{12.6253} & \P{78.5892}{52.63271999999999} \\ 
\P{2000.0}{21990.0} & 13 & \P{3.6995000000000005}{2.8804719999999997} & \P{14.584727665900793}{7.907669298642722} & \P{53.95619999999999}{22.77782} & \P{59.842600000000004}{59.96307999999999} \\ 
\P{2000.0}{21990.0} & 26.4 & \P{4.7306}{3.587822} & \P{14.021202384475542}{9.392021120334286} & \P{66.3287}{33.6969} & \P{77.3018}{99.79042000000001} \\ 
\P{2000.0}{21990.0} & 40.4 & \P{5.1345}{3.909024} & \P{14.172460804362643}{10.705449749093379} & \P{72.76849999999999}{41.84786} & \P{109.63}{157.98160000000001} \\ 
}
\def\BPATH{
\Q{0.0073666700000000005}{0.002747418} & \Q{0.009743176000000001}{0.005318612} & \Q{0.806926}{0.043058400000000004} & \Q{2.1776194}{0.7250786} & \Q{2.1063754}{1.0948201999999998}  \\
\Q{0.010613934}{0.00486981} & \Q{0.009286446}{0.004562775999999999} & \Q{0.9339156}{0.1894444} & \Q{2.6880485999999997}{1.2314805999999998} & \Q{0.3268254}{0.07763020000000001}  \\
\Q{0.01534944}{0.009184104} & \Q{0.015764000000000004}{0.011999491999999999} & \Q{1.3346998}{0.5797764} & \Q{3.3435078}{1.9124546000000002} & \Q{0.2137108}{0.0377448}  \\
\Q{0.02316012}{0.01597122} & \Q{0.032656679999999993}{0.02878926} & \Q{1.7645298}{0.9934118000000002} & \Q{3.1378703999999997}{1.9517571999999999} & \Q{0.2471162}{0.0676956}  \\
\Q{0.0375287}{0.028454840000000002} & \Q{0.054104740000000005}{0.050400900000000005} & \Q{2.1370156000000002}{1.3587817999999998} & \Q{1.0752138000000002}{0.6035672} & \Q{0.350985}{0.16332739999999996}  \\
\Q{0.06689866}{0.055237499999999995} & \Q{0.0989714}{0.09491962} & \Q{3.8275722}{3.0590192} & \Q{0.9934622000000001}{0.6133778} & \Q{0.8578994}{0.6193514}  \\
\Q{0.08795198}{0.0753786} & \Q{0.15757559999999998}{0.15326040000000002} & \Q{7.9084506}{7.1394766} & \Q{1.2716216}{0.8939523999999999} & \Q{1.6946138000000002}{1.4178226}  \\
}

\def\APRONE{
\P{400.0}{1986.8} & 3 & \P{1.798}{1.7995959999999998} & \P{10.596496106785319}{6.5000811293201375} & \P{19.052500000000002}{11.69752} & \P{397.26}{397.85119999999995} \\ 
\P{800.0}{7151.4} & 2 & \P{1.0}{1.0} & \P{16.79076}{8.537836} & \P{16.79076}{8.537836} & \P{799.0}{799.0} \\ 
\P{1200.0}{15522.4} & 2 & \P{1.0}{1.0} & \P{17.62816}{8.793786} & \P{17.62816}{8.793786} & \P{1199.0}{1199.0} \\ 
\P{1600.0}{27019.6} & 2 & \P{1.0}{1.0} & \P{18.42612}{9.238428} & \P{18.42612}{9.238428} & \P{1599.0}{1599.0} \\ 
\P{2000.0}{41731.2} & 2 & \P{1.0}{1.0} & \P{19.0845}{9.58615} & \P{19.0845}{9.58615} & \P{1999.0}{1999.0} \\ 
}
\def\BPRONE{
\Q{0.002381114}{0.0014811528000000002} & \Q{0.019637800000000004}{0.018959359999999998} & \Q{0.026901}{0.011756199999999998} & \Q{0.0485594}{0.022336599999999998} & \Q{0.06090720000000001}{0.0389552}  \\
\Q{0.005001152}{0.0031771799999999995} & \Q{0.08813365999999999}{0.08645076} & \Q{0.18034920000000002}{0.0669838} & \Q{0.3084632}{0.1137678} & \Q{0.3059436}{0.17173}  \\
\Q{0.008049794}{0.0049934919999999995} & \Q{0.2689858}{0.2654656} & \Q{0.5379122}{0.211463} & \Q{0.9107466000000001}{0.33122060000000003} & \Q{0.9531152}{0.5259576}  \\
\Q{0.010931951999999998}{0.006607394} & \Q{0.6208102}{0.6150266} & \Q{1.1087745999999998}{0.4500438} & \Q{1.9571268}{0.7091194} & \Q{2.1108108000000003}{1.2160256}  \\
\Q{0.014285619999999999}{0.008176918} & \Q{1.232836}{1.223436} & \Q{2.0556042}{0.8411012} & \Q{3.5724080000000002}{1.2775850000000002} & \Q{3.8054734000000003}{2.1680018}  \\
}

\def\APRFIVE{
\P{400.0}{1986.8} & 5 & \P{2.006}{1.798198} & \P{13.064307078763711}{5.850984151912082} & \P{26.207}{10.521228000000002} & \P{187.411}{152.1126} \\ 
\P{800.0}{7151.4} & 3.2 & \P{2.0025}{1.670016} & \P{15.49487141073658}{6.558430577910632} & \P{31.028480000000002}{10.952684000000001} & \P{403.686}{271.9886} \\ 
\P{1200.0}{15522.4} & 3 & \P{2.00167}{1.6138379999999999} & \P{16.35051731803944}{6.471619828012478} & \P{32.72834}{10.444146} & \P{601.4972}{369.3806} \\ 
\P{1600.0}{27019.6} & 3 & \P{2.00125}{1.588752} & \P{17.39262960649594}{6.858037000110779} & \P{34.807}{10.89572} & \P{801.5980000000001}{472.259} \\ 
\P{2000.0}{41731.2} & 3 & \P{2.001}{1.571962} & \P{18.099300349825086}{7.1772345641943} & \P{36.216699999999996}{11.28234} & \P{1001.5}{573.1028} \\ 
}
\def\BPRFIVE{
\Q{0.0037951260000000002}{0.0024553659999999996} & \Q{0.010152294000000001}{0.009438956} & \Q{0.0340306}{0.019634199999999997} & \Q{0.0299308}{0.0149} & \Q{0.036836999999999995}{0.025055}  \\
\Q{0.009154187999999999}{0.005900996} & \Q{0.04569867999999999}{0.044137039999999995} & \Q{0.24485980000000002}{0.1158166} & \Q{0.1527476}{0.062363800000000004} & \Q{0.1822162}{0.1170558}  \\
\Q{0.01210436}{0.007826528} & \Q{0.1157314}{0.112869} & \Q{0.581723}{0.2762494} & \Q{0.4436508}{0.1697812} & \Q{0.49621560000000003}{0.30294239999999995}  \\
\Q{0.01667496}{0.010622103999999999} & \Q{0.22658040000000002}{0.2212348} & \Q{1.249615}{0.5935474} & \Q{0.9591008000000001}{0.36362719999999993} & \Q{1.0971123999999999}{0.668677}  \\
\Q{0.022043280000000002}{0.0138561} & \Q{0.3889874000000001}{0.3817322} & \Q{2.2495443999999996}{1.0602588000000002} & \Q{1.751387}{0.6541156000000001} & \Q{1.960414}{1.1896044}  \\
}

\def\APRSIX{
\P{400.0}{8345.2} & 3 & \P{2.005}{1.5664840000000002} & \P{13.13142144638404}{5.231890016112515} & \P{26.3285}{8.195671999999998} & \P{201.2936}{114.16119999999998} \\ 
\P{800.0}{32601.6} & 3 & \P{2.0025}{1.534882} & \P{15.779895131086143}{6.52069409896005} & \P{31.599239999999998}{10.008496} & \P{402.39399999999995}{216.4632} \\ 
\P{1200.0}{72936.4} & 3 & \P{2.00167}{1.523372} & \P{17.000964194897264}{6.935720231171376} & \P{34.03032}{10.565682} & \P{602.0964}{316.1228} \\ 
\P{1600.0}{129196.4} & 3 & \P{2.00125}{1.517706} & \P{17.36783510306059}{6.800239308535382} & \P{34.757380000000005}{10.320764} & \P{801.9975999999999}{416.1316} \\ 
\P{2000.0}{201455.4} & 3 & \P{2.001}{1.5148740000000003} & \P{18.35192403798101}{7.363885049185607} & \P{36.7222}{11.155358} & \P{1001.6}{516.2506000000001} \\ 
}
\def\BPRSIX{
\Q{0.005108492}{0.003099696} & \Q{0.02349172}{0.02248646} & \Q{0.09242239999999999}{0.0405032} & \Q{0.053622199999999995}{0.020683800000000002} & \Q{0.0680786}{0.041434}  \\
\Q{0.011342379999999999}{0.006496608000000001} & \Q{0.131782}{0.128463} & \Q{0.5824317999999999}{0.2588174} & \Q{0.4072694}{0.1340894} & \Q{0.4868226}{0.28272699999999995}  \\
\Q{0.02272376}{0.012778859999999998} & \Q{0.38678739999999995}{0.379527} & \Q{1.8438668}{0.8115045999999999} & \Q{1.2989602}{0.40281599999999995} & \Q{1.5399652}{0.8867345999999999}  \\
\Q{0.039394599999999995}{0.0201892} & \Q{0.7847017999999999}{0.7714198} & \Q{5.453562}{2.2815984} & \Q{2.9573443999999998}{0.8851516} & \Q{3.4693917999999995}{1.9733209999999999}  \\
\Q{0.06132749999999999}{0.03019212} & \Q{1.495344}{1.4750159999999999} & \Q{12.268595000000001}{5.1404206} & \Q{6.569169199999999}{1.8337896} & \Q{7.6596876}{4.2394804}  \\
}

\def\APRSEVEN{
\P{400.0}{40252.6} & 3 & \P{2.005}{1.510638} & \P{13.297755610972569}{5.343166264849686} & \P{26.662}{8.07159} & \P{201.39319999999998}{103.2126} \\ 
\P{800.0}{160309.0} & 3 & \P{2.0025}{1.505316} & \P{15.38450936329588}{6.226931753864305} & \P{30.807479999999998}{9.3735} & \P{401.7958}{203.8168} \\ 
\P{1200.0}{360614.8} & 3 & \P{2.00167}{1.5037740000000002} & \P{16.551259698152045}{6.543488582725861} & \P{33.130160000000004}{9.839928} & \P{601.597}{303.6588} \\ 
\P{1600.0}{640753.2} & 3 & \P{2.00125}{1.5026220000000001} & \P{17.265199250468456}{6.700707163877542} & \P{34.55198}{10.06863} & \P{801.4984000000001}{403.3424} \\ 
\P{2000.0}{1000790.2} & 3 & \P{2.001}{1.502188} & \P{18.551974012993504}{7.532927969069116} & \P{37.1225}{11.315873999999999} & \P{1001.5}{503.43559999999997} \\ 
}
\def\BPRSEVEN{
\Q{0.010607292}{0.005216429999999999} & \Q{0.07659988000000001}{0.07374458} & \Q{0.29192840000000003}{0.1292162} & \Q{0.2263382}{0.0635478} & \Q{0.2476318}{0.1304236}  \\
\Q{0.03773282}{0.017594859999999997} & \Q{0.5300471999999999}{0.5183586} & \Q{3.4157494}{1.4153976} & \Q{1.6687176000000001}{0.4659584} & \Q{1.8786798}{1.0094496}  \\
\Q{0.09350652}{0.039000179999999995} & \Q{1.7277740000000001}{1.69765} & \Q{12.7323704}{5.124987} & \Q{9.5755488}{2.5985326} & \Q{10.535532}{5.5012998}  \\
\Q{0.1954142}{0.0757693} & \Q{3.9850919999999994}{3.9261220000000003} & \Q{30.3950048}{11.9309734} & \Q{25.361935999999996}{6.7016696} & \Q{28.0148194}{14.6309866}  \\
\Q{0.3330646}{0.12605539999999998} & \Q{7.694462}{7.596446} & \Q{59.606494}{23.299953600000002} & \Q{51.235544399999995}{13.1722094} & \Q{59.677415399999994}{32.355335600000004}  \\
}

\def\APREIGHT{
\P{400.0}{79010.0} & 3 & \P{2.005}{1.5038040000000001} & \P{14.032668329177058}{6.286611819093444} & \P{28.1355}{9.453831999999998} & \P{202.28879999999998}{103.20219999999999} \\ 
\P{800.0}{316399.4} & 3 & \P{2.0025}{1.501872} & \P{15.928229712858926}{6.8351510648044576} & \P{31.896279999999997}{10.265522} & \P{402.19440000000003}{203.0444} \\ 
\P{1200.0}{712166.0} & 3 & \P{2.00167}{1.501284} & \P{16.503299744713164}{6.548345283104329} & \P{33.03416}{9.830926} & \P{601.3974000000001}{301.8676} \\ 
\P{1600.0}{1266394.4} & 3 & \P{2.00125}{1.500896} & \P{17.76320799500312}{7.264412724132785} & \P{35.54862}{10.903128} & \P{801.8976}{402.5656} \\ 
\P{2000.0}{1978984.6} & 3 & \P{2.001}{1.500758} & \P{18.751324337831086}{7.828883804051018} & \P{37.5214}{11.749259999999998} & \P{1002.2}{503.05640000000005} \\ 
}
\def\BPREIGHT{
\Q{0.01533746}{0.006747132} & \Q{0.1536886}{0.1488222} & \Q{0.6775711999999999}{0.29852259999999997} & \Q{0.409197}{0.11640819999999999} & \Q{0.42455479999999995}{0.22027380000000002}  \\
\Q{0.06854025999999999}{0.027011779999999996} & \Q{1.168636}{1.145604} & \Q{7.2790194}{2.9259812} & \Q{5.3351234}{1.4638732} & \Q{5.638775600000001}{2.857622}  \\
\Q{0.1809806}{0.06377150000000001} & \Q{3.727198}{3.669444} & \Q{25.142334800000004}{9.832786800000001} & \Q{21.325452}{5.539895} & \Q{22.9044576}{11.6447934}  \\
\Q{0.3526242}{0.1212822} & \Q{8.700972}{8.587641999999999} & \Q{60.1293036}{23.1732886} & \Q{51.92189040000001}{13.158719600000001} & \Q{55.499157000000004}{27.7716558}  \\
\Q{0.5950294}{0.20696459999999997} & \Q{18.22782}{18.021819999999998} & \Q{117.45659479999999}{45.4418312} & \Q{102.387973}{25.7238026} & \Q{111.4467392}{56.483566399999994}  \\
}

\def\AREGONE{
\P{1000.0}{1000.0} & 2-49 & \P{1.0}{1.0} & \P{66.632}{63.275} & \P{66.632}{63.275} & \P{999.0}{999.0} \\ 
\P{1000.0}{4000.0} & 2 & \P{1.0}{1.0} & \P{17.921}{10.32972} & \P{17.921}{10.32972} & \P{999.0}{997.6256} \\ 
\P{1000.0}{16000.0} & 2 & \P{1.0}{1.0} & \P{19.0726}{10.90336} & \P{19.0726}{10.90336} & \P{999.0}{989.4626000000001} \\ 
\P{1000.0}{64000.0} & 2 & \P{1.0}{1.0} & \P{18.0844}{9.882278} & \P{18.0844}{9.882278} & \P{999.0}{958.5840000000001} \\ 
\P{1000.0}{128000.0} & 2 & \P{1.0}{1.0} & \P{18.085}{9.880774} & \P{18.085}{9.880774} & \P{999.0}{921.0234} \\ 
}
\def\BREGONE{
\Q{0.0037123399999999993}{0.002323888} & \Q{0.0037323960000000002}{0.000953588} & \Q{0.10434460000000001}{0.028914799999999997} & \Q{0.13426860000000002}{0.0399204} & \Q{0.146659}{0.0460428}  \\
\Q{0.00517072}{0.0034599340000000004} & \Q{0.07136400000000001}{0.06943834} & \Q{0.18109419999999998}{0.0775322} & \Q{0.2439194}{0.10901779999999998} & \Q{0.2536236}{0.1378904}  \\
\Q{0.007095417999999999}{0.0046687659999999995} & \Q{0.20662180000000002}{0.2030602} & \Q{0.4473574}{0.21792540000000002} & \Q{0.6688274}{0.24549379999999998} & \Q{0.5801968}{0.325101}  \\
\Q{0.017391200000000002}{0.007598903999999999} & \Q{0.8054012}{0.7961278} & \Q{1.2930784}{0.5774144} & \Q{1.9394174}{0.6362943999999999} & \Q{2.4543345999999997}{1.599165}  \\
\Q{0.031511660000000004}{0.011178440000000001} & \Q{1.162472}{1.14669} & \Q{3.1154332000000005}{1.2738704} & \Q{4.687476599999999}{1.3466917999999999} & \Q{8.044814999999998}{5.358095400000001}  \\
}

\def\ATREE{
\P{800.0}{160600.0} & 2 & \P{1.0}{1.0} & \P{16.4175}{8.506514} & \P{16.4175}{8.506514} & \P{799.0}{685.501} \\ 
\P{800.0}{160600.0} & 4 & \P{1.0}{1.0} & \P{14.97326}{4.885458} & \P{14.97326}{4.885458} & \P{454.47299999999996}{263.3916} \\ 
\P{800.0}{160600.0} & 4 & \P{1.0}{1.0} & \P{14.590259999999997}{4.062878} & \P{14.590259999999997}{4.062878} & \P{378.1278}{205.47820000000002} \\ 
\P{800.0}{160600.0} & 4 & \P{1.282}{1.112526} & \P{11.426271450858032}{3.2238221848298374} & \P{14.648479999999998}{3.5865859999999996} & \P{291.4662}{150.28298} \\ 
\P{800.0}{160600.0} & 4.4 & \P{1.491}{1.2125059999999999} & \P{11.352783366867873}{3.5801257890682607} & \P{16.927}{4.340924} & \P{236.6808}{121.89962} \\ 
\P{800.0}{160600.0} & 6 & \P{1.6164999999999998}{1.328058} & \P{12.492731209403031}{4.354504095453662} & \P{20.194499999999998}{5.783034} & \P{219.1462}{128.5964} \\ 
\P{800.0}{160600.0} & 6.2 & \P{1.78325}{1.481706} & \P{12.669418197112014}{4.944635440498993} & \P{22.59274}{7.326496000000001} & \P{212.19279999999998}{134.93779999999998} \\ 
\P{800.0}{160600.0} & 7.2 & \P{2.24375}{1.758194} & \P{12.235204456824512}{5.780295007263136} & \P{27.45274}{10.162880000000001} & \P{227.57840000000002}{158.9234} \\ 
\P{800.0}{160600.0} & 7.8 & \P{2.5782499999999997}{1.984366} & \P{12.771856879666442}{6.515924985612533} & \P{32.92904}{12.92998} & \P{249.44539999999998}{191.9864} \\ 
}
\def\BTREE{
\Q{0.02791666}{0.00539632} & \Q{0.178679}{0.1643538} & \Q{2.3134384}{0.1594554} & \Q{5.6924092}{1.714434} & \Q{5.9569350000000005}{3.0586946}  \\
\Q{0.027473020000000004}{0.006871756} & \Q{0.10297448000000001}{0.09162398000000001} & \Q{2.6681528}{0.5210778} & \Q{6.0219302}{2.1675032} & \Q{1.8413650000000001}{0.7798733999999999}  \\
\Q{0.0290756}{0.008409709999999999} & \Q{0.09281606}{0.08229782000000001} & \Q{2.8128237999999994}{0.6908982} & \Q{6.442002}{2.5327758} & \Q{0.9119997999999999}{0.4053196}  \\
\Q{0.032065800000000005}{0.010275567999999999} & \Q{0.07734906}{0.06721414} & \Q{3.0672572000000002}{0.9226130000000001} & \Q{6.6964163999999995}{2.8650114} & \Q{0.4889962}{0.2195754}  \\
\Q{0.03743265999999999}{0.013095260000000001} & \Q{0.06935972}{0.06026222} & \Q{3.2926843999999997}{1.1470145999999999} & \Q{6.711486600000001}{2.976668} & \Q{0.3663232}{0.1730826}  \\
\Q{0.04064004}{0.01670582} & \Q{0.09235661999999999}{0.0825451} & \Q{3.4652554}{1.3242606000000001} & \Q{6.9960252}{3.3143496} & \Q{0.5506506}{0.287116}  \\
\Q{0.049454}{0.02468402} & \Q{0.1782562}{0.16820739999999998} & \Q{3.5990572}{1.4568562} & \Q{7.305594999999999}{3.5863416} & \Q{0.8834500000000001}{0.520389}  \\
\Q{0.06059388}{0.03357365999999999} & \Q{0.3682796}{0.3563916} & \Q{3.7814563999999997}{1.6377758} & \Q{7.0789990000000005}{3.5618590000000006} & \Q{1.8389412}{1.257272}  \\
\Q{0.07137492}{0.04209148} & \Q{0.3616461999999999}{0.3496392} & \Q{3.9694781999999997}{1.838635} & \Q{6.771694800000001}{3.5117933999999997} & \Q{3.3231803999999996}{2.4595125999999996}  \\
}

\def\AWHE{
\P{1024.0}{2046.0} & 2 & \P{1.0}{1.0} & \P{32.6836}{22.1877} & \P{32.6836}{22.1877} & \P{1023.0}{1023.0} \\ 
\P{2048.0}{4094.0} & 2 & \P{1.0}{1.0} & \P{35.0293}{22.6778} & \P{35.0293}{22.6778} & \P{2047.0}{2047.0} \\ 
\P{4196.0}{8390.0} & 2 & \P{1.0}{1.0} & \P{47.7776}{31.5483} & \P{47.7776}{31.5483} & \P{4195.0}{4195.0} \\ 
}
\def\BWHE{
\Q{0.007839697999999999}{0.006147472} & \Q{0.07077937999999999}{0.0699804} & \Q{2.9001520000000003}{2.8118811999999997} & \Q{2.9881026}{2.8736327999999998} & \Q{0.168514}{0.06910039999999999}  \\
\Q{0.0183575}{0.01599202} & \Q{0.2751012}{0.27331019999999995} & \Q{21.0215974}{20.651191} & \Q{20.3091396}{19.8176036} & \Q{0.6886864000000001}{0.30155580000000004}  \\
\Q{0.057898700000000004}{0.05401754000000001} & \Q{1.0666108}{1.062039} & \Q{168.3279844}{166.73970960000003} & \Q{148.8357484}{146.434732} & \Q{3.1668634}{1.4448196000000002}  \\
}

\def\AALL{
BIKEWHE &
\P{4196.0}{8389.0} & 2 & \P{1.0}{1.0} & \P{43.2896}{28.9809} & \P{43.2896}{28.9809} & \P{4195.0}{4195.0} \\ 
CYC1 &
\P{16784.0}{16784.0} & 15267-15885 & \P{7.39377}{7.39377} & \P{21.891403167802086}{20.7876360774003} & \P{161.86}{153.699} & \P{26.7837}{26.7837} \\ 
DBLCYC &
\P{8392.0}{16784.0} & 5 & \P{2.00024}{2.00024} & \P{53.9960204775427}{50.61742590889094} & \P{108.005}{101.247} & \P{4194.51}{4194.35} \\ 
NOI1 &
\P{1000.0}{249750.0} & 2 & \P{1.0}{1.0} & \P{17.106}{8.818312} & \P{17.106}{8.818312} & \P{999.0}{855.083} \\ 
NOI2 &
\P{1000.0}{249750.0} & 3 & \P{2.002}{1.502422} & \P{17.091808191808195}{7.412832080467405} & \P{34.217800000000004}{11.137201999999998} & \P{505.6704}{223.12239999999997} \\ 
NOI3 &
\P{1000.0}{499500.0} & 2 & \P{1.0}{1.0} & \P{17.106}{8.84085} & \P{17.106}{8.84085} & \P{999.0}{750.9464} \\ 
NOI4 &
\P{1000.0}{499500.0} & 3 & \P{2.002}{1.4957699999999998} & \P{17.303196803196805}{7.760796111701667} & \P{34.641}{11.608366} & \P{506.2496}{196.9808} \\ 
NOI5 &
\P{1000.0}{249750.0} & 6.2 & \P{3.6124}{2.6222439999999994} & \P{12.26481009854944}{9.650924932996322} & \P{44.3054}{25.307080000000003} & \P{562.3804}{474.6188} \\ 
NOI6 &
\P{1000.0}{249750.0} & 3 & \P{2.002}{1.49961} & \P{18.193106893106897}{8.53136482152026} & \P{36.4226}{12.793719999999999} & \P{506.0072}{223.6072} \\ 
PATH &
\P{2000.0}{21990.0} & 40.4 & \P{5.1345}{3.909024} & \P{14.172460804362643}{10.705449749093379} & \P{72.76849999999999}{41.84786} & \P{109.63}{157.98160000000001} \\ 
PR1 &
\P{2000.0}{41731.2} & 2 & \P{1.0}{1.0} & \P{19.0845}{9.58615} & \P{19.0845}{9.58615} & \P{1999.0}{1999.0} \\ 
PR5 &
\P{2000.0}{41731.2} & 3 & \P{2.001}{1.571962} & \P{18.099300349825086}{7.1772345641943} & \P{36.216699999999996}{11.28234} & \P{1001.5}{573.1028} \\ 
PR6 &
\P{2000.0}{201455.4} & 3 & \P{2.001}{1.5148740000000003} & \P{18.35192403798101}{7.363885049185607} & \P{36.7222}{11.155358} & \P{1001.6}{516.2506000000001} \\ 
PR7 &
\P{2000.0}{1000790.2} & 3 & \P{2.001}{1.502188} & \P{18.551974012993504}{7.532927969069116} & \P{37.1225}{11.315873999999999} & \P{1001.5}{503.43559999999997} \\ 
PR8 &
\P{2000.0}{1978984.6} & 3 & \P{2.001}{1.500758} & \P{18.751324337831086}{7.828883804051018} & \P{37.5214}{11.749259999999998} & \P{1002.2}{503.05640000000005} \\ 
REG1 &
\P{1000.0}{128000.0} & 2 & \P{1.0}{1.0} & \P{18.085}{9.880774} & \P{18.085}{9.880774} & \P{999.0}{921.0234} \\ 
TREE &
\P{800.0}{160600.0} & 7.8 & \P{2.5782499999999997}{1.984366} & \P{12.771856879666442}{6.515924985612533} & \P{32.92904}{12.92998} & \P{249.44539999999998}{191.9864} \\ 
WHE &
\P{4196.0}{8390.0} & 2 & \P{1.0}{1.0} & \P{47.7776}{31.5483} & \P{47.7776}{31.5483} & \P{4195.0}{4195.0} \\ 
}
\def\BALL{
\Q{0.098969}{0.09543172} & \Q{1.068994}{1.065664} & \Q{177.8264038}{176.33547059999998} & \Q{154.74469919999999}{152.45422979999998} & \Q{22.2316876}{20.7858218}  \\
\Q{0.173987}{0.0994704} & \Q{0.0334107}{0.02139} & \Q{53.268906}{16.128618} & \Q{27.557215}{0.040244} & \Q{21.189104}{0.039382}  \\
\Q{0.0565502}{0.0390912} & \Q{0.0190084}{0.0138528} & \Q{1733.210693}{1725.780884} & \Q{60.329689}{51.719501} & \Q{6.63112}{1.432986}  \\
\Q{0.05597426}{0.01929914} & \Q{3.231376}{3.203306} & \Q{7.481244200000001}{2.874071} & \Q{11.7368366}{3.3123036} & \Q{13.0124244}{7.0519663999999995}  \\
\Q{0.07140118}{0.0294554} & \Q{0.8887740000000001}{0.8690198} & \Q{7.595719000000001}{3.0078178} & \Q{3.8159959999999997}{0.9948924000000001} & \Q{4.4666652}{2.3928225999999997}  \\
\Q{0.131793}{0.0374498} & \Q{6.396582}{6.340027999999999} & \Q{15.8665428}{5.6599488000000004} & \Q{24.6179486}{6.3898852} & \Q{25.937971999999995}{12.971258599999999}  \\
\Q{0.15832759999999999}{0.05505462000000001} & \Q{1.663372}{1.6259059999999999} & \Q{16.130343600000003}{5.913485200000001} & \Q{11.1345242}{2.7733228} & \Q{12.4172796}{6.2421884}  \\
\Q{0.09590067999999999}{0.038454880000000004} & \Q{1.375261}{1.3520244} & \Q{8.668163400000001}{4.0823278} & \Q{13.679091}{5.6879982} & \Q{22.284239200000002}{15.421907199999998}  \\
\Q{0.07220422000000001}{0.03185632} & \Q{0.8713092}{0.852392} & \Q{7.6602184}{3.0474357999999997} & \Q{3.940063}{1.1133278} & \Q{4.6385844}{2.5013844}  \\
\Q{0.08795198}{0.0753786} & \Q{0.15757559999999998}{0.15326040000000002} & \Q{7.9084506}{7.1394766} & \Q{1.2716216}{0.8939523999999999} & \Q{1.6946138000000002}{1.4178226}  \\
\Q{0.014285619999999999}{0.008176918} & \Q{1.232836}{1.223436} & \Q{2.0556042}{0.8411012} & \Q{3.5724080000000002}{1.2775850000000002} & \Q{3.8054734000000003}{2.1680018}  \\
\Q{0.022043280000000002}{0.0138561} & \Q{0.3889874000000001}{0.3817322} & \Q{2.2495443999999996}{1.0602588000000002} & \Q{1.751387}{0.6541156000000001} & \Q{1.960414}{1.1896044}  \\
\Q{0.06132749999999999}{0.03019212} & \Q{1.495344}{1.4750159999999999} & \Q{12.268595000000001}{5.1404206} & \Q{6.569169199999999}{1.8337896} & \Q{7.6596876}{4.2394804}  \\
\Q{0.3330646}{0.12605539999999998} & \Q{7.694462}{7.596446} & \Q{59.606494}{23.299953600000002} & \Q{51.235544399999995}{13.1722094} & \Q{59.677415399999994}{32.355335600000004}  \\
\Q{0.5950294}{0.20696459999999997} & \Q{18.22782}{18.021819999999998} & \Q{117.45659479999999}{45.4418312} & \Q{102.387973}{25.7238026} & \Q{111.4467392}{56.483566399999994}  \\
\Q{0.031511660000000004}{0.011178440000000001} & \Q{1.162472}{1.14669} & \Q{3.1154332000000005}{1.2738704} & \Q{4.687476599999999}{1.3466917999999999} & \Q{8.044814999999998}{5.358095400000001}  \\
\Q{0.07137492}{0.04209148} & \Q{0.3616461999999999}{0.3496392} & \Q{3.9694781999999997}{1.838635} & \Q{6.771694800000001}{3.5117933999999997} & \Q{3.3231803999999996}{2.4595125999999996}  \\
\Q{0.057898700000000004}{0.05401754000000001} & \Q{1.0666108}{1.062039} & \Q{168.3279844}{166.73970960000003} & \Q{148.8357484}{146.434732} & \Q{3.1668634}{1.4448196000000002}  \\
}

\def\ATSP{
       &
\P{5933.0}{11868.0} & 145-523 & \P{3.48753}{3.29331} & \P{44.953018325290394}{31.929578448430302} & \P{156.775}{105.154} & \P{1574.29}{1312.29} \\ 
rl5934 &
\P{5934.0}{23736.0} & 77-95 & \P{3.80435}{3.75426} & \P{22.58979852011513}{19.45339427743417} & \P{85.9395}{73.0331} & \P{871.341}{866.298} \\ 
       &
\P{5934.0}{47472.0} & 17-19 & \P{2.90529}{2.90795} & \P{24.96766243645213}{18.569301397891987} & \P{72.5383}{53.9986} & \P{4515.03}{4544.7} \\ 
\hline
        &
\P{11846.0}{23698.0} & 171-605 & \P{3.77838}{3.60313} & \P{44.36848596488442}{35.4303064280223} & \P{167.641}{127.66} & \P{2308.02}{2005.44} \\ 
rl11849 &
\P{11845.0}{47396.0} & 69-112 & \P{4.08651}{4.06494} & \P{27.669576239872168}{23.482363823328267} & \P{113.072}{95.4544} & \P{1436.38}{1498.93} \\ 
        &
\P{11847.0}{94792.0} & 24-28 & \P{3.02751}{3.00714} & \P{31.434446128997095}{24.63626568766336} & \P{95.1681}{74.0847} & \P{9178.85}{9299.45} \\ 
\hline
         &
\P{13504.0}{27018.0} & 79-5517 & \P{2.39078}{2.63002} & \P{704.1551292883495}{58.98320164865667} & \P{1683.48}{155.127} & \P{7312.77}{3968.04} \\ 
usa13509 &
\P{13499.0}{54036.0} & 116-2508 & \P{3.20157}{3.34827} & \P{158.18270411079564}{38.607101577829745} & \P{506.433}{129.267} & \P{4312.63}{2761.13} \\ 
         &
\P{13492.0}{108072.0} & 59-884 & \P{3.12977}{3.15546} & \P{46.63377820095406}{23.817890260057172} & \P{145.953}{75.1564} & \P{3806.49}{3058.18} \\ 
\hline
         &
\P{14049.0}{28102.0} & 215-2111 & \P{3.67355}{3.94178} & \P{114.11686243551877}{44.707974569864376} & \P{419.214}{176.229} & \P{3639.37}{2227.69} \\ 
brd14051 &
\P{14051.0}{56204.0} & 79-252 & \P{3.28902}{3.261} & \P{31.087983654705656}{24.679116835326585} & \P{102.249}{80.4786} & \P{3682.79}{3440.98} \\ 
         &
\P{14051.0}{112408.0} & 18-27 & \P{2.9805}{2.99087} & \P{27.039020298607614}{21.396683908026763} & \P{80.5898}{63.9947} & \P{11585.7}{11661.8} \\ 
\hline
       &
\P{18512.0}{37024.0} & 232-2251 & \P{3.79295}{3.99001} & \P{110.58305540542322}{46.6557727925494} & \P{419.436}{186.157} & \P{4125.23}{2514.42} \\ 
d18512 &
\P{18510.0}{74048.0} & 69-186 & \P{3.51831}{3.56224} & \P{26.79525112909323}{21.719507950053895} & \P{94.274}{77.3701} & \P{7309.69}{7327.03} \\ 
       &
\P{18510.0}{148096.0} & 15-21 & \P{2.99924}{2.99961} & \P{28.66089409317027}{23.718483402842367} & \P{85.9609}{71.1462} & \P{18145.2}{18202.6} \\ 
\hline
         &
\P{33678.0}{67620.0} & 555-643 & \P{3.99766}{3.95618} & \P{32.16206480791263}{27.80940199889793} & \P{128.573}{110.019} & \P{7423.89}{7578.13} \\ 
pla33810 &
\P{33678.0}{135240.0} & 214-389 & \P{2.42981}{2.35454} & \P{45.76983385532203}{37.28193192725543} & \P{111.212}{87.7818} & \P{12593.7}{13065.8} \\ 
         &
\P{33678.0}{270480.0} & 350-459 & \P{2.21931}{2.12964} & \P{34.79536432494784}{27.270383726827067} & \P{77.2217}{58.0761} & \P{13562.4}{13796.5} \\ 
\hline
         &
\P{85895.0}{171800.0} & 708-724 & \P{3.77097}{3.70335} & \P{39.716571598289036}{34.471222001701165} & \P{149.77}{127.659} & \P{24412.7}{24813.0} \\ 
pla85900 &
\P{85900.0}{343600.0} & 586-684 & \P{2.79994}{2.71933} & \P{30.462938491539106}{25.7809460418559} & \P{85.2944}{70.1069} & \P{28884.4}{30651.6} \\ 
         &
\P{85900.0}{687200.0} & 542-716 & \P{2.19753}{2.12504} & \P{34.167997706515955}{26.86881188118812} & \P{75.0852}{57.0973} & \P{34413.5}{35254.3} \\ 
}
\def\BTSP{
\Q{0.0395688}{0.0215097} & \Q{0.0172281}{0.0110908} & \Q{3.982885}{0.517756} & \Q{2.491096}{0.02837} & \Q{1.730519}{0.024981}  \\
\Q{0.0741346}{0.0529951} & \Q{0.0591459}{0.05218} & \Q{8.503441}{4.142922} & \Q{3.270704}{0.667202} & \Q{1.8126}{0.399798}  \\
\Q{0.132454}{0.109746} & \Q{0.173008}{0.162865} & \Q{15.929574}{9.49995} & \Q{19.154417}{10.346512} & \Q{24.176291}{18.335863}  \\
\hline
\Q{0.0830688}{0.0478546} & \Q{0.0360326}{0.0249556} & \Q{22.916979}{3.044926} & \Q{11.858}{0.114285} & \Q{8.685229}{0.085098}  \\
\Q{0.213536}{0.161026} & \Q{0.152704}{0.138231} & \Q{43.183525}{17.804699} & \Q{18.741196}{5.347511} & \Q{12.354114}{4.256952}  \\
\Q{0.313533}{0.260349} & \Q{0.393464}{0.371319} & \Q{82.610153}{46.725613} & \Q{112.335091}{58.300358} & \Q{102.824944}{69.662407}  \\
\hline
\Q{0.0449206}{0.0247488} & \Q{0.164436}{0.0468961} & \Q{30.726727}{3.515384} & \Q{20.134825}{1.336602} & \Q{23.005323}{1.395459}  \\
\Q{0.103451}{0.0678338} & \Q{0.151017}{0.113968} & \Q{41.060879}{6.57196} & \Q{17.965162}{0.760942} & \Q{14.213658}{1.003081}  \\
\Q{0.15934}{0.115505} & \Q{0.392597}{0.36995} & \Q{86.476883}{37.819153} & \Q{29.464273}{8.583704} & \Q{20.325426}{7.190765}  \\
\hline
\Q{0.138019}{0.0941339} & \Q{0.101866}{0.0692317} & \Q{37.827507}{8.05111} & \Q{19.643707}{1.092053} & \Q{15.112179}{0.894004}  \\
\Q{0.218189}{0.174468} & \Q{0.247082}{0.227885} & \Q{74.860077}{37.648521} & \Q{59.101749}{29.133675} & \Q{26.053017}{12.196156}  \\
\Q{0.341008}{0.293915} & \Q{0.642283}{0.611013} & \Q{125.569542}{74.222801} & \Q{218.205475}{115.221214} & \Q{202.566559}{152.845154}  \\
\hline
\Q{0.180889}{0.122185} & \Q{0.130319}{0.0894885} & \Q{71.222557}{15.228923} & \Q{36.852749}{1.439736} & \Q{27.913727}{1.177244}  \\
\Q{0.314427}{0.254742} & \Q{0.441308}{0.408947} & \Q{144.667221}{75.779602} & \Q{131.33075}{65.538567} & \Q{87.334496}{54.312943}  \\
\Q{0.545342}{0.482168} & \Q{1.00892}{0.962819} & \Q{235.861908}{145.111618} & \Q{422.161041}{228.38295} & \Q{557.235901}{437.847931}  \\
\hline
\Q{0.38173}{0.278496} & \Q{0.215483}{0.188335} & \Q{244.405609}{49.689354} & \Q{143.068848}{10.898605} & \Q{97.177185}{2.383076}  \\
\Q{0.431362}{0.339895} & \Q{0.531054}{0.492299} & \Q{296.713226}{62.850746} & \Q{189.91571}{50.404537} & \Q{104.324585}{18.867157}  \\
\Q{0.620161}{0.513759} & \Q{1.02062}{0.968147} & \Q{387.664429}{74.970398} & \Q{236.400772}{85.356934} & \Q{123.825859}{36.422699}  \\
\hline
\Q{1.1803}{0.912832} & \Q{0.881368}{0.804327} & \Q{1824.549438}{445.145233} & \Q{1069.449463}{132.392761} & \Q{672.013489}{32.977211}  \\
\Q{1.47488}{1.22513} & \Q{1.63362}{1.53277} & \Q{2219.099365}{579.746948} & \Q{1409.31665}{427.812286} & \Q{718.629639}{152.045441}  \\
\Q{1.87947}{1.59451} & \Q{2.81132}{2.6745} & \Q{2869.375}{711.488586} & \Q{1786.32312}{698.623474} & \Q{867.133911}{250.825699}  \\
}

\def\AUNITALL{
BIKEWHE &
\P{4195.0}{4195.0} & 2 & \P{1.0}{1.0} & \P{21.2381}{9.87591} & \P{21.2381}{9.87591} & \P{4195.0}{4195.0} \\ 
CYC1 &
\P{16783.0}{16783.0} & 2-73 & \P{1.0}{1.0} & \P{165.22}{163.055} & \P{165.22}{163.055} & \P{16783.0}{16783.0} \\ 
DBLCYC &
\P{8391.0}{8391.0} & 2 & \P{1.0}{1.0} & \P{129.892}{126.551} & \P{129.892}{126.551} & \P{8391.0}{8391.0} \\ 
NOI1 &
\P{999.0}{999.0} & 2 & \P{1.0}{1.0} & \P{17.106}{8.816906} & \P{17.106}{8.816906} & \P{999.0}{999.0} \\ 
NOI2 &
\P{999.0}{999.0} & 2 & \P{1.0}{1.0} & \P{17.106}{8.816906} & \P{17.106}{8.816906} & \P{999.0}{999.0} \\ 
NOI3 &
\P{999.0}{999.0} & 2 & \P{1.0}{1.0} & \P{17.106}{8.848578} & \P{17.106}{8.848578} & \P{999.0}{999.0} \\ 
NOI4 &
\P{999.0}{999.0} & 2 & \P{1.0}{1.0} & \P{17.106}{8.848578} & \P{17.106}{8.848578} & \P{999.0}{999.0} \\ 
NOI5 &
\P{999.0}{999.0} & 2 & \P{1.0}{1.0} & \P{17.106}{8.816906} & \P{17.106}{8.816906} & \P{999.0}{999.0} \\ 
NOI6 &
\P{999.0}{999.0} & 2 & \P{1.0}{1.0} & \P{17.106}{8.816906} & \P{17.106}{8.816906} & \P{999.0}{999.0} \\ 
PATH &
\P{1999.0}{1999.0} & 2 & \P{1.0}{1.0} & \P{19.080999999999996}{9.493960000000001} & \P{19.080999999999996}{9.493960000000001} & \P{1999.0}{1999.0} \\ 
PR1 &
\P{1999.0}{1999.0} & 2 & \P{1.0}{1.0} & \P{19.6797}{10.167972} & \P{19.6797}{10.167972} & \P{1999.0}{1999.0} \\ 
PR5 &
\P{1999.0}{1999.0} & 2 & \P{1.0}{1.0} & \P{19.6797}{10.167972} & \P{19.6797}{10.167972} & \P{1999.0}{1999.0} \\ 
PR6 &
\P{1999.0}{1999.0} & 2 & \P{1.0}{1.0} & \P{19.0845}{9.741628} & \P{19.0845}{9.741628} & \P{1999.0}{1999.0} \\ 
PR7 &
\P{1999.0}{1999.0} & 2 & \P{1.0}{1.0} & \P{19.2822}{10.033252000000001} & \P{19.2822}{10.033252000000001} & \P{1999.0}{1999.0} \\ 
PR8 &
\P{1999.0}{1999.0} & 2 & \P{1.0}{1.0} & \P{19.8753}{10.669112} & \P{19.8753}{10.669112} & \P{1999.0}{1999.0} \\ 
REG1 &
\P{999.0}{999.0} & 2 & \P{1.0}{1.0} & \P{17.693600000000004}{9.449549999999999} & \P{17.693600000000004}{9.449549999999999} & \P{999.0}{999.0} \\ 
TREE &
\P{799.0}{799.0} & 2 & \P{1.0}{1.0} & \P{16.4175}{8.460388} & \P{16.4175}{8.460388} & \P{799.0}{799.0} \\ 
WHE &
\P{4195.0}{4195.0} & 2 & \P{1.0}{1.0} & \P{21.2376}{9.62098} & \P{21.2376}{9.62098} & \P{4195.0}{4195.0} \\ 
}
\def\BUNITALL{
\Q{0.006680573999999999}{0.0029321819999999998} & \Q{0.01521458}{0.002500896} & 0.0034804000000000002 & 0.002898 \\
\Q{0.0264175}{0.0145923} & \Q{0.18981}{0.00760612} & 0.006292 & 0.006201 \\
\Q{0.023071}{0.0153404} & \Q{0.0124385}{0.00725613} & 0.933069 & 0.872467 \\
\Q{0.03598053999999999}{0.012427039999999999} & \Q{2.8011920000000003}{2.781016} & 0.2044368 & 0.077954 \\
\Q{0.035540119999999994}{0.01215232} & \Q{2.8264199999999997}{2.805504} & 0.20497139999999997 & 0.07747760000000001 \\
\Q{0.05902562}{0.01745724} & \Q{4.08036}{4.046036} & 0.49701520000000005 & 0.1794032 \\
\Q{0.061020160000000004}{0.01812552} & \Q{4.144284}{4.109056} & 0.49780860000000005 & 0.1836344 \\
\Q{0.03543986}{0.01216714} & \Q{2.82851}{2.808174} & 0.2040622 & 0.0783018 \\
\Q{0.03516022}{0.012094480000000001} & \Q{2.824118}{2.80359} & 0.2019472 & 0.078781 \\
\Q{0.007942988000000002}{0.004881628} & \Q{0.5102059999999999}{0.5051013999999999} & 0.0100034 & 0.008529 \\
\Q{0.01112624}{0.006208542} & \Q{1.069212}{1.0618024000000001} & 0.017178 & 0.013854 \\
\Q{0.012799839999999998}{0.00724041} & \Q{1.0701612}{1.0627110000000002} & 0.017155800000000002 & 0.013500999999999999 \\
\Q{0.03897346}{0.016176279999999998} & \Q{5.258972}{5.234688} & 0.12651020000000002 & 0.0591424 \\
\Q{0.335084}{0.0839281} & \Q{28.530160000000002}{28.36366} & 2.5613274000000006 & 0.8959620000000001 \\
\Q{0.8929817999999999}{0.09918372} & \Q{41.341359999999995}{40.990500000000004} & 6.8050538000000005 & 0.4921458 \\
\Q{0.02905544}{0.01098592} & \Q{1.6140599999999998}{1.5987680000000002} & 0.1040858 & 0.050642200000000005 \\
\Q{0.02196612}{0.007890616} & \Q{1.5061520000000002}{1.49339} & 0.1148008 & 0.048926800000000006 \\
\Q{0.00957687}{0.004396716} & \Q{0.023096899999999997}{0.003056076} & 0.0047388 & 0.0042398 \\
}

\def\AUNITTSP{
       &
\P{1771.61}{1488.85} & 134-509 & \P{3.30671}{3.08536} & \P{65.43029174012841}{45.66468742707496} & \P{216.359}{140.892} & \P{1771.61}{1488.85} \\ 
rl5934 &
\P{636.632}{634.994} & 85-93 & \P{3.53101}{3.46436} & \P{23.7573385518591}{20.227776559018118} & \P{83.8874}{70.0763} & \P{636.632}{634.994} \\ 
       &
\P{4413.24}{4446.35} & 19-25 & \P{2.93259}{2.92693} & \P{25.823009694502133}{19.932523155661393} & \P{75.7283}{58.3411} & \P{4413.24}{4446.35} \\ 
\hline
        &
\P{2440.42}{2110.42} & 143-587 & \P{3.53388}{3.37092} & \P{61.07904060126547}{47.0948583769416} & \P{215.846}{158.753} & \P{2440.42}{2110.42} \\ 
rl11849 &
\P{2906.67}{2835.26} & 69-118 & \P{4.0227}{3.96384} & \P{28.469435950978195}{24.19900904173731} & \P{114.524}{95.921} & \P{2906.67}{2835.26} \\ 
        &
\P{9065.62}{9119.11} & 23-30 & \P{3.02726}{3.00666} & \P{33.9799026182092}{27.55406331277896} & \P{102.866}{82.8457} & \P{9065.62}{9119.11} \\ 
\hline
         &
\P{7633.19}{4192.75} & 70-5528 & \P{2.28203}{2.52979} & \P{825.7297230974177}{78.95082200498855} & \P{1884.34}{199.729} & \P{7633.19}{4192.75} \\ 
usa13509 &
\P{4562.4}{2947.46} & 80-2486 & \P{2.9923}{3.19328} & \P{189.50472880393005}{39.94231636436516} & \P{567.055}{127.547} & \P{4562.4}{2947.46} \\ 
         &
\P{3366.36}{2709.36} & 39-884 & \P{3.12895}{3.155} & \P{57.913677112130266}{29.24944532488114} & \P{181.209}{92.282} & \P{3366.36}{2709.36} \\ 
\hline
         &
\P{3740.33}{2362.81} & 129-2132 & \P{3.33001}{3.72578} & \P{159.3250470719307}{61.35440095765182} & \P{530.554}{228.593} & \P{3740.33}{2362.81} \\ 
brd14051 &
\P{4257.59}{3933.93} & 62-234 & \P{3.46509}{3.59725} & \P{37.63048001639207}{30.313155882966154} & \P{130.393}{109.044} & \P{4257.59}{3933.93} \\ 
         &
\P{11408.7}{11401.7} & 15-25 & \P{2.86542}{2.91562} & \P{28.486120708308036}{23.053450038070807} & \P{81.6247}{67.2151} & \P{11408.7}{11401.7} \\ 
\hline
       &
\P{4869.86}{3124.46} & 171-2284 & \P{3.43653}{3.75022} & \P{153.62269498593056}{59.692498040114984} & \P{527.929}{223.86} & \P{4869.86}{3124.46} \\ 
d18512 &
\P{8044.19}{7907.21} & 54-192 & \P{3.3562}{3.47076} & \P{28.89017341040463}{23.59532206202676} & \P{96.9612}{81.8937} & \P{8044.19}{7907.21} \\ 
       &
\P{18108.5}{18155.7} & 12-19 & \P{2.98974}{2.99298} & \P{27.68498264063096}{22.49022713148768} & \P{82.7709}{67.3128} & \P{18108.5}{18155.7} \\ 
\hline
         &
\P{8720.32}{8897.5} & 513-663 & \P{3.42949}{3.40337} & \P{40.056976401739036}{33.23529325345173} & \P{137.375}{113.112} & \P{8720.32}{8897.5} \\ 
pla33810 &
\P{13961.2}{14624.5} & 214-410 & \P{2.31923}{2.24913} & \P{49.41381406760002}{40.58702698376705} & \P{114.602}{91.2855} & \P{13961.2}{14624.5} \\ 
         &
\P{16823.4}{17095.8} & 401-459 & \P{2.18299}{2.11247} & \P{37.63439136230583}{29.615189801511974} & \P{82.1555}{62.5612} & \P{16823.4}{17095.8} \\ 
\hline
         &
\P{32060.8}{32484.6} & 596-652 & \P{3.35228}{3.30736} & \P{42.139976374288544}{35.79078177156403} & \P{141.265}{118.373} & \P{32060.8}{32484.6} \\ 
pla85900 &
\P{30976.4}{32646.8} & 471-686 & \P{2.44783}{2.39128} & \P{35.31662738016937}{30.20917667525342} & \P{86.4491}{72.2386} & \P{30976.4}{32646.8} \\ 
         &
\P{34521.4}{35170.7} & 461-714 & \P{2.19527}{2.11928} & \P{37.249449953764234}{29.67904193877166} & \P{81.7726}{62.8982} & \P{34521.4}{35170.7} \\ 
}
\def\BUNITTSP{
\Q{0.0258322}{0.0136152} & \Q{0.015596}{0.0097764} & 0.009005 & 0.008836 \\
\Q{0.0554348}{0.0384047} & \Q{0.0400921}{0.0344079} & 0.035197 & 0.038269 \\
\Q{0.0906538}{0.0734127} & \Q{0.135854}{0.127841} & 2.620322 & 2.688721 \\
\hline
\Q{0.0783284}{0.0449405} & \Q{0.035649}{0.023704} & 0.026674 & 0.026861 \\
\Q{0.186912}{0.136237} & \Q{0.137835}{0.124196} & 1.31369 & 1.4113 \\
\Q{0.301017}{0.249581} & \Q{0.371959}{0.350458} & 20.732124 & 21.807738 \\
\hline
\Q{0.0524207}{0.0284387} & \Q{0.176925}{0.0521505} & 0.038268 & 0.036595 \\
\Q{0.10618}{0.0673304} & \Q{0.157509}{0.118679} & 0.073089 & 0.071499 \\
\Q{0.184817}{0.132467} & \Q{0.360457}{0.335983} & 0.67336 & 0.695362 \\
\hline
\Q{0.0915992}{0.0624141} & \Q{0.0658232}{0.0422179} & 0.044809 & 0.043116 \\
\Q{0.15976}{0.122908} & \Q{0.187466}{0.171567} & 0.462677 & 0.480098 \\
\Q{0.227506}{0.194416} & \Q{0.426376}{0.403891} & 16.312923 & 19.570724 \\
\hline
\Q{0.117512}{0.079} & \Q{0.0866971}{0.0569216} & 0.065474 & 0.067244 \\
\Q{0.228397}{0.181414} & \Q{0.312515}{0.287054} & 1.690695 & 1.8216 \\
\Q{0.339002}{0.289005} & \Q{0.785043}{0.749396} & 40.865703 & 45.151446 \\
\hline
\Q{0.249216}{0.177219} & \Q{0.145243}{0.122535} & 0.436157 & 0.451767 \\
\Q{0.288595}{0.219966} & \Q{0.327482}{0.297853} & 4.448652 & 4.694508 \\
\Q{0.37595}{0.296517} & \Q{0.709025}{0.66822} & 11.602148 & 12.444321 \\
\hline
\Q{0.769017}{0.569932} & \Q{0.52834}{0.470386} & 5.623785 & 5.768864 \\
\Q{0.952175}{0.768038} & \Q{1.097}{1.01585} & 34.689222 & 37.409061 \\
\Q{1.27082}{1.04339} & \Q{2.07587}{1.96543} & 86.069119 & 95.086323 \\
}

\def\AUNITSOC{
facebook &
\P{906.489}{937.106} & 10-11 & \P{2.36428}{2.564} & \P{17.036603109614767}{10.253393135725428} & \P{40.2793}{26.2897} & \P{906.489}{937.106} \\ 
wiki-Vote &
\P{4783.06}{4783.95} & 3-4 & \P{1.0}{1.0} & \P{21.4185}{8.11042} & \P{21.4185}{8.11042} & \P{4783.06}{4783.95} \\ 
ca-CondMat &
\P{12670.3}{13633.7} & 10-14 & \P{2.80978}{2.89049} & \P{25.168767661525102}{15.434061352919404} & \P{70.7187}{44.612} & \P{12670.3}{13633.7} \\ 
soc-Epinions1 &
\P{31540.9}{33132.9} & 8-9 & \P{2.90324}{2.98198} & \P{25.693501053994847}{15.742258499386313} & \P{74.5944}{46.9431} & \P{31540.9}{33132.9} \\ 
soc-Slashdot0811 &
\P{45431.2}{46065.9} & 7-8 & \P{2.96408}{2.99368} & \P{26.962835011200777}{17.26944763635392} & \P{79.92}{51.6992} & \P{45431.2}{46065.9} \\ 
twitter &
\P{65058.4}{67436.4} & 8 & \P{2.91421}{2.95529} & \P{28.34054512200562}{17.392675507310617} & \P{82.5903}{51.4004} & \P{65058.4}{67436.4} \\ 
com-dblp &
\P{118245.0}{132385.0} & 16-18 & \P{4.04616}{4.37513} & \P{32.33312572908634}{22.975317304857224} & \P{130.825}{100.52} & \P{118245.0}{132385.0} \\ 
}
\def\BUNITSOC{
\Q{0.0346289}{0.0190134} & \Q{0.15706}{0.145618} & 0.044846 & 0.032164 \\
\Q{0.0317201}{0.0147456} & \Q{0.623322}{0.555552} & 0.057813 & 0.037263 \\
\Q{0.1594}{0.108079} & \Q{1.41696}{1.29902} & 0.068107 & 0.06597 \\
\Q{0.31085}{0.160931} & \Q{4.08302}{3.37373} & 0.21321 & 0.160391 \\
\Q{0.478451}{0.259984} & \Q{8.09458}{6.61972} & 0.255493 & 0.219538 \\
\Q{1.61045}{0.93908} & \Q{31.3256}{26.6782} & 1.999582 & 1.808726 \\
\Q{5.25383}{4.00329} & \Q{250.152}{220.93} & 3.808826 & 3.932471 \\
}

\renewcommand\P[2]{#1 &,& #2}
\renewcommand\Q[2]{#1 &/& #2}
\newcommand\QQQQ { \multicolumn{3}{c|}{} & \multicolumn{3}{c|}{} & \multicolumn{3}{c|}{} & \multicolumn{3}{c|}{} }
\newcommand\QQQ { \multicolumn{3}{c|}{} & \multicolumn{3}{c|}{} & \multicolumn{3}{c|}{} }

\def\tableA#1#2#3#4#5#6{{
\begin{tabular}{
	!{\vrule} R{#2}{0} @{\hspace{0pt}} c @{\hspace{2pt}} R{#2}{0}
	!{\vrule} c
	!{\vrule} R{#3}{1} @{\hspace{2pt}} c @{\hspace{0pt}} R{#3}{1}
	!{\vrule} R{#4}{1} @{\hspace{2pt}} c @{\hspace{0pt}} R{#4}{1}
	!{\vrule} R{#5}{1} @{\hspace{2pt}} c @{\hspace{0pt}} R{#5}{1}
	!{\vrule} R{#6}{0} @{\hspace{2pt}} c @{\hspace{0pt}} R{#6}{0}|
	} 
 \cline{5-16} \multicolumn{4}{c|}{}    & \multicolumn{9}{c|}{{\tt OC} }   & \multicolumn{3}{c|}{{\tt GH}}  \\ 
 \hline \multicolumn{3}{ !{\vrule} c !{\vrule} }{${\tt size}(G)$}   
 	& \ensuremath{D} 
 	& \multicolumn{3}{c !{\vrule} }{$\tfrac{{\tt size}(OC)}{{\tt size}(G)}$} 
 	& \multicolumn{3}{c !{\vrule} }{$\tfrac{{\tt size}(MF)}{{\tt size}(OC)}$} 
 	& \multicolumn{3}{c !{\vrule} }{$\tfrac{{\tt size}(MF)}{{\tt size}(G)}$} 
 	& \multicolumn{3}{c !{\vrule} }{$\tfrac{{\tt size}(MF)}{{\tt size}(G)}$}  \\ \hline
 #1
\hline
\end{tabular} 
}}

\def\tableAx#1#2#3#4#5#6#7{{
\begin{tabular}{|@{\hspace{2pt}}c @{\hspace{3pt}}
	!{\vrule} @{\hspace{-6pt}} R{#2}{0} @{\hspace{0pt}} c @{\hspace{2pt}} R{#2}{0} @{\hspace{3pt}}
	!{\vrule} @{\hspace{3pt}} c @{\hspace{3pt}}
	!{\vrule} @{\hspace{0pt}} R{#3}{1} @{\hspace{0pt}} c @{\hspace{-7pt}} R{#3}{1} @{\hspace{0pt}}
	!{\vrule} R{#4}{1} @{\hspace{2pt}} c @{\hspace{0pt}} R{#4}{1}
	!{\vrule} R{#5}{1} @{\hspace{2pt}} c @{\hspace{0pt}} R{#5}{1}
	!{\vrule} R{#6}{0} @{\hspace{2pt}} c @{\hspace{0pt}} R{#7}{0}|
	} 
 \cline{6-17} \multicolumn{5}{c !{\vrule} }{}    & \multicolumn{9}{c|}{{\tt OC} }   & \multicolumn{3}{c|}{{\tt GH}}  \\ 
 \hline &\multicolumn{3}{c !{\vrule} @{\hspace{3pt}} }{${\tt size}(G)$}   
        & \ensuremath{D} 
        & \multicolumn{3}{c !{\vrule} }{$\tfrac{{\tt size}(OC)}{{\tt size}(G)}$} 
        & \multicolumn{3}{c !{\vrule} }{$\tfrac{{\tt size}(MF)}{{\tt size}(OC)}$} 
        & \multicolumn{3}{c !{\vrule} }{$\tfrac{{\tt size}(MF)}{{\tt size}(G)}$} 
        & \multicolumn{3}{c !{\vrule} }{$\tfrac{{\tt size}(MF)}{{\tt size}(G)}$}  \\ \hline
 #1
\hline
\end{tabular} 
}}

\def\tableB#1#2#3#4#5#6#7#8{{
\raisebox{-6.5pt}{
\begin{tabular}{
	!{\vrule} @{\hspace{0pt}} R{#2}{#3} @{\hspace{2pt}} c @{\hspace{0pt}} R{#2}{#3} @{\hspace{4pt}}
	!{\vrule} @{\hspace{-2pt}}  R{#4}{#5} @{\hspace{2pt}} c @{\hspace{0pt}} R{#4}{#5} @{\hspace{4pt}} 
	!{\vrule} @{\hspace{2pt}} R{#6}{1} @{\hspace{2pt}} c @{\hspace{0pt}} R{#6}{1} @{\hspace{4pt}}
	!{\vrule} @{\hspace{2pt}} R{#7}{1} @{\hspace{2pt}} c @{\hspace{0pt}} R{#7}{1} @{\hspace{4pt}} 
	!{\vrule} @{\hspace{2pt}} R{#8}{1} @{\hspace{2pt}} c @{\hspace{0pt}} R{#8}{1} @{\hspace{4pt}} !{\vrule}
	} 
 \hline
 	\multicolumn{3}{!{\vrule} c !{\vrule} @{\hspace{-2pt}} }{{\tt OC}} & 
 	\multicolumn{3}{c !{\vrule} @{\hspace{2pt}}  }{{\tt GH}} & 
 	\multicolumn{3}{c !{\vrule} @{\hspace{2pt}} }{gus}   &
 	\multicolumn{3}{c !{\vrule}  @{\hspace{2pt}} }{gh}  & 
 	\multicolumn{3}{c !{\vrule} }{ghg} 
 	\\ \hline
 #1
\hline
\end{tabular}
}
}}

\def\tableUNITB#1#2#3#4#5#6#7{{
\raisebox{-6.5pt}{
\begin{tabular}{
	!{\vrule} @{\hspace{0pt}} R{#2}{#3} @{\hspace{2pt}} c @{\hspace{0pt}} R{#2}{#3} @{\hspace{4pt}}
	!{\vrule} @{\hspace{-2pt}}  R{#4}{#5} @{\hspace{2pt}} c @{\hspace{0pt}} R{#4}{#5} @{\hspace{4pt}} 
	|R{#6}{#7} @{\hspace{2pt}}  
	|R{#6}{#7} @{\hspace{2pt}}  |  
	} 
 \hline
 	\multicolumn{3}{!{\vrule} c !{\vrule} @{\hspace{-2pt}} }{{\tt OC}} & 
 	\multicolumn{3}{c !{\vrule} @{\hspace{0pt}}  }{{\tt GH}} & 
	\multicolumn{1}{c !{\vrule} }{\tt mg+}  &
 	\multicolumn{1}{c !{\vrule} }{\tt mg-} 
 	\\ \hline
 #1
\hline
\end{tabular}
}
}}


\begin{table}\centering {\scriptsize
\resizebox{1.10\textwidth}{!}{\hspace{-60pt}
\tableAx{\AALL}{7}{2}{3}{3}{4}{5}
\hspace{-10pt}
\tableB{\BALL}{2}{2}{3}{2}{4}{3}{3}
}
} \caption{\footnotesize  Synthetic instances: summary (see tables \ref{table:BIKEWHE} - \ref{table:WHE}). } \label{table:ALL} \end{table}

\begin{table}\centering {\scriptsize
\resizebox{1.10\textwidth}{!}{\hspace{-60pt}
\tableAx{\ATSP}{7}{2}{3}{3}{5}{6}
\hspace{-10pt}
\tableB{\BTSP}{2}{2}{2}{2}{4}{4}{3}
}
} \caption{\footnotesize  TSP instances (for average degree $k=2,4,8$).  ({\tt Lemon}, first row: 2.63). } \label{table:TSP} \end{table}

\begin{table}\centering {\scriptsize
\resizebox{1.05\textwidth}{!}{\hspace{-30pt}
\tableAx{\AUNITALL}{7}{2}{3}{3}{4}{5}
\hspace{-10pt}
\tableUNITB{\BUNITALL}{2}{3}{3}{2}{1}{3}
}
} \caption{\footnotesize Unweighted simple graphs: synthetic instances. } \label{table:unit:ALL} \end{table}

\begin{table}\centering {\scriptsize
\resizebox{1.05\textwidth}{!}{\hspace{-30pt}
\tableAx{\AUNITTSP}{7}{2}{3}{3}{5}{6}
\hspace{-10pt}
\tableUNITB{\BUNITTSP}{2}{2}{3}{2}{1}{2}
}
} \caption{\footnotesize Unweighted simple graphs: TSP instances. } \label{table:unit:TSP} \end{table}

\begin{table}\centering {\scriptsize
\resizebox{1.10\textwidth}{!}{\hspace{-40pt}
\tableAx{\AUNITSOC}{8}{2}{3}{3}{5}{6}
\hspace{-10pt}
\tableUNITB{\BUNITSOC}{2}{2}{4}{2}{1}{2}
}
} \caption{\footnotesize Unweighted simple graphs: social and web graphs. } \label{table:unit:SOC} \end{table}


\addtocounter{table}{1} 
\cleardoublepage 
\edef\maintableno{\thetable} 

\renewcommand{\thetable}{\maintableno a}
\begin{table}\centering {\scriptsize
\resizebox{1.10\textwidth}{!}{\hspace{-40pt}
\tableA{\ABIKEWHE}{4}{1}{2}{2}{5}
\hspace{-10pt}
\tableB{\BBIKEWHE}{2}{3}{2}{3}{3}{3}{2}
}
} \caption{\footnotesize BIKEWHE: bikewheelgen 1024, ..., 4196. ({\tt Lemon}, first row: 155). } \label{table:BIKEWHE} \end{table}


\renewcommand{\thetable}{\maintableno b}
\begin{table}\centering {\scriptsize
\resizebox{1.10\textwidth}{!}{\hspace{-40pt}
\tableA{\ACYCONE}{5}{2}{2}{3}{2}
\hspace{-10pt}
\tableB{\BCYCONE}{2}{3}{2}{3}{2}{2}{2}
}
} \caption{\footnotesize CYC1: cyclegen 4196, ..., 16784. ({\tt Lemon}, first row: 10.0). } \label{table:CYC1} \end{table}


\renewcommand{\thetable}{\maintableno c}
\begin{table}\centering {\scriptsize
\resizebox{1.10\textwidth}{!}{\hspace{-40pt}
\tableA{\ADBLCYC}{5}{2}{2}{2}{5}
\hspace{-10pt}
\tableB{\BDBLCYC}{2}{3}{2}{3}{4}{2}{2}
}
} \caption{\footnotesize DBLCYC1: dblcyclegen 2048, ..., 8392. ({\tt Lemon}, first row: 2642). } \label{table:DBLCYC1} \end{table}


\renewcommand{\thetable}{\maintableno d}
\begin{table}\centering {\scriptsize
\resizebox{1.10\textwidth}{!}{\hspace{-40pt}
\tableA{\ANOIONE}{6}{2}{2}{2}{4}
\hspace{-10pt}
\tableB{\BNOIONE}{2}{3}{2}{2}{2}{2}{2}
}
} \caption{\footnotesize NOI1: noigen 400 50 1 400, ..., 1000 50 1 1000. ({\tt Lemon}, first row: 1.42). } \label{table:NOI1} \end{table}


\renewcommand{\thetable}{\maintableno e}
\begin{table}\centering {\scriptsize
\resizebox{1.10\textwidth}{!}{\hspace{-40pt}
\tableA{\ANOITWO}{6}{2}{2}{2}{4}
\hspace{-10pt}
\tableB{\BNOITWO}{2}{3}{2}{2}{2}{2}{2}
}
} \caption{\footnotesize NOI2: noigen 400 50 2 400, ..., 1000 50 2 1000.  ({\tt Lemon}, first row: 1.51). } \label{table:NOI2} \end{table}


\renewcommand{\thetable}{\maintableno f}
\begin{table}\centering {\scriptsize
\resizebox{1.10\textwidth}{!}{\hspace{-40pt}
\tableA{\ANOITHREE}{6}{2}{2}{2}{4}
\hspace{-10pt}
\tableB{\BNOITHREE}{2}{2}{2}{1}{2}{2}{2}
}
} \caption{\footnotesize NOI3: noigen 1000 d 1 1000, d=5,10,25,50,75,100.  ({\tt Lemon}, first row: 1.54). } \label{table:NOI3} \end{table}


\renewcommand{\thetable}{\maintableno g}
\begin{table}\centering {\scriptsize
\resizebox{1.10\textwidth}{!}{\hspace{-40pt}
\tableA{\ANOIFOUR}{6}{2}{2}{2}{4}
\hspace{-10pt}
\tableB{\BNOIFOUR}{2}{2}{2}{1}{2}{2}{2}
}
} \caption{\footnotesize NOI4: noigen 1000 d 2 1000, d=5,10,25,50,75,100.  ({\tt Lemon}, first row: 2.41). } \label{table:NOI4} \end{table}


\renewcommand{\thetable}{\maintableno h}
\begin{table}\centering {\scriptsize
\resizebox{1.10\textwidth}{!}{\hspace{-40pt}
\tableA{\ANOIFIVE}{6}{2}{2}{2}{4}
\hspace{-10pt}
\tableB{\BNOIFIVE}{2}{2}{2}{2}{2}{2}{2}
}
} \caption{\footnotesize NOI5: noigen 1000 50 k 1000, k=1,2,5,10,30,50,100,200,400,500.  ({\tt Lemon}, first row: 28.6). } \label{table:NOI5} \end{table}


\renewcommand{\thetable}{\maintableno i}
\begin{table}\centering {\scriptsize
\resizebox{1.10\textwidth}{!}{\hspace{-40pt}
\tableA{\ANOISIX}{6}{2}{2}{2}{4}
\hspace{-10pt}
\tableB{\BNOISIX}{2}{2}{2}{1}{2}{2}{2}
}
} \caption{\footnotesize NOI6: noigen 1000 50 2 P, P=1,10,50,100,150,200,500,1000,5000.  ({\tt Lemon}, first row: 28.4). } \label{table:NOI6} \end{table}


\renewcommand{\thetable}{\maintableno j}
\begin{table}\centering {\scriptsize
\resizebox{1.10\textwidth}{!}{\hspace{-40pt}
\tableA{\APATH}{5}{2}{2}{2}{4}
\hspace{-10pt}
\tableB{\BPATH}{2}{3}{2}{3}{2}{2}{2}
}
} \caption{\footnotesize PATH: pathgen 2000 1 k 1000, k=1,4,15,50,200,800,2000. ({\tt Lemon}, first row: 2.64). } \label{table:PATH} \end{table}


\renewcommand{\thetable}{\maintableno k}
\begin{table}\centering {\scriptsize
\resizebox{1.10\textwidth}{!}{\hspace{-40pt}
\tableA{\APRONE}{5}{2}{2}{2}{5}
\hspace{-10pt}
\tableB{\BPRONE}{2}{3}{2}{2}{2}{2}{2}
}
} \caption{\footnotesize PR1: prgen 400 2 1, ..., 2000 2 1.  ({\tt Lemon}, first row: 0.071). } \label{table:PR1} \end{table}


\renewcommand{\thetable}{\maintableno l}
\begin{table}\centering {\scriptsize
\resizebox{1.10\textwidth}{!}{\hspace{-40pt}
\tableA{\APRFIVE}{5}{2}{2}{2}{4}
\hspace{-10pt}
\tableB{\BPRFIVE}{2}{3}{2}{2}{2}{2}{2}
}
} \caption{\footnotesize PR5: prgen 400 2 2, ..., 2000 2 2.  ({\tt Lemon}, first row: 0.092). } \label{table:PR5} \end{table}


\renewcommand{\thetable}{\maintableno m}
\begin{table}\centering {\scriptsize
\resizebox{1.10\textwidth}{!}{\hspace{-40pt}
\tableA{\APRSIX}{6}{2}{2}{2}{4}
\hspace{-10pt}
\tableB{\BPRSIX}{2}{3}{2}{2}{2}{2}{2}
}
} \caption{ PR6: prgen 400 10 2, ..., 2000 10 2. ({\tt Lemon}, first row: 0.31). } \label{table:PR6} \end{table}


\renewcommand{\thetable}{\maintableno n}
\begin{table}\centering {\scriptsize
\resizebox{1.10\textwidth}{!}{\hspace{-40pt}
\tableA{\APRSEVEN}{7}{2}{2}{2}{4}
\hspace{-10pt}
\tableB{\BPRSEVEN}{2}{2}{2}{1}{2}{2}{2}
}
} \caption{\footnotesize PR7: prgen 400 50 2, ..., 2000 50 2.  ({\tt Lemon}, first row: 1.98). } \label{table:PR7} \end{table}


\renewcommand{\thetable}{\maintableno o}
\begin{table}\centering {\scriptsize
\resizebox{1.10\textwidth}{!}{\hspace{-40pt}
\tableA{\APREIGHT}{7}{2}{2}{2}{4}
\hspace{-10pt}
\tableB{\BPREIGHT}{2}{2}{3}{2}{3}{4}{4}
}
} \caption{\footnotesize PR8: prgen 400 100 2, ..., 2000 100 2.  ({\tt Lemon}, first row: 5.5). } \label{table:NOI8} \end{table}


\renewcommand{\thetable}{\maintableno p}
\begin{table}\centering {\scriptsize
\resizebox{1.10\textwidth}{!}{\hspace{-40pt}
\tableA{\AREGONE}{6}{2}{2}{2}{3}
\hspace{-10pt}
\tableB{\BREGONE}{2}{3}{2}{2}{2}{2}{2}
}
} \caption{\footnotesize REG1: regulargen 1000 k -w1000, k=1,4,16,64.  ({\tt Lemon}, first row: 0.38). } \label{table:REG1} \end{table}


\renewcommand{\thetable}{\maintableno q}
\begin{table}\centering {\scriptsize
\resizebox{1.10\textwidth}{!}{\hspace{-40pt}
\tableA{\ATREE}{6}{2}{2}{2}{3}
\hspace{-10pt}
\tableB{\BTREE}{2}{3}{2}{2}{2}{2}{2}
}
} \caption{\footnotesize TREE: treegen 800 50 k 1000, k=1,3,5,10,20,50,100,200,400. ({\tt Lemon}, first row: 11.5). } \label{table:TREE} \end{table}


\renewcommand{\thetable}{\maintableno r}
\begin{table}\centering {\scriptsize
\resizebox{1.10\textwidth}{!}{\hspace{-40pt}
\tableA{\AWHE}{4}{2}{2}{2}{5}
\hspace{-10pt}
\tableB{\BWHE}{2}{3}{2}{2}{3}{3}{2}
}
} \caption{\footnotesize WHE: wheelgen 1024, ..., 4196.  ({\tt Lemon}, first row: 162). } \label{table:WHE} \end{table}

\renewcommand{\thetable}{\arabic{table}} 

\clearpage
\newpage

\bibliographystyle{alpha}
\bibliography{maxflow}

\newcommand{\etalchar}[1]{$^{#1}$}
\begin{thebibliography}{GHK{\etalchar{+}}15}

\bibitem[AIS{\etalchar{+}}16]{MassiveGraphs}
Takuya Akiba, Yoichi Iwata, Yosuke Sameshima, Naoto Mizuno, and Yosuke Yano.
\newblock Cut tree construction from massive graphs.
\newblock In {\em Proceedings of the IEEE 16th International Conference on Data
  Mining (ICDM)}, pages 775--780, 2016.

\bibitem[AKL{\etalchar{+}}22]{GH:subcubic}
Amir Abboud, Robert Krauthgamer, Jason Li, Debmalya Panigrahi, Thatchaphol
  Saranurak, and Ohad Trabelsi.
\newblock Breaking the cubic barrier for all-pairs max-flow: {G}omory-{H}u tree
  in nearly quadratic time.
\newblock In {\em FOCS}, 2022.

\bibitem[AKL{\etalchar{+}}25]{GHdeterminstic:FOCS25}
Amir Abboud, Rasmus Kyng, Jason Li, Debmalya Panigrahi, Maximilian Probst,
  Thatchaphol Saranurak, Wuwei Yuan, and Weixuan Yuan.
\newblock Deterministic almost-linear-time {G}omory-{H}u {T}rees.
\newblock In {\em FOCS}, 2025.

\bibitem[AKT20]{Abboud:FOCS20}
Amir Abboud, Robert Krauthgamer, and Ohad Trabelsi.
\newblock Cut-equivalent trees are optimal for min-cut queries.
\newblock In {\em FOCS}, 2020.
\newblock Full version: arXiv:2009.06090.

\bibitem[AKT21a]{Abboud:FOCS21}
Amir Abboud, Robert Krauthgamer, and Ohad Trabelsi.
\newblock {APMF} $<$ {APSP}? {Gomory}-{Hu} tree for unweighted graphs in
  almost-quadratic time.
\newblock In {\em FOCS}, pages 1135--1146, 2021.

\bibitem[AKT21b]{Abboud:STOC21}
Amir Abboud, Robert Krauthgamer, and Ohad Trabelsi.
\newblock Subcubic algorithms for {G}omory-{H}u tree in unweighted graphs.
\newblock In {\em STOC}, pages 1725--1737, 2021.
\newblock Full version at arXiv:2012.10281.

\bibitem[AKT22]{Abboud:SODA22}
Amir Abboud, Robert Krauthgamer, and Ohad Trabelsi.
\newblock Friendly cut sparsifiers and faster {G}omory-{H}u trees.
\newblock In {\em SODA}, pages 3630--3649, 2022.

\bibitem[ALPS23]{GH:linear}
Amir Abboud, Jason Li, Debmalya Panigrahi, and Thatchaphol Saranurak.
\newblock All-pairs max-flow is no harder than single-pair max-flow:
  {G}omory-{H}u trees in almost-linear time.
\newblock In {\em FOCS}, pages 2204--2212, 2023.

\bibitem[BHKP07]{Bhalgat:07}
A.~Bhalgat, R.~Hariharan, T.~Kavitha, and D.~Panigrahi.
\newblock An o(mn) {G}omory-{H}u tree construction algorithm for unweighted
  graphs.
\newblock In {\em STOC}, pages 605--614, 2007.

\bibitem[CRJ12]{GH:parallel:12}
Jaime Cohen, Luiz~A. Rodrigues, and Elias P.~Duarte Jr.
\newblock A parallel implementation of {G}omory-{H}u's cut tree algorithm.
\newblock In {\em IEEE 24th International Symposium on Computer Architecture
  and High Performance Computing}, 2012.

\bibitem[CRS{\etalchar{+}}11]{GH:parallel:11}
Jaime Cohen, Luiz~A. Rodrigues, Fabiano Silva, Renato Carmo, André L.~P.
  Guedes, and Elias P.~Duarte Jr.
\newblock Parallel implementations of {G}usfield’s cut tree algorithm.
\newblock In {\em Algorithms and Architectures for Parallel Processing}, pages
  258--269, 2011.

\bibitem[DJK11]{Lemon}
B.~Dezs, A.~J\"uttner, and P.~Kov\'acs.
\newblock Lemon - an open source {C}++ graph template library.
\newblock {\em Electron. Notes Theor. Comput. Sci.}, 264(5):23--45, 2011.

\bibitem[Gab91]{Gabow:FOCS91}
H.~N. Gabow.
\newblock Applications of a poset representation to edge connectivity and graph
  rigidity.
\newblock In {\em FOCS}, pages 812--821, 1991.

\bibitem[GGT89]{Gallo:SICOMP89}
G.~Gallo, M.~D. Grigoriadis, and R.~E. Tarjan.
\newblock A fast parametric maximum flow algorithm and applications.
\newblock {\em SIAM J. Computing}, 18:30--55, 1989.

\bibitem[GH61]{GomoryHu:61}
R.~E. Gomory and T.~C. Hu.
\newblock Multi-terminal network flows.
\newblock {\em Journal of the Society for Industrial and Applied Mathematics},
  9:551--570, 1961.

\bibitem[GHK{\etalchar{+}}11]{IBFS:11}
Andrew~V. Goldberg, Sagi Hed, Haim Kaplan, Robert~E. Tarjan, and Renato~F.
  Werneck.
\newblock Maximum flows by incremental breadth-first search.
\newblock In {\em Algorithms - ESA}, pages 457--468, 2011.

\bibitem[GHK{\etalchar{+}}15]{IBFS:15}
Andrew~V. Goldberg, Sagi Hed, Haim Kaplan, Pushmeet Kohli, Robert~E. Tarjan,
  and Renato~F. Werneck.
\newblock Faster and more dynamic maximum flow by incremental breadth-first
  search.
\newblock In {\em Algorithms - ESA}, pages 619--630, 2015.

\bibitem[GT01]{Goldberg:01}
A.~V. Goldberg and K.~Tsioutsiouliklis.
\newblock Cut tree algorithms: an experimental study.
\newblock {\em Journal of Algorithms}, 38(1):51--83, 2001.
\newblock Code:
  \url{https://web.archive.org/web/20140115171857/http://www.cs.princeton.edu/~kt/cut-tree/experiments/}.

\bibitem[Gus90]{Gusfield:90}
D.~Gusfield.
\newblock Very simple methods for all pairs network flow analysis.
\newblock {\em SIAM Journal on Computing}, 19(1):143--155, 1990.

\bibitem[HO94]{HaoOrlin:94}
J.~Hao and J.~B. Orlin.
\newblock A faster algorithm for ﬁnding the minimum cut in a directed graph.
\newblock {\em Algorithmica}, 17:424--446, 1994.

\bibitem[LK14]{snapnets}
Jure Leskovec and Andrej Krevl.
\newblock {SNAP Datasets}: {Stanford} large network dataset collection.
\newblock \url{http://snap.stanford.edu/data}, June 2014.

\bibitem[LNPS22]{FairCuts}
Jason Li, Danupon Nanongkai, Debmalya Panigrahi, and Thatchaphol Saranurak.
\newblock Fair cuts, approximate isolating cuts, and approximate {G}omory-{H}u
  trees in near-linear time.
\newblock {\em arXiv:2203.00751}, March 2022.

\bibitem[LP20]{Li:FOCS20}
Jason Li and Debmalya Panigrahi.
\newblock Deterministic min-cut in poly-logarithmic maxflows.
\newblock In {\em FOCS}, 2020.

\bibitem[LP21]{LiPanigrahi:STOC21}
Jason Li and Debmalya Panigrahi.
\newblock Approximate {G}omory-{H}u tree is faster than $n-1$ max-flows.
\newblock In {\em STOC}, pages 1738--1748, 2021.

\bibitem[LP24]{LiPanigrahi:SICOMP24}
Jason Li and Debmalya Panigrahi.
\newblock Approximate {G}omory-{H}u tree is faster than $n-1$ max-flows.
\newblock {\em SIAM Journal on Computing}, 53(4):1162--1180, 2024.

\bibitem[LPS21]{Li:FOCS21}
Jason Li, Debmalya Panigrahi, and Thatchaphol Saranurak.
\newblock A nearly optimal all-pairs min-cuts algorithm in simple graphs.
\newblock In {\em FOCS}, pages 1135--1146, 2021.

\bibitem[MCJ17]{GH:parallel:17}
Charles Maske, Jaime Cohen, and Elias P.~Duarte Jr.
\newblock Parallel cut tree algorithms.
\newblock {\em J. Parallel Distrib. Comput.}, 109:1--14, 2017.

\bibitem[MCJ20]{GH:parallel:20}
Charles Maske, Jaime Cohen, and Elias P.~Duarte Jr.
\newblock Speeding up the {G}omory-{H}u parallel cut tree algorithm with
  efficient graph contractions.
\newblock {\em Algorithmica}, 82:1601--1615, 2020.

\bibitem[MVV87]{MVV:87}
K.~Mulmuley, U.~V. Vazirani, and V.~V. Vazirani.
\newblock Matching is as easy as matrix inversion.
\newblock {\em Combinatorica}, 7(1):105--113, 1987.

\bibitem[TSP]{TSPLIB}
{TSPLIB} dataset.
\newblock \url{http://elib.zib.de/pub/mp-testdata/tsp/tsplib/tsplib.html}.

\end{thebibliography}

\appendix


\section{Additional related work}\label{sec:related}
In this section we describe some existing results that reduce the problem of computing the GH tree to other problems.

Abboud, Krauthgamer and Trabelsi considered in~\cite{Abboud:FOCS20}
the ``Single Source Min Cut problem'' (SSMC) whose goal is to compute minimum $s$-$v$ cuts $C_{sv}$
for a fixed source node $s$ and all other nodes $v$. More precisely, Abboud et al.\ considered two variants;
one of them  asks for the set of edges of $C_{sv}$,
and the other one asks only for the cost of $C_{sv}$ which we denote as $f(s,v)={\tt cost}(C_{sv})$.
These problems can be formulated as the follows.

\begin{algorithm}[H]
  \DontPrintSemicolon
  for each $v\in V-\{s\}$ compute minimum $s$-$v$ cut $C_{sv}$ \\
  build data structure $\calD$ that supports operation ${\tt Query}(v)$ for given $v\in V-\{s\}$,  return $\calD$
\SetNoFillComment
  \tcc{${\tt Query}(v)$ should return $\delta C_{sv}$ in time $\tilde O(|\delta C_{sv}|)$, where $\delta C_{sv}$ is the set of edges of the cut $(C_{sv},V-C_{sv})$}  

      \caption{${\tt SingleSourceMinCut}(s;G)$. 
      }
\end{algorithm}

\begin{algorithm}[H]
  \DontPrintSemicolon
  for each $v\in V-\{s\}$ compute $\lambda(v)=f(s,v)$, 
  return $\lambda$

      \caption{${\tt SingleSourceMaxFlow}(s;G)$. 
      }
\end{algorithm}

It was shown in \cite{Abboud:FOCS20}  that the GH tree can be constructed by $\tilde O(1)$ calls to ${\tt SingleSourceMinCut}(\cdot)$
on graphs with $O(n)$ nodes and $O(m)$ edges plus $\tilde O(m)$ additional work. 
For the second procedure \cite{Abboud:FOCS20} was able to show a weaker result, namely how to construct a {\em flow-equivalent tree} for $G$
by $\tilde O(1)$ calls to ${\tt SingleSourceMaxFlow}(\cdot)$
on graphs with $O(n)$ nodes and $O(m)$ edges plus $\tilde O(n)$ additional work.\footnote{
A flow-equivalent tree for $G$ is a spanning tree $\calT$ on $V$ such that for any $s,t\in V$, the costs of the minimum $s$-$t$ cut in $G$ and in $\calT$ coincide.
Every GH tree is a also flow-equivalent tree, but the reverse is not true.}
Li, Panigrahi and Saranurak observed in~\cite{Li:FOCS21}
that instead of computing values $f(s,v)$ it suffices
to verify whether $f(s,v)$ equals a given upper bound $\mu(v)$. More precisely, they considered the following problem; value $\gamma$ in the precondition is a positive constant.

\begin{algorithm}[H]
  \DontPrintSemicolon
  compute $Y=\{v\in X\::\: \mu(v)=f(s,v)\}$,  return $Y$
      \caption{${\tt SingleSourceMaxflowVerification}(s,X,\mu;G)$. \\
      {\em\small \hspace{5pt}$\slash\slash$~ 
      preconditions: $X\subseteq V-\{s\}$,~~\! $\max\limits_{v\in X} f(s,v)\le (1+\gamma) \min\limits_{v\in X} f(s,v)$,~~\! $\mu(v)\ge f(s,v)$ for all 
      $v\in X$\!\!\!\!\!\!\!\!\!\!\!\!\!\!\!\!}
      }

\end{algorithm}

\begin{sloppypar}
It was shown in~\cite{Li:FOCS21} that the GH tree can be constructed via $\tilde O(1)$ calls to ${\tt SingleSourceMaxflowVerification}$
and to the maxflow algorithm on graphs of size $(n,m)$; if $\gamma=+\infty$ then the answer is correct with probability 1,
otherwise it is correct w.h.p.. Abboud et al. designed in~\cite{GH:subcubic} and~\cite{GH:linear}
Monte-Carlo algorithms for ${\tt SingleSourceMaxflowVerification}$ with $\gamma=0.1$
whose complexities are $\tilde O(n^2)$ and $O(m^{1+o(1)})$ respectively,
and thus obtained the improvement over $O(n)$ maxflow computations mentioned earlier.
\end{sloppypar}

We remark that ${\tt OrderedCuts}$  and  ${\tt SingleSourceMaxflowVerification}$ 
can be viewed as non-standard relaxations of the ${\tt SingleSourceMaxflow}$ problem.
The former two problems appear incomparable.

Many of the results mentioned earlier rely on the following procedure.

\begin{algorithm}[H]
  \DontPrintSemicolon
  for each $v\in Y$ compute minimum $(Y\cup\{s\}-\{v\})$-$v$ cut $S_v$,
  return these cuts
      \caption{${\tt IsolatingCuts}(s,Y;G)$. {\em\small \hspace{20pt}$\slash\ast$~ $s\in V_G$, $Y\subseteq V_G-\{s\}$  ~$\ast\slash$\!\!\!}
      }

\end{algorithm}
As shown in~\cite{Li:FOCS20}
(and independently in~\cite{Abboud:STOC21}), it can be efficiently implemented.

\begin{theorem}[{\cite{Li:FOCS20}}]\label{th:IsolatingCuts}
There exists an implementation of ${\tt IsolatingCuts}(s,\{v_1,\ldots,v_\ell\};G)$ that runs in time $O(\log \ell)\cdot t_{\tt MC}(n,m)$
and outputs a set of disjoint cuts $\{S_{v_i}\::\:i\in[\ell]\}$.~\footnote{
Ignoring base cases, the algorithm can be described as follows: 
\begin{itemize}[topsep=0pt, partopsep=0pt, itemsep=-3pt]
\item[(1)] compute minimum $\{s,v_1,\ldots,v_k\}$-$\{v_{k+1},\ldots,v_\ell\}$ cut $(S,T)$ in $G$ where $k=\lfloor \ell/2 \rfloor$; 
\item[(2)] make recursive calls ${\tt IsolatingCuts}(s'',\{v_1,\ldots,v_k\};G'')$ and ${\tt IsolatingCuts}(s',\{v_{k+1},\ldots,v_\ell\};G')$ 
where $G'$ and $G''$ are the graphs obtained from $G$ by contracting respectively sets $S$ and $T$ to a single nodes $s'$ and $s''$; 
\item[(3)] combine the results. 
\end{itemize}
}
\end{theorem}
We will need this result later on.

\section{Computing Gomory-Hu tree via depth-1 {\tt OC} trees}\label{sec:depth1} 

In this section we introduce the notion of a {\em depth-1 {\tt OC} tree}, or ${\tt OC}_1$ tree.
We then show how we can use ${\tt OC}_1$ tree to compute the GH tree.
\begin{definition}\label{def:OC1}
Consider  sequence $\varphi=s\ldots$ in graph $G$.
A {\em ${\tt OC}_1$ tree for $(\varphi,G)$}  is a pair $(\Omega,\calE)$ such that there exists subsequence $\psi=s\ldots$ of $\varphi$ with the
following properties.
\begin{itemize}
\item $(\psi,\calE)$ is a star graph with the root $s$, i.e.\ $\calE=\{vs\::\:v\in\psi-\{s\}\}$.
\item $(\Omega,\calE)$ is an {\tt OC} tree for $(\psi,G)$.
\item For every $v\in \varphi-\{s\}$ there exists $u\in\psi-\{s\}$ such that $u\sqsubseteq v$ and $v\in[u]$.
\end{itemize}
\end{definition}

The following result shows the existence of an ${\tt OC}_1$ tree.
\begin{lemma}\label{lemma:OC1}
Let $(\Omega,\calE)$ be an {\tt OC} tree for $(\varphi,G)$.
Define $\psi=\varphi$, and let us modify tuple $(\psi,\Omega,\calE)$ using the following algorithm:
while $(\psi,\calE)$ has a leaf node of depth two or larger, pick the rightmost such node $u\in\psi$
(w.r.t.\ $\sqsubseteq$)
and replace $(\psi,\Omega,\calE)$ with $(\psi^{-u},\Omega^{-u},\calE^{-u})$.
Then the resulting pair $(\Omega,\calE)$ is an ${\tt OC}_1$ tree for $(\varphi,G)$.
\end{lemma}
\begin{proof}
We will use Lemma~\ref{lemma:free-leaf}. In the light of this lemma, it suffices to show the following:
if $(\Omega,\calE)$ is an {\tt OC} tree for sequence $\varphi$ and $t$ is the rightmost
leaf of $(\varphi,\calE)$ that has depth two or larger then $t$ is a free leaf in $(\varphi,\calE)$.
Suppose not, then $\varphi=\ldots t \ldots u\ldots $ where $t,u$ have the same parent. Thus,
$u$ has depth two or larger. Let $v$ be a leaf of the subtree of $(\varphi,\calE)$ rooted at $u$,
then $v$ also has depth two or larger. Furthermore, either $v=u$ or $\varphi=\ldots t \ldots u \ldots v \ldots$.
This contradicts the choice of $t$.
\end{proof}

The lemma means that an ${\tt OC}_1$ tree for $(\varphi,G)$ can be trivially constructed if we have ${\tt OC}$ tree for $(\varphi,G)$.
However, an ${\tt OC}_1$ tree can potentially be constructed faster than an {\tt OC} tree,
since we do not need to recurse on subproblems corresponding to depths larger than 1.

Let ${\tt OrderedCuts}_1(\varphi;G)$ be a procedure that outputs an ${\tt OC}_1$ tree for $(\varphi,G)$.
The main result of this section is the following theorem.

\begin{theorem}\label{th:main:OC1}
There exists a randomized (Las-Vegas) algorithm for computing GH tree for graph~$G$ with
expected complexity $\tilde O(1)\cdot (t_{{\tt OC}_1}(n,m))$,
where $t_{{\tt OC}_1}(n,m)$ is the complexity of ${\tt OrderedCuts}_1(\cdot)$ on
a graph with $O(n)$ nodes and $O(m)$ edges.
\end{theorem}

It will be convenient to define a {\em named subpartition of set $V$}
as a pair $(\calS,\Pi)$ where $\calS\subseteq V$ and $\Pi$ is a set of disjoint subsets of $V$ such that (i) $\calS\subseteq \langle\Pi\rangle$
and (ii) for each $v\in\calS$ there exists unique $U\in\Pi$ with $v\in U$.
With some abuse of terminology we will denote a named subpartition with a single letter $\calS$ (treating it as a subset of $V$ when needed),
and will denote set $U\in\Pi$ containing $v$ as $\calS_v$.
We denote $\langle\calS\rangle=\bigcup_{v\in\calS} \calS_v$.

Note that an ${\tt OC}_1$ tree $(\Omega,\calE)$ can be uniquely represented by a named subpartition $\calS$
with $\calS=\{v\::vs\in\calE\}$ and $\calS_v\in\Omega$ for $v\in\calS$.
We can thus assume that procedure ${\tt OrderedCuts}_1(\varphi;G)$ for sequence $\varphi=s\ldots$
returns named subpartition $\calS$ of $V_G-\{s\}$, with $\calS\subseteq\varphi-\{s\}$.
In the new notation Definition~\ref{def:OC1} can be reformulated as follows.

\begin{definition}\label{def:OC1'}
Named subpartition $\calS$ of $V_G-\{s\}$ is an ${\tt OC}_1$ tree for sequence $\varphi=sv_1\ldots v_\ell$ and graph~$G$ if \\
{\rm (i)} $\calS\subseteq \{v_1,\ldots,v_\ell\}$. \\
{\rm (ii)} For every $v_k\in\calS$ set $\calS_{v_k}$ is a minimum $s((v_1\ldots v_{k-1})\cap\calS)$-$v_k$ cut in $G$. \\
{\rm (iii)} For every $k\in[\ell]$ there exists $v_i\in\{v_1,\ldots,v_k\}\cap \calS$ with $v_k\in\calS_{v_i}$.
\end{definition}

\begin{remark}
In the algorithm below instead of condition {\rm (ii)} we will only need a slightly weaker property: 
condition {\rm (ii)} needs to hold only for $k\in[\ell]$ such that $\{v_1,\ldots,v_{k-1}\}\cap C_{sv_k}=\varnothing$
where $C_{sv_k}$ is a minimum $s$-$v_k$ cut.
\end{remark}

\subsection{Computing cuts for a fixed source}\label{sec:FixedSourcePartition1}
Let us consider graph $G=(V,E,w)$, node $s\in V$ and subset $X\subseteq V-\{s\}$.
We first show how we can use procedure  ${\tt OrderedCuts}_1$ to compute minimum $s$-$v$ cuts for many nodes $v\in X$
under the following condition.
\begin{assumption}\label{as}
For every $v\in X$ the minimum $s$-$v$ cut in $G$ is unique. 
\end{assumption}
In particular, we will present an algorithm with the following properties.
\begin{theorem}\label{th:FixedSourcePartition1}
There exists randomized algorithm ${\tt FixedSourcePartition}_1(s,X;G)$ that outputs
a named subpartition $\calS$ of $V_G-\{s\}$ with $\calS\subseteq X$ satisfying the following: \\
{\rm (a)} $\calS_v$ is a minimum $s$-$v$ cut in $G$ for every $v\in\calS$. \\
{\rm (b)} If Assumption~\ref{as} holds  then $X\subseteq \langle\calS\rangle$ with probability at least  $1-1/poly(|X|)$ (for an arbitrary fixed polynomial). \\
It makes $O(\log^3|X|)$ calls to procedure ${\tt OrderedCuts}_1(\cdot)$  for graph~$G$ and performs $ O(|X|\log^4|X|)$ additional work.
\end{theorem}

The algorithm is given below. 
${\tt RandomSubset}(X;\alpha)$ in line 3  is a random subset of $X$ in which each element $v\in X$ is
included independently with prob.\ $\alpha$.

\begin{algorithm}[H]
\setcounter{AlgoLine}{0}

  \DontPrintSemicolon
    set $\lambda(v)={\tt cost}(\{v\})$ for all $v\in X$ \\
	\For{$i=1,\ldots,N+1$}
  	{
  		sample $Y\leftarrow{\tt RandomSubset}(X;\alpha_i)$ \hspace{105pt} \tcc{\small use $\alpha_{N+1}=1$} 
  		sort nodes in $Y$ as $v_1,\ldots,v_\ell$ so that $\lambda(v_1)\ge \ldots \ge \lambda(v_\ell)$, $\ell=|Y|$ \\
  		call $\calS\leftarrow{\tt OrderedCuts}_1(sv_1,\ldots v_\ell;G)$ \hspace{105pt} \tcc{\small $\calS\subseteq Y\subseteq X$}  
		\If{$i\le N$}
		{
	  		for each $v\in \calS$ update $X:=X-(\calS_v-\{v\})$ \\
	  		for each $v\in \calS$ and $u\in\calS_v$ update $\lambda(u):=\min\{\lambda(u),{\tt cost}(\calS_v)\}$ 
	  	}
	  	\Else
	  	{
			let $\calS^\ast=\{v_k\in \calS\::\: {\tt cost}(\calS_{v_k})\le {\tt cost}(\calS_v)\;\;\forall v\in\{v_1,\ldots,v_{k}\}\cap \calS\}$ \\
			return named subpartition $\calS^\ast$ where $\calS^\ast_v=\calS_v$ for $v\in\calS^\ast$
	  	}
  	}

      \caption{${\tt FixedSourcePartition}_1(s,X;G)$. 
      }\label{alg:SSPQ1}
\end{algorithm}

Note that the algorithm depends on the sequence $\alpha_1,\ldots,\alpha_{N}$.
We set it by taking sequence $(2^0,2^{-1},\ldots,2^{-d})$ for $d=\lfloor \log_2 |X| \rfloor$
and repeating it $\Theta(\log^2 |X|)$ times with an appropriate constant
(so that $N=\Theta(\log^3|X|)$).
The idea of this construction comes from~\cite{LiPanigrahi:STOC21}.
In the remainder of this section we prove that with this choice the complexity
of Algorithm~\ref{alg:SSPQ1} is as stated in Theorem~\ref{th:FixedSourcePartition1}.

Note that each update at line 7 preserves the property $\calS\subseteq X$,
since the sets in $\{\calS_v\::v\in \calS\}$ are disjoint. This means, in particular,
that the update $X:=X-(\calS_v-\{v\})$ is equivalent to the update $X:=(X-\calS_v)\cup\{v\}$.

It can be seen that the output of  Algorithm~\ref{alg:SSPQ1} satisfies property {\rm (a)} of Theorem~\ref{th:FixedSourcePartition1};
this follows from the definition of ${\tt OrderedCuts}_1$ and from Lemma~\ref{lemma:Goldberg}.
We will show next that there exists  sequence $\alpha_1,\ldots,\alpha_N$ with $N=\Theta(\log^3 |X|)$ 
so that property {\rm (b)} of Theorem~\ref{th:FixedSourcePartition1} is satisfied.
In the remainder of the analysis we assume that Assumption~\ref{as} holds.

We let $\lambda_i,X_i,Y_i$ for $i\in[N+1]$ to be the  variables at the end of the $i$-th iteration,
and $\lambda_0,X_0$ be the variables after initialization at line 1. 
Note that $Y_i={\tt RandomSubset}(X_{i-1};\alpha_i)$ for $i\in[N+1]$. We always use letter $X$ to denote the original set $X=X_0$.

We will need some additional notation.
Recall that  $C_{sv}$ for $v\in X$ is the minimum $s$-$v$ cut in~$G$ (which is now unique),
and the family $\{C_{sv}\::\:v\in X\}$ is laminar.
We write $u\sim v$ for nodes $u,v\in X$ if $C_{su}=C_{sv}$.
Clearly, $\sim$ is an equivalence relation on $X$.
Let $\Phi$ be the set of equivalence classes of this relation with $\langle\Phi\rangle=X$,
and for $v\in X$ let $\br{v}\in\Phi$ be the class to which $v$ belongs.
We write $\br{u}\preceq \br{v}$ if $C_{su}\subseteq C_{sv}$.
Clearly, $\preceq$ is a partial order on $\Phi$. Furthermore, laminarity of cuts $C_{sv}$ implies that $\preceq$ can be represented by a rooted
forest $\calF$ on nodes $\Phi$ so that $A\preceq B$ for $A,B\in\Phi$ iff $A$ belongs to the subtree of $\calF$ rooted at $B$.
We denote ${\tt Roots}(\calF)$ to be the set of roots of this forest.
For node $A\in \Phi$ we denote $A^\downarrow =\{B\in\Phi\::\:B\preceq A\}$; equivalently, $A^\downarrow$
is the set of nodes of the subtree of $\calF$ rooted at $A$.
It can be seen that $C_{sv}\cap X=\langle \br{v}^\downarrow \rangle$ for $v\in X$.
We define function $f:\Phi\rightarrow\mathbb R$ via  $f(\br{v})={\tt cost}(C_{sv})=f(s,v)$ for $v\in X$.
Assumption~\ref{as} implies that $f(B)<f(A)$ for $B\prec A$.

\begin{lemma} \label{lemma:PPP}
For each $R\in{\tt Roots}(\calF)$ and $i\in [N]$ there holds $R\cap X_i\ne \varnothing$.
\end{lemma}
\begin{proof}
We use induction on $i$ (for a fixed $R\in{\tt Roots}(\calF)$).
Consider the $i$-th iteration, and let  $\calS$  be the ${\tt OC}_1$ tree computed at line 5.
We make the following claim:
\begin{itemize}
\item If $v\in \calS$ and $v\notin R$ then $\calS_v\cap R=\varnothing$.
Indeed, conditions $v\notin R$ and $R\in{\tt Roots}(\calF)$ imply that $C_{sv}\cap R=\varnothing$.

By the definition of ${\tt OC}_1$ tree, set $\calS_v$ is a minimum $s \alpha$-$v$ cut for some sequence $\alpha$.
Lemma~\ref{lemma:parametric} gives that $\calS_v\subseteq C_{sv}$, implying the claim.
\end{itemize}
The claim of Lemma~\ref{lemma:PPP} now easily follows
(recalling that the update $X:=X-(\calS_v-\{v\})$ at line 7 is equivalent to the update $X:=(X-\calS_v)\cup\{v\}$).
\end{proof}

Consider node $R\in\Phi$.
We say that subset $\Psi\subseteq\Phi$ is  {\em $R$-rooted} if it can be represented as $\Psi=R^\downarrow-\bigcup_{A\in \tilde\Psi} A^\downarrow$
for some subset $\tilde\Psi\subseteq R^\downarrow$. The minimal set $\tilde\Psi$ in this representation will be denoted as $\Psi^-$;
clearly, such minimal $\tilde\Psi$ exists, and equals the set of maximal elements of $R^\downarrow -\Psi$ w.r.t.\ $\preceq$.
One exception is the set $\Psi=\varnothing$: for such $\Psi$ we define $\Psi^-=\{R\}$ 
(node $R$ will always be clear from the context).
It can be seen that if $\Psi$ is $R$-rooted and $A,B$ are distinct nodes in $\Psi^-$ then sets $A^\downarrow ,B^\downarrow $ are disjoint.

Now consider subsets $Y\subseteq X$, $\Psi\subseteq \Phi$ and node $v\in X$. 
Borrowing terminology from~\cite{LiPanigrahi:STOC21}, we say that
$Y$ {\em hits $v$ in $\Psi$} if $C_{sv}\cap \langle \Psi\rangle\cap Y=\{v\}$. 
We denote
$$
\Psi[Y]=\Psi-\bigcup_{\mbox{\scriptsize $v:Y$ hits $v$ in $\Psi$}} \br{v}^\downarrow
$$
For subsets $Y_1,\ldots,Y_k$ we also denote $\Psi[Y_1,\ldots,Y_k]=\Psi[Y_1]\ldots[Y_k]$.
Clearly, if $\Psi$ is $R$-rooted then so is $\Psi[Y_1,\ldots,Y_k]$.
With this notation in place, we proceed with the analysis of Algorithm~\ref{alg:SSPQ1}.

\begin{lemma}\label{lemma:QQQ}
Consider root $R\in{\tt Roots}(\calF)$,
and for $i\in[N]$ denote $\Psi_i=R^\downarrow[Y_1,\ldots,Y_i]$, so that
$R^\downarrow=\Psi_0\supseteq\Psi_1\supseteq\ldots\supseteq\Psi_N$.
Then the following holds for every $A\in \Psi^-_i$, $i\in[N]$:
\begin{itemize}
\item $|\langle A^\downarrow\rangle \cap X_i|\le 1$. If $\langle A^\downarrow\rangle \cap X_i=\{v\}$ then $v\in A$ and $\lambda_i(v)=f(A)$.
\end{itemize}
\end{lemma}
\begin{proof}
We use induction on $i$. 
If $i=0$ then $\Psi_i^-=\varnothing$ and the claim is vacuous.
 Suppose it holds for $i-1$; let us prove for $i\in[N]$.

Note that $\Psi_i =\Psi_{i-1}[ Y_i]$.
 Consider $A\in \Psi_i^-$. 
 If $A\in\Psi_{i-1}^-$ then the claim holds by the induction hypothesis and the facts that $X_i\subseteq X_{i-1}$
 and $\mu_i(v)\le \mu_{i-1}(v)$ for all~$v$.
 Suppose that $A\notin\Psi_{i-1}^-$. Condition $A\in\Psi_i^-=(\Psi_{i-1}[Y_i])^-$ then
 implies that there exists $v\in A$ such that $Y_i$ hits $v$ in $\Psi_{i-1}$, i.e.\ $C_{sv}\cap \langle\Psi_{i-1}\rangle\cap Y_i=\{v\}$.
 We then have $A=\br{v}\in\Psi_{i-1}$.

 Let $\varphi=s\alpha v\ldots$ be the ordering of nodes in $\{s\}\cup Y_i$ used in the algorithm.
 We claim that $\alpha\cap C_{sv}=\varnothing$. 
 Indeed, suppose that there exists $u\in \alpha\cap C_{sv} \subseteq Y_i\cap C_{sv}$,
 then we must have $u\notin\langle\Psi_{i-1}\rangle$ and hence $\br{u}\notin \Psi_{i-1}$.
 Condition $u\in C_{sv}$ implies that $\br{u}\preceq {\br v}=A$.
 Conditions $A\in\Psi_{i-1}$, $\br{u}\preceq A$, $\br{u}\notin \Psi_{i-1}$ imply that
 there exists $B\in\Psi_{i-1}^-$ with $\br{u}\preceq B\prec A$.
 We have $u\in \langle B^\downarrow \rangle$ and $u\in Y_i\subseteq X_{i-1}$.
 The induction hypothesis gives that $\langle B^\downarrow\rangle\cap X_{i-1}=\{u\}$, $u\in B$ and $\lambda_{i-1}(u)=f(B)$.
  We obtain that $\lambda_{i-1}(u)=f(B)<f(A)\le \lambda_{i-1}(v)$ and so $u$ should come after $v$ in $\varphi$ - a contradiction.
 
 Let $\calS={\tt OrderedCuts}_1(\varphi;G)$  be the ${\tt OC}_1$ tree 
 computed at iteration $i$. By the definition ${\tt OC}_1$ tree, there exists $u\in(\alpha v)\cap \calS$ such that 
 $\calS_u$ is a minimum $s\alpha'$-$u$ cut for some $\alpha'\subseteq\alpha-\{u\}$, and $v\in \calS_u$. 
 As shown above, we have $\alpha'\cap C_{sv}=\varnothing$. 
  By construction, we have $X_i\subseteq X_{i-1}-(\calS_u-\{u\})$. 
  Recall that $C_{sv}\cap X=\langle A^\downarrow\rangle$ and so $\langle A^\downarrow \rangle\cap X_i\subseteq C_{sv}\cap X_i$.
  We now consider two possible cases.
  \begin{itemize}
 \item $u=v$. The conditions above imply that $\calS_v=C_{sv}$, and so
 $\langle A^\downarrow \rangle\cap X_i\subseteq C_{sv}\cap X_i\subseteq\{v\}$.
 Furthermore, if $v\in X_i$ then we will have $\lambda_i(v)\le {\tt cost}(\calS_v)=f(A)$.
\item $u\ne v$. Since $\alpha'\cap C_{sv}=\varnothing$, set $C_{sv}$ is also (the unique) minimum $s\alpha'$-$v$ cut.
Recall that $\calS_u$ is a minimum $s\alpha'$-$u$ cut and $v\in\calS_u$,
and so $\calS_u$ is a minimum $s\alpha'$-$uv$ cut.
Lemma~\ref{lemma:parametric} gives that $C_{sv}\subseteq \calS_u$.
Condition $\alpha\cap C_{sv}=\varnothing$ implies that $u\notin C_{sv}$
and so $C_{sv}\subseteq \calS_u-\{u\}$.
Thus, we will have $\langle A^\downarrow \rangle\cap X_i\subseteq C_{sv}\cap X_i=\varnothing$.
 \end{itemize}
In both cases the claim of Lemma~\ref{lemma:QQQ} holds.
 
\end{proof}

\begin{lemma}\label{lemma:RRR}
Suppose that $R^\downarrow[Y_1,\ldots,Y_N]=\varnothing$ for all $R\in{\tt Roots}(\calF)$.
Then the output of Algorithm~\ref{alg:SSPQ1} satisfies $X\subseteq\langle\calS^\ast\rangle$.
\end{lemma}
\begin{proof}
Condition $R^\downarrow[Y_1,\ldots,Y_N]=\varnothing$ implies that $(R^\downarrow[Y_1,\ldots,Y_N])^-=\{R\}$.
Lemma~\ref{lemma:QQQ} gives that $|\langle R^\downarrow \rangle \cap X_N|\le 1$.
By Lemma~\ref{lemma:PPP} we conclude that $\langle R^\downarrow \rangle \cap X_N=\{v\}$ for some $v\in R$.
Furthermore, $\lambda_N(v)=f(R)$ by Lemma~\ref{lemma:QQQ}. 

To summarize, we showed that $X_N\subseteq\langle{\tt Roots}(\calF)\rangle$, and for each $R\in{\tt Roots}(\calF)$ there exists unique $v\in R\cap X_N$;
this $v$ satisfies $\lambda_N(v)=f(R)$. Let $\varphi=sv_1\ldots v_\ell$ be the ordering of nodes in $Y_N=X_N$ used in the algorithm,
and let $\calS={\tt OrderedCuts}_1(\varphi;G)$ be the computed ${\tt OC}_1$ tree for $(\varphi,G)$.
By the definition of ${\tt OC}_1$ tree, for each $v\in \calS$ set $\calS_v$ is a minimum $s\alpha$-$v$ cut for some $\alpha\subseteq X_N-\{v\}$.
We have $\alpha\cap C_{sv}=\varnothing$, and thus $\calS_v=C_{sv}$.
We must also have $\calS=X_N$ (if there exists $v\in X_N-\calS$ then $v\in \calS_u=C_{su}$ for some $u\in \calS$, $u\ne v$ - a contradiction).
Since $\lambda(v_1)\ge \ldots \ge \lambda(v_\ell)$, we will have ${\tt cost}(\calS_{v_1})\ge \ldots{\tt cost}(\calS_{v_\ell})$ and so line 9 outputs set $\calS^\ast=\calS$.
The claim follows.
\end{proof}

Our next goal is to construct sequence $\alpha_1,\ldots,\alpha_{N}$ such that the precondition of Lemma~\ref{lemma:RRR}
holds w.h.p.. For that we will use the argument from~\cite{LiPanigrahi:STOC21}. First, we show the following.

\begin{lemma}\label{lemma:expectation-reduction}
Let $\Psi$ be an $R$-rooted set for some $R\in\Phi$, and 
let $Z_0,\ldots,Z_d$ be a sequence of independent subsets generated via $Z_i\leftarrow {\tt RandomSubset}(X,2^{-i})$ where $d=\lfloor \log_2 |X| \rfloor$.
Then $\mathbb E|\langle\Psi[Z_0,\ldots,Z_d]\rangle|\le \left(1-\tfrac{\Theta(1)}{1+d}\right) |\langle \Psi\rangle|=\left(1-\tfrac{\Theta(1)}{\log|X|}\right) |\langle \Psi\rangle|$.
\end{lemma}
\begin{proof}
Define directed multigraph $(\langle\Psi\rangle,P)$ as follows: for every $i\in[0,d]$ and every $v\in \langle\Psi\rangle$ such that $Z_i$ hits $v$ in $\Psi$
add edges $\{(v,u)\::\:u\in C_{sv}\cap\langle\Psi\rangle \}$ to $P$. (We say that these edges are {\em at level~$i$}). We will show the following: \\
{\rm (a)}  The expected out-degree of each $v\in \langle \Psi\rangle$ is $\Omega(1)$. \\
{\rm (b)} The in-degree of each $u\in \langle \Psi\rangle$ is at most $d+1$. \\
{\rm (c)} If $(v,u)\in P$ then $u\notin \Psi[Z_0,\ldots,Z_d]$. 

Part (a) will imply that $P$ has $\Omega(\langle\Psi\rangle)$ edges in expectation;
combined with (b), this will mean that the expected number of nodes $u\in\langle \Psi\rangle$  with at least one incoming edge
is $\Omega(|\langle\Psi\rangle|/(d+1))$.
Part~(c) will then give the claim of the lemma. We now focus on proving (a)-(c).

\myparagraph{Part (a)}
Denote $n_v=|C_{sv}\cap\langle \Psi\rangle|$. There must exist $i\in[0,d]$ such that $n_v\in[2^i,2^{i+1}]$.
The probability that $Z_i$ hits $v$ in $\Psi$
equals $2^{-i}\cdot (1-2^{-i})^{n_v-1}=\Theta(2^{-i})$. If this event happens then $n_v$ edges of the form $(v,u),u\in C_{sv}\cap\langle\Psi\rangle$ are added to $P$.
Thus, the expected number of outgoing edges from $v$ at level $i$ is at least $\Theta(2^{-i})\cdot n_v=\Theta(1)$.

\myparagraph{Part (b)} 
Suppose that $P$ contains edges $(v,u)$, $(v',u)$ with $v\ne v'$ at the same level $i$. 
$Z_i$ hits both $v$ and $v'$ in $\Psi$, so $u\in C_{sv}\cap C_{sv'}$ and thus by laminarity
either $C_{sv}\subsetneq C_{sv'}$ or $C_{sv'}\subsetneq C_{sv}$.
Assume w.l.o.g.\ that $C_{sv'}\subsetneq C_{sv}$. 
Then $C_{sv}\cap \langle \Psi\rangle\cap Z_i$ contains at least two nodes ($v$ and $v'$), which is a contradiction.

\myparagraph{Part (c)}
Denote $\Psi_i=\Psi[Z_1,\ldots,Z_i]$ for $i\in[0,d]$,
and consider edge $(v,u)\in P$ at level $i$.
Let $A=\br{v}\in R^\downarrow$.
We claim that $A\notin \Psi_i$ (implying that $A\notin \Psi_d$ and thus $u\notin\langle \Psi_d\rangle$ since $u\in \langle A^\downarrow\rangle=C_{sv}\cap X$).
Indeed, assume that $A\in \Psi_i\subseteq\Psi_{i-1}$.
Since $Z_i$ hits $v$ in $\Psi$ we have $C_{sv}\cap \langle \Psi\rangle\cap Z_i=\{v\}$.
This means that $C_{sv}\cap \langle \Psi_{i-1}\rangle\cap Z_i=\{v\}$,
i.e.\ $Z_i$ hits $v$ in $\Psi_{i-1}$. This in turn implies that $A\notin \Psi_{i-1}[Z_i]= \Psi_i$, which is a contradiction.

\end{proof}

\begin{corollary}\label{cor:AHGASFAS}
Denote $n=|X|$.
For every fixed polynomial $p(n)$ there exists a sequence $\alpha_1,\ldots,\alpha_{N}$ with $N=\Theta(\log^3 n)$
such that the precondition of Lemma~\ref{lemma:RRR} holds with probability at least $1-\tfrac 1 {p(n)}$.
\end{corollary}

\begin{proof}
Let $\alpha_1,\ldots,\alpha_{N}$ be the sequence obtained by repeating the sequence used in Lemma~\ref{lemma:expectation-reduction}
(i.e.\ $2^{-0},2^{-1},\ldots,2^{-d}$) $K$ times, so that $N=K(d+1)$. If $K=\Theta(\log n)$ (with an appropriate constant)
then for any fixed $R\in\Phi$ we will have $\mathbb E[|\langle R^\downarrow[Y_1,\ldots,Y_N]\rangle|]\le \tfrac 1 {np(n)}$,
and thus $R^\downarrow[Y_1,\ldots,Y_N]=\varnothing$ with probability at least $1-\tfrac 1 {np(n)}$.
The precondition of Lemma~\ref{lemma:RRR}  has at most $n$ events of the form above,
so by the union bound they will all hold with probability at least $1-\tfrac 1 {p(n)}$.
\end{proof}
This concludes the proof of Theorem~\ref{th:FixedSourcePartition1}.

\subsection{Overall algorithm} \label{sec:select-s1}

We now come back to the problem of computing the GH tree.
Let us recall  the main computational problem at line 4 of Algorithm~\ref{alg:GH'}:
given graph $H=(V_H,E_H)$ and subset $X\subseteq V_H$ with $|X|\ge 2$,
we need to compute node $s\in X$ and a set of disjoint subsets $\Pi$
such that each $S\in\Pi$ is an $s$-$t$ cut in $H$ for some $t\in X$.
From now on we fix graph $H$ and set $X$.

Let us introduce a total order $\preceq$ on cuts $S\subseteq V_H$ as follows:
(i)~if ${\tt cost}(S)<{\tt cost}(S')$ then $S\prec S'$;
(ii)~if ${\tt cost}(S)={\tt cost}(S')$ and $|S\cap X|<|S'\cap X|$ then $S\prec S'$;
(iii)~otherwise use an arbitrary rule.
We write $u\sqsubset v$ for distinct nodes $u,v\in X$ if $C_{vu}\prec C_{uv}$.
Equivalently, $u\sqsubset v$ if $u\in S$ where $S$ is the minimum subset w.r.t.\ $\preceq$
that separates $u$ and $v$.
We write $u\sqsubseteq v$ if either $u\sqsubset v$ or $u=v$.
\begin{lemma}\label{lemma:total-order}
Relation $\sqsubseteq$ is a total order on $X$.
\end{lemma}
\begin{proof}
Clearly, $\sqsubseteq$ is reflexive and antisymmetric.
It thus suffices to show that for any distinct $a,b,c\in X$ relation $\sqsubseteq$ is not cyclic on $a,b,c$.
We can assume w.l.o.g.\  that  $C_{ab}$ is the smallest subset w.r.t.\ $\preceq$ among $C_{xy}$ for distinct $x,y\in \{a,b,c\}$.
Note that $b\sqsubset a$.
If $c\in C_{ab}$ then $C_{ab}$ is an $a$-$c$ cut with $C_{ab}\preceq C_{ac}$, implying $C_{ac}=C_{ab}$ and $c\sqsubset a$.
If $c\notin C_{ab}$ then $C_{ab}$ is an $c$-$b$ cut with $C_{ab}\preceq C_{cb}$, implying $C_{cb}=C_{ab}$ and $b\sqsubset c$.
In both cases $\sqsubseteq$ is not cyclic on $a,b,c$.
\end{proof}

Total order $\sqsubseteq$ will play an important role in the algorithm.
We will also need to make sure that Assumption~\ref{as} holds.
For that we can use a standard technique of adding a small random perturbation
to the edge weights. The correctness of this technique follows from the Isolation Lemma of~\cite{MVV:87} (see e.g.\ \cite[Proposition 3.14]{Abboud:FOCS20}).
\begin{proposition}\label{prop:GASF}
One can add random polynomially-bounded values to the edge weights in $H$
such that any GH tree in the new graph $H'$ is also a valid GH tree for $H$,
and furthermore with high probability $H'$ has a unique GH tree.
\end{proposition}

We are now ready to describe the algorithm (see below);  the sequence of probabilities $\beta_1,\ldots,\beta_M$ will be specified later.

\begin{algorithm}[H]
\setcounter{AlgoLine}{0}

  \DontPrintSemicolon
    \While{\tt true}
    {
       pick arbitrary $s\in X$  \\
    	apply random perturbation to $H$ as in Proposition~\ref{prop:GASF} to get graph $H'$ \\
	\For{$i=1,\ldots,M$}
	{
	    	sample $Y\leftarrow{\tt RandomSubset}(X-\{s\};\beta_i)$ \\
	    	call $\calS\leftarrow {\tt FixedSourcePartition}_1(s,Y;H')$  \\
	    	if exists $v\in\calS$ with $V_H-\calS_v\prec \calS_v$ then pick arbitrary such $v$ and update $s:=v$ \\
		else if $X-\{s\}\subseteq\langle\calS\rangle$ then terminate and return $(s,\{\calS_v\::\:v\in\calS\})$
	}
    }
      \caption{Line 4 of Algorithm~\ref{alg:GH'}. 
      }\label{alg:select-s1}
\end{algorithm}

\begin{lemma}\label{lemma:GAKHDGAG}
Suppose that 
$(\beta_1,\ldots,\beta_M)=(
\underbrace{2^{-d},\ldots,2^{-d}}_{K\mbox{\scriptsize\em~times}},\ldots,
\underbrace{2^{-2},\ldots,2^{-2}}_{K\mbox{\scriptsize\em~times}},
\underbrace{2^{-1},\ldots,2^{-1}}_{K\mbox{\scriptsize\em~times}},
1)$ where $d=\lfloor\log_2|X|\rfloor$ and $K=\Theta(\log \log n)$ (with an appropriate constant).
Then each run of Algorithm~\ref{alg:select-s1} (i.e.\ lines 2-8) terminates with probability $\Omega(1)$.
\end{lemma}
\begin{proof}
We can assume that the modification of $H$ at line 3 produces graph $H'$ with a unique GH tree (this happens with probability $\Omega(1)$),
and so Assumption~\ref{as} always holds. We have, in particular, $C_{uv}=V_H-C_{vu}$ for any $u,v\in X$.

Denote $n=|X|$, and for $s\in X$ let ${\tt rank}(s)\in[n]$ be the rank of $s$ according to the reverse of total order $\sqsubseteq$
(so that the largest node $s$ w.r.t.\ $\sqsubseteq$ satisfies ${\tt rank}(s)=1$).
Clearly, if node $s$ is updated at line 7 then we must have $s\sqsubset v$, i.e.\ every such update strictly decreases the rank of $s$.

For $p\in[d]$ let ``$p$-stage'' be the sequence of iterations $i$ that use $\beta_i=2^{-p}$.
We say the $p$-stage is successful if at the end of this stage node $s$ satisfies ${\tt rank}(s)\le 2^{p-1}$ (or if the algorithm has terminated earlier).
We claim that the $p$-stage is successful with probability at least $1-\tfrac{1}{2d}$ assuming that either $p=d$ or $p\in[d-1]$ and $(p+1)$-stage is successful.
Indeed, consider a single iteration (lines 5-8) with $\beta_i=2^{-p}$.
If ${\tt rank}(s)\le 2^{p-1}$ then the claim holds; we can thus assume that ${\tt rank}(s)\in (2^{p-1},2^p]$.
Let $Y\subseteq X-\{s\}$ be the set sampled at this iteration.
The following events will hold jointly with probability at least~$1-\gamma$ for some constant $\gamma>0$:
$Y\cap \{v\in X\::\:{\tt rank}(v)\le 2^{p-1}\}\ne\varnothing$ and
$Y\cap \{v\in X\::\:2^{p-1}<{\tt rank}(v)\le 2^{p}\}=\varnothing$.
The first event implies that this iteration will update $s$,
and the second event implies that the new $s$ will satisfy ${\tt rank}(s)\le 2^{p-1}$.
The $p$-stage has $K$ iterations, and thus it will be successful with probability at least $1-\gamma^K=1-\gamma^{\Theta(\log\log n)}\ge 1-\tfrac{1}{2d}$
with an appropriate constant in the $\Theta(\cdot)$ notation.

There are $d$ such stages in total, so all of them will be successful with probability at least $\tfrac{1}{2}$.
If this happens, $s$ becomes the largest node in $X$ w.r.t.\ $\sqsubseteq$,
and so the last iteration (with $\beta_M=1$) will produce named subpartition $\calS$ with $X-\{s\}\subseteq\langle\calS\rangle$.
\end{proof}

Next, we analyze the overall complexity of Algorithm~\ref{alg:GH'}.
Each supernode $X$ that appears during the execution of  can be assigned a depth using natural rules:
(i) the initial supernode $V$ is at depth $0$;
(ii) if supernode $X$ at depth $d$ is split into supernodes $X_0,\ldots,X_k$ then the latter supernodes are assigned depth $d+1$.
By construction, in the latter case we have $|X_i|\le |X|/2$ for $i\in[0,k]$.
Therefore, supernode $X$ at depth $d$ satisfies $|X|\le n\cdot 2^{-d}$, and so the maximum depth is $O(\log n)$.
We will need the following result from~\cite{Abboud:FOCS20}.
\begin{lemma}[{\cite[Lemma 3.12]{Abboud:FOCS20}}]
Let $X_1,\ldots,X_{k(d)}$ be the supernodes at depth $d$,
and let $H_1,\ldots,H_{k(d)}$ be the corresponding auxiliary graphs.
Then the total number of edges in $H_1,\ldots,H_{k(d)}$ is $O(m)$.
\end{lemma}
An analogous result also holds for the number of nodes.
\begin{lemma}
Let $X_1,\ldots,X_{k(d)}$ be the supernodes at depth $d$,
and let $H_1,\ldots,H_{k(d)}$ be the corresponding auxiliary graphs.
Then the total number of nodes in $H_1,\ldots,H_{k(d)}$ is $O(n)$.
\end{lemma}
\begin{proof}
Consider graph $H=H_i$ corresponding to supernode $X=X_i$.
Recall that $V_H=X\cup\{v_Y\::\:XY\in\calT\}$ where $\calT$ is the partition tree at the moment when $X$ is processed.
Clearly, sets $X_1,\ldots,X_{k(d)}$ are disjoint, so we need to bound
the total number of nodes of $H$ due to edges $XY\in\calT$ (let us call them ``contracted nodes'').
After $X$ is processed, each $XY$ is transformed to an edge $X'Y$ for some supernode $X'\subset X$,
other edges of $\calT$ are kept intact (or ``transformed to themselves''), and some new edges are added to $\calT$.
Let us track how edge $XY$ transforms during the algorithm; eventually, it becomes an edge $xy$ of the final GH tree, with $x\in X$ and $y\in Y$.
Let us charge node $v_Y$ in graph $H$ to node $x$. Clearly, no other node in graphs $H_1,\ldots,H_{k(d)}$ is charged to $x$.
The number of edges in the final GH tree is $n-1$, and thus the total of contracted nodes in $H_1,\ldots,H_{k(d)}$ is at most $2(n-1)$.
\end{proof}


By putting all results together, we obtain 
\begin{theorem}\label{th:GNNAHGDHDJADHGLKDGALDHGADKHADFAAHGLKADHF}
The expected complexity of Algorithm~\ref{alg:GH'} with line 4 implemented as in Algorithm~\ref{alg:select-s1} is
$O((t_{{\tt OC}_1}(n,m))+n\log n)\cdot \log^5 n \cdot \log\log n)$ where $t_{{\tt OC}_1}(n,m)$ is the complexity of procedure ${\tt OrderedCuts}_1(\cdot;G)$ on a graph with $O(n)$ nodes and $O(m)$ edges.
\end{theorem}

\section{Computing Gomory-Hu tree via {\tt OrderedCuts}}\label{sec:alg}
In this section we show how to compute  GH tree via $\tilde O(1)$ calls to (weak) ${\tt OrderedCuts}$.
Note that ${\tt OC}_1$ tree and weak {\tt OrderedCuts} for given input $(\varphi,G)$ are incomparable.
In particular, cuts contained in an ${\tt OC}_1$ tree may not give a valid output of weak ${\tt OrderedCuts}$ (as defined in Section~\ref{sec:intro}),
while the output of weak ${\tt OrderedCuts}$ comes without a tree structure.
Note, if we have a procedure that computes an ${\tt OC}$ tree for $\varphi$ then we can easily compute both ${\tt OC}_1$ tree and valid output of weak {\tt OrderedCuts}.
In that case the algorithm in this section will be more efficient than the one in the previous section by a factor of $(\log n)\cdot (\log \log n)$.

\subsection{Procedure ${\tt CertifiedOrderedCuts}$}

Consider undirected weighted graph $G=(V,E,w)$ and a sequence of distinct nodes $s,v_1,\ldots,v_\ell$ in~$V$.
First, we describe a procedure that computes some $s$-$v_i$ cut for each $i\in\ell$, and certifies
some of these cuts as ``true'' minimum $s$-$v_i$ cuts.

\begin{algorithm}[H]
\setcounter{AlgoLine}{0}
  \DontPrintSemicolon
\SetNoFillComment
  call $\lambda\leftarrow{\tt OrderedCuts}(s,v_1,\ldots,v_\ell;G)$ \\
  compute $Y^\ast=\{v_k\::\: \lambda(v_k)\le \min_{i\in[k-1]} \lambda(v_i), k\in[\ell]\}$ \\

  		call ${\tt IsolatingCuts}(s,Y^\ast;G)$ (cf. Theorem~\ref{th:IsolatingCuts}) to get cut $S_v\subseteq V-\{s\}$ for each $v\in Y^\ast$\!\!\!\!\!\!\!\!\!\!\! \\
  		define $\calC^\ast=\{S_v\::\:v\in Y^\ast, ~{\tt cost}(S_v)=\lambda(v)\}$ \\
  return $(\lambda,\calC^\ast)$
      \caption{${\tt CertifiedOrderedCuts}(s,v_1,\ldots,v_\ell;G)$. 
      }\label{alg:SC-via-OC}
\end{algorithm}

\begin{lemma}\label{lemma:SC}
Let  $(\lambda,\calC^\ast)$ be the output of ${\tt CertifiedOrderedCuts}(s,v_1,\ldots,v_\ell;G)$, and denote $Y=\{v_1,\ldots,v_\ell\}$. Pair $(\lambda,\calC^\ast)$ has the following properties. \\
{\rm (a)} Subsets in $\calC^\ast$ are disjoint, and each $S\in \calC^\ast$ is a minimum $s$-$v$ cut for some $v\in Y$. \\
{\rm (b)} If $Y\cap C_{sv_k}=\{v_k\}$ and $f(s,v_k)\le \min\nolimits_{i\in[k]}f(s,v_i)$ then $\calC^\ast$ contains a minimum $s$-$v_k$ cut. \\
{\rm (c)} There exists a family of cuts $\calC\subseteq\{S:\varnothing\ne S\subseteq V-\{s\}\}$ such that \\
\indent \!\!\!\!\!{\rm (i)~} $\lambda(v)=\min_{S\in\calC:v\in S} {\tt cost}(S)$ for each $v\in V-\{s\}$ (and hence $\lambda(v)\ge f(s,v)$); \\
\indent \!\!\!\!\!{\rm (ii)} if $\{v_1,\ldots,v_k\}\cap C_{sv_k}=\{v_k\}$ then $\calC$ contains a minimum $s$-$v_k$ cut (and hence $\lambda(v_k)=$ 
\indent $f(s,v_k)$ 
 and $\lambda(v)\le f(s,v_k)$ for all $v\in C_{sv_k}$).  
\end{lemma}
\begin{proof}
We take $\calC=\{S_k\::\:k\in[\ell]\}$ to be the set of cuts considered inside the call to ${\tt OrderedCuts}(\cdot)$,
then the condition in {\rm (c.i)} holds by the definition of ${\tt OrderedCuts}$.
We claim that for each $v_k\in Y^\ast$ we have $\lambda(v_k)=f(s,v_k)$.
Indeed, let $S$ be a minimum $s$-$v_k$ cut 
and $i$ be the largest index in $[k]$ such that $\{s,v_1,\ldots,v_{i-1}\}\subseteq V-S$, then $v_i\in S$
and so $f(s,v_k)={\tt cost}(S)\ge f(\{s,v_1,\ldots,v_{i-1}\},v_i)\ge \lambda(v_i)\ge \lambda(v_k)$, implying the claim.
(The middle inequality holds by the definition of ${\tt OrderedCuts}$ and the last inequality holds since $v_k\in Y^\ast$).
We are now ready to prove the lemma.

\myparagraph{Part (a)} Sets $S_v$ at line 3 (and thus sets in $\calC^\ast$) are disjoint by Theorem~\ref{th:IsolatingCuts}.
Set $S_v$ is added to $\calC^\ast$ only if $v\in Y^\ast$ and ${\tt cost}(S_v)=\lambda(v)=f(s,v)$, and so the claim holds.

\myparagraph{Parts (b) and (c.ii)} Consider $k\in[\ell]$ with $\{v_1,\ldots,v_k\}\cap C_{sv_k}=\{v_k\}$.
The condition implies that $C_{sv_k}$ is a minimum $\{s,v_1,\ldots,v_{k-1}\}$-$v_k$ cut
and hence $f(\{s,v_1,\ldots,v_{k-1}\},v_k)=f(s,v_k)$.
Therefore, set $S_k\in\calC$  must satisfy ${\tt cost}(S_k)\le f(s,v_k)$,
i.e. $S_k$ is a minimum $s$-$v_k$ cut. This proves part {\rm (c.ii)}.
Now suppose that $Y\cap C_{sv_k}=\{v_k\}$ and $f(s,v_i)\ge f(s,v_k)$ for each $i\in[k]$.
We have $\lambda(v_k)=f(s,v_k)$ and $\lambda(v_i)\ge f(s,v_i)$ for $i\in[k]$,
therefore node $v_k$ will be added to $Y^\ast$.
We also have $Y^\ast\cap C_{sv_k}=\{v_k\}$, and so set $C_{v_k}$
returned by ${\tt IsolatingCuts}(s,Y^\ast;G)$ 
is a minimum $s$-$v_k$ cut.
This set satisfies ${\tt cost}(C_{sv_k})=\lambda(v_k)=f(s,v_k)$, and so $S_{v_k}$ will be added to $\calC^\ast$.
\end{proof}

\begin{remark}\label{remark:CertifiedOrderedCuts}
We will use only properties of ${\tt CertifiedOrderedCuts}(\cdot)$
given in Lemma~\ref{lemma:SC}. Any other implementation of ${\tt CertifiedOrderedCuts}(\cdot)$ would work,
as long as Lemma~\ref{lemma:SC} holds. If, for example, a strong version of ${\tt OrderedCuts}$ is used
that explicitly constructs {\tt OC} tree then set $\calC^\ast$ can be constructed from this tree
without calling ${\tt IsolatingCuts}$, see Lemma~\ref{lemma:pistar}.
Note that this lemma gives a weaker criterion for adding cuts to $\calC^\ast$ that in general can certify a larger set of cuts stored in an {\tt OC} tree.
\end{remark}

\subsection{Computing cuts for a fixed source}\label{sec:fixed-s} 

Let us consider graph $G=(V,E,w)$, node $s\in V$, subset $X\subseteq V-\{s\}$ and number $L\ge 1$.
Denote 
\begin{eqnarray*}
\widehat\Pi_s^{(L)}&=&\{S\subseteq V-\{s\}\::\:S\mbox{ is a minimum $s$-$t$ cut for some $t\in X,~$}|S\cap X|\le L\} \\
\Pi_s^{(L)}&=&\{S=C_{st}\::\:t\in X,\;|S\cap X|\le L\} \;\;\subseteq\;\;   \widehat\Pi_s^{(L)}
\end{eqnarray*}
We say that the output of ${\tt CertifiedOrderedCuts}(\cdot)$ is {\em minimal} if under the condition {\rm (b)} of Lemma~\ref{lemma:SC}
family $\calC^\ast$ contains the minimal minimum $s$-$v_k$ cut $C_{sv_k}$, rather than an arbitrary minimum $s$-$v_k$ cut.
The main result of this section is given by the following theorem.

\begin{theorem}\label{th:FixedSourcePartition}
There exists an algorithm that outputs a laminar family of subsets $\Pi\subseteq \widehat\Pi_s^{(L)}$.
It makes $O(\log^3|X|)$ calls to procedure ${\tt CertifiedOrderedCuts}$  for graph~$H$ and performs $ O(|X|\log^4|X|)$ additional work.
Furthermore, it satisfies $\langle \Pi\rangle=\langle \Pi_s^{(L)}\rangle$ 
with probability at least $1-1/poly(|X|)$
(for an arbitrary fixed polynomial) assuming that all outputs of ${\tt CertifiedOrderedCuts}$ are minimal.
\end{theorem}

Our algorithm is given below. 

\begin{algorithm}[H]
\setcounter{AlgoLine}{0}

  \DontPrintSemicolon
    set $\Pi=\varnothing$ and $\mu(v)={\tt cost}(\{v\})$ for all $v\in V-\{s\}$ \\
	\For{$i=1,\ldots,2N$}
  	{
  		sample $Y\leftarrow{\tt RandomSubset}(X-\langle\Pi\rangle;\alpha_i)$ \\
  		sort nodes in $Y$ as $v_1,\ldots,v_\ell$ so that $\mu(v_1)\ge \ldots \ge \mu(v_\ell)$, $\ell=|Y|$ \\
  		call $(\lambda,\calC^\ast)\leftarrow{\tt CertifiedOrderedCuts}(s,v_1,\ldots,v_\ell;G)$ \\
  		update $\mu(v):=\min\{\mu(v),\lambda(v)\}$ for $v\in Y$ \\
  		update $\Pi:=\Pi\cup\{S\in \calC^\ast\::\:|S\cap X|\le L\}$ \\
  		if $\Pi$ becomes non-laminar then  terminate and return $\varnothing$
  	}
  return $\Pi$

      \caption{${\tt FixedSourcePartition}(s,X,L;G)$. 
      }\label{alg:SSPQ}
\end{algorithm}

We set sequence $\alpha_1,\ldots,\alpha_{2N}$ in the same way as in Section~\ref{sec:FixedSourcePartition1}
(with $N=\Theta(\log^3|X|)$).
In the remainder of this section we prove that with this choice the complexity
of Algorithm~\ref{alg:SSPQ} is as stated in Theorem~\ref{th:FixedSourcePartition}.

By Lemma~\ref{lemma:SC}(a), set $\Pi$ in Algorithm~\ref{alg:SSPQ} always satisfies $\Pi\subseteq\widehat\Pi_s^{(L)}$.
For the rest of analysis we assume that  all outputs of ${\tt CertifiedOrderedCuts}$ are minimal,
and thus $\Pi\subseteq\Pi_s^{(L)}$.
We will show next that there exists
 sequence $\alpha_1,\ldots,\alpha_{2N}$ with $N=\Theta(\log^3 |X|)$ 
for which $\langle\Pi\rangle= \langle \Pi_s^{(L)}\rangle$  w.h.p..

We use the same notation as in Section~\ref{sec:FixedSourcePartition1}; let us recap it briefly. We write $u\sim v$ for nodes $u,v\in X$ if $C_{su}=C_{sv}$.
Let $\Phi$ be the set of equivalence classes of  relation $\sim$,
and for $v\in X$ let $\br{v}\in\Phi$ be the class to which $v$ belongs.
We write $\br{u}\preceq \br{v}$ if $C_{su}\subseteq C_{sv}$,
then $\preceq$ is a partial order on $\Phi$ that can be represented by a rooted
forest $\calF$ on nodes $\Phi$ so that $A\preceq B$ for $A,B\in\Phi$ iff $A$ belongs to the subtree of $\calF$ rooted at $B$.
For node $A\in \Phi$ we denote $A^\downarrow = \{B\in\Phi\::\:B\preceq A\}$.

Consider node $R\in\Phi$.
We say that subset $\Psi\subseteq\Phi$ is {\em $R$-rooted} if it can be represented as $\Psi=R^\downarrow -\bigcup_{A\in \tilde\Psi} A^\downarrow$
for some subset $\tilde\Psi\subseteq R^\downarrow$. The minimal set $\tilde\Psi$ in this representation will be denoted as $\Psi^-$.
One exception is the set $\Psi=\varnothing$: for such $\Psi$ we define $\Psi^-=\{R\}$ 
(node $R$ will always be clear from the context).

Consider subsets $Y\subseteq X$, $\Psi\subseteq \Phi$ and node $v\in \langle \Psi\rangle$. 
We say that
$Y$ {\em hits $v$ in $\Psi$} if $C_{sv}\cap \langle \Psi\rangle\cap Y=\{v\}$. 
We denote
$$
\Psi[Y]=\Psi-\bigcup_{\mbox{\scriptsize $v:Y$ hits $v$ in $\Psi$}} \br{v}^\downarrow
$$
For subsets $Y_1,\ldots,Y_k$ we also denote $\Psi[Y_1,\ldots,Y_k]=\Psi[Y_1]\ldots[Y_k]$.

We can assume that subset $Y$ in line 3 is generated by
first sampling $\hat Y\leftarrow{\tt RandomSubset}(X;\alpha_i)$ and then taking $Y=\hat Y-\langle\Pi\rangle$.
Let $\Pi_i,\lambda_i,\calC^\ast_i,\mu_i,\hat Y_i,Y_i$ for $i\in[2N]$ be the variables at the end of the $i$-th iteration,
and $\Pi_0,\mu_0$ be as initialized at line 1. 
We refer to iterations $i=1,\ldots,N$ as the {\em first stage},
and to iterations $i=N+1,\ldots,2N$ as the {\em second stage}.
The lemma below analyzes the first stage.

\begin{lemma} \label{lemma:first-stage}
Suppose that $R^\downarrow [\hat Y_1,\ldots,\hat Y_N]=\varnothing$ for $R\in\Phi$.
Then $\mu_N(v)=f(s,v)$ for all $v\in R-\langle\Pi_N\rangle$. 
\end{lemma}
\begin{proof}
Denote $\Psi_i= R^\downarrow [\hat Y_1,\ldots,\hat Y_i]$, then 
$R^\downarrow =\Psi_0\supseteq\Psi_1\supseteq\ldots\supseteq\Psi_N=\varnothing$.
Let us define function $f:\Phi\rightarrow\mathbb R$ via  $f(\br{v})={\tt cost}(C_{sv})=f(s,v)$ for $v\in X$.
We make the following claim: 
\begin{itemize}
\item[{$(\ast)$}]
 For each $i\in[0,N]$ and $A\in \Psi_i^-$ there holds $\mu_i(v)\le f(A)$ for $v\in\langle A^\downarrow \rangle-\langle\Pi_i\rangle$.
 \end{itemize}

 To prove this, we use induction on $i$. If $i=0$ then $\Psi_i^-=\varnothing$ and the claim is vacuous.
 Suppose it holds for $i-1$; let us prove for $i\in[N]$.
Note that $\Psi_i =\Psi_{i-1}[\hat Y_i]$.
 Consider $A\in \Psi_i^-$. 
 If $A\in \Psi^-_{i-1}$ then for each $v\in\langle A^\downarrow \rangle-\langle\Pi_i\rangle\subseteq \langle A^\downarrow \rangle-\langle\Pi_{i-1}\rangle$
 we have $\mu_i(v)\le \mu_{i-1}(v)\le f(A)$ where the last inequality is by the induction hypothesis.
 We thus suppose that $A\notin \Psi^-_{i-1}$.
Condition $A\in\Psi_i^-=(\Psi_{i-1}[\hat Y_i])^-$ then
 implies that there exists $a\in A$ such that $\hat Y_i$ hits $a$ in $\Psi_{i-1}$, i.e.\ $C_{sa}\cap \langle\Psi_{i-1}\rangle\cap \hat Y_i=\{a\}$.
 We then have $A=\br{a}\in\Psi_{i-1}$.
 
First, suppose that $a\notin Y_{i}$. Then $a\in  \langle\Pi_{i}\rangle$
(since $a\in\hat Y_{i}$ and $Y_{i}=\hat Y_{i}-\langle\Pi_{i}\rangle$) and hence $A\subseteq \langle\Pi_i\rangle$ and 
$\langle A^\downarrow \rangle\subseteq \langle\Pi_i\rangle$, so $(\ast)$ vacuously holds.
Next, suppose that $\mu_{i-1}(a)=f(s,a)=f(A)$. 
If $C_{sa}=\langle A^\downarrow \rangle=\{a\}$ then $(\ast)$ clearly holds.
Otherwise $\mu_0(a)<f(s,a)$ and hence $\lambda_j(a)=f(s,a)$ for some $j\in[i-1]$.
This means that set $\calC_j$ computed in the $j$-th call to ${\tt CertifiedOrderedCuts}$
contained set $S$ which is a minimum $s$-$a$ cut.
We have $C_{sa}\subseteq S$, and hence $\mu_{i-1}(v)\le \lambda_j(v)\le{\tt cost}(S)= f(s,a)$ for all $v\in C_{sa}\cap X=\langle A^\downarrow \rangle$,
i.e.\ $(\ast)$ holds.

It remains to consider the case when $a\in Y_i$ and $\mu_{i-1}(a)>f(s,a)=f(A)$.
 We claim that all $v\in (C_{sa}\cap Y_i)-\{a\}$  satisfy $\mu_{i-1}(v)\le f(A)$.
Indeed, we have $C_{sa}\cap \langle \Psi_{i-1}\rangle\cap \hat Y_i=\{a\}$, 
and so $v\notin \langle\Psi_{i-1}\rangle$. Furthermore, $v\in C_{sa}\cap X=\langle A^\downarrow \rangle$ and $A\in\Psi_{i-1}$.
These conditions imply that there exists $B\in \Psi_{i-1}^-$ with $B\preceq A$ (implying $f(B)\le f(A)$) and $v\in \langle B^\downarrow \rangle$.
By the choice of $Y_i$ at line 3 we have $Y_i\cap \langle \Pi_{i-1}\rangle=\varnothing$,
and so $v\notin \langle \Pi_{i-1}\rangle$ and $v\in \langle B^\downarrow \rangle-\langle \Pi_{i-1}\rangle$.
By the induction hypothesis, $\mu_{i-1}(u)\le f(B)$, which gives the claim.

Now consider the $i$-th step of Algorithm~\ref{alg:SSPQ}. As we just showed, we have
$\mu_{i-1}(v)\le f(A)<\mu_{i-1}(a)$ for all $v\in (C_{sa}\cap Y_i)-\{a\}$.
Thus, nodes in $Y_i$ will be sorted as $v_1,\ldots,v_{k-1},a,\ldots$
where $\{v_1,\ldots,v_{k-1}\}\cap C_{sa}=\varnothing$.
Lemma~\ref{lemma:SC}(c) now gives that $\mu_i(v)\le \lambda_i(v)\le {\tt cost}(C_{sa})=f(A)$ for all $v\in C_{sa}\cap X=\langle A^\downarrow \rangle$.

 We have proved claim $(\ast)$. Let us use it for $i=N$.
 We have  $\Psi_N=\varnothing$, so  $\Psi_N^-=\{R\}$
and thus $\mu_N(v)\le f(R)$ for all $v\in R-\langle\Pi_N\rangle$.
 We also have $\mu_N(v)\ge f(R)$ for such $v$, which gives the lemma.
 
\end{proof}

We say that the first stage is {\em successful} if the precondition of Lemma~\ref{lemma:first-stage} holds for every $R\in\Phi$.
By the lemma, if the first stage is successful then $\mu_N(v)=f(s,v)$ for all $v\in X-\langle\Pi_N\rangle$.
\begin{lemma}\label{lemma:second-stage}
Suppose that the first stage is successful and $R^\downarrow[\hat Y_{N+1},\ldots,\hat Y_{2N}]=\varnothing$
for $R\in\Phi$     with  $|\langle R^\downarrow\rangle|\le L$. 
Then $\langle R^\downarrow \rangle\subseteq \langle \Pi_{2N}\rangle$.
\end{lemma}
\begin{proof}
For $i\in[N,2N]$ 
denote $\Psi_i= R^\downarrow[\hat Y_{N+1},\ldots,\hat Y_i]$, then 
$R^\downarrow =\Psi_N\supseteq\Psi_{N+1}\supseteq\ldots\supseteq\Psi_{2N}=\varnothing$.
We make the following claim: 
\begin{itemize}
\item[{$(\ast)$}]
 For each $i\in[N,2N]$ and $A\in \Psi_i^-$ there holds $\langle A^\downarrow \rangle \subseteq \langle\Pi_i\rangle$.
 \end{itemize}

 To prove this, we use induction on $i$. If $i=N$ then $\Psi_i^-=\varnothing$ and the claim is vacuous.
 Suppose it holds for $i-1$; let us prove for $i\in[N+1,2N]$.
Note that $\Psi_i =\Psi_{i-1}[\hat Y_i]$.
 Consider $A\in \Psi_i^-$. 
 If $A\in \Psi^-_{i-1}$ then 
 $\langle A^\downarrow \rangle \subseteq \langle\Pi_{i-1}\rangle\subseteq \langle\Pi_i\rangle$ where the first inclusion is by the induction hypothesis.
 We thus suppose that $A\notin \Psi^-_{i-1}$.
Condition $A\in\Psi_i^-=(\Psi_{i-1}[\hat Y_i])^-$ then
 implies that there exists $a\in A$ such that $\hat Y_i$ hits $a$ in $\Psi_{i-1}$, i.e.\ $C_{sa}\cap \langle\Psi_{i-1}\rangle\cap \hat Y_i=\{a\}$.
 We then have $A=\br{a}\in\Psi_{i-1}$.

First, suppose that $a\notin Y_{i}$. Then $a\in  \langle\Pi_{i}\rangle$
(since $a\in\hat Y_{i}$ and $Y_{i}=\hat Y_{i}-\langle\Pi_{i}\rangle$) and hence $A\subseteq \langle\Pi_i\rangle$ and 
$\langle A^\downarrow \rangle\subseteq \langle\Pi_i\rangle$, i.e.\ $(\ast)$ holds.
We thus assume from now on that $a\in Y_i$.
We claim that $C_{sa}\cap Y_i=\{a\}$.
Indeed, suppose there exists $v\in (C_{sa}\cap Y_i)-\{a\}$.
We have $C_{sa}\cap \langle \Psi_{i-1}\rangle\cap \hat Y_i=\{a\}$ 
and thus $v\notin\langle \Psi_{i-1}\rangle$.
Furthermore, $v\in C_{sa}\cap X=\langle A^\downarrow \rangle$ and $A\in\Psi_{i-1}$.
These conditions imply that there exists $B\in \Psi_{i-1}^-$ with $B\preceq A$ and $v\in \langle B^\downarrow\rangle$.
By the induction hypothesis, $\langle B^\downarrow \rangle\subseteq \langle\Pi_{i-1}\rangle$
and so $v\in\langle\Pi_{i-1}\rangle$.
By the choice of $Y_i$ at line 3 we thus have $v\notin Y_i$, which is a contradiction.
 
Now consider the $i$-th step of Algorithm~\ref{alg:SSPQ}. 
Let $v_1,\ldots,v_{k-1},a,\ldots$ be the sorting of $Y_i$.
Since the first stage is successful and nodes in $Y_i$ are sorted according to $\mu_{i-1}(\cdot)$,
we have $f(s,v_1)\ge \ldots f(s,v_{k-1})\ge f(s,a)$.
Thus, the precondition of Lemma~\ref{lemma:SC}(b) holds for node $v_k=a$,
and hence set $\calC^\ast_i$ contains $C_{sa}$. We have $C_{sa}\cap X=\langle A^\downarrow \rangle\subseteq \langle R^\downarrow \rangle$
and hence $|C_{sa}\cap X|\le L$, therefore $C_{sa}$ will be added to $\Pi_i$ at line 7 and hence $(\ast)$ holds.

 We have proved claim $(\ast)$. Let us use it for $i=2N$.
 We have  $\Psi_{2N}=\varnothing$, so  $\Psi_{2N}^-=\{R\}$
and thus $\langle R^\downarrow  \rangle \subseteq \langle\Pi_{2N}\rangle$, as claimed.

\end{proof}

It remains to construct sequence $\alpha_1,\ldots,\alpha_{2N}$ such that 
the preconditions of Lemmas~\ref{lemma:first-stage} and~\ref{lemma:second-stage} hold w.h.p..
For that we use the same technique as in Section~\ref{sec:FixedSourcePartition1},
which is based on Lemma~\ref{lemma:expectation-reduction}.
Applying the same argument as in Corollary~\ref{cor:AHGASFAS} gives
the following result, and concludes the proof of Theorem~\ref{th:FixedSourcePartition}.

\begin{corollary}
Denote $n=|X|$.
For every fixed polynomial $p(n)$ there exists a sequence $\alpha_1,\ldots,\alpha_{2N}$ with $N=\Theta(\log^3 n)$
such that the preconditions of Lemmas~\ref{lemma:first-stage} and~\ref{lemma:second-stage} hold with probability at least $1-\tfrac 1 {p(n)}$.
\end{corollary}


\subsection{Overall algorithm} \label{sec:select-s}
We now need to show how to implement line 4 of Algorithm~\ref{alg:GH'}
for given graph $H$ and subset $X\subseteq V_H$ with $|X|\ge 2$.
One possibility is to use Algorithm~\ref{alg:select-s1}
where the call ${\tt FixedSourcePartition}_1(s,Y;H')$ is replaced with ${\tt FixedSourcePartition}(s,Y,L;H')$
for a sufficiently large $L$ (namely $L\ge |X|$); clearly, in this case the outputs of these two calls
have the same guarantee. Below we describe an alternative procedure whose complexity is smaller by a factor of $(\log n)\cdot (\log \log n)$.
It uses $L=|X|/2$, which means that large cuts are discarded inside the call ${\tt FixedSourcePartition}(s,Y,L;H')$.
This allows to simplify the selection of the source $s\in X$: it can now be selected uniformly at random (as in~\cite{Abboud:FOCS20}).

\begin{algorithm}[H]
\setcounter{AlgoLine}{0}

  \DontPrintSemicolon
    \While{\tt true}
    {
    	sample $s\in X$ uniformly at random \\
    	apply random permutation to $H$ as in Proposition~\ref{prop:GASF} to get graph $H'$ \\
    	call $\Pi\leftarrow {\tt FixedSourcePartition}(s,X-\{s\},L;H')$ where $L\eqdef |X|/2$ \\
    	remove non-maximal subsets from $\Pi$\hspace{25pt} \tcp{\small \em this is to simplify the analysis}
    	if $|(V_H-\langle\Pi\rangle)\cap X| \le L$ then terminate and return $(s,\Pi)$
    }
      \caption{Line 4 of Algorithm~\ref{alg:GH'}. 
      }\label{alg:select-s}
\end{algorithm}

\begin{lemma}
Each iteration of Algorithm~\ref{alg:select-s} succeeds with probability $\Omega(1)$
(and thus the expected number of iterations is $O(1)$).
\end{lemma}
\begin{proof}
\begin{sloppypar} 
We use for the analysis total order $\sqsubseteq$ on $X$ defined in Section~\ref{sec:select-s1}.
Let $a$ be the $L$-th smallest node in $X$ w.r.t.\ $\sqsubseteq$.
With probability $\Omega(1)$ the following events will hold jointly:
(i)~$a\sqsubset s$; 
(ii)~$H'$ has the unique GH tree (and thus the outputs of ${\tt CertifiedOrderedCuts}$ will be minimal);
(iii)~${\tt FixedSourcePartition}(s,X,L;H')$ returns $\Pi$ with $\langle \Pi\rangle=\langle \Pi_s^{(L)}\rangle$.
In that case for any $v\in U\eqdef \{v\in X\::\:v\sqsubseteq a\}$ we have $v\sqsubset s$,
or equivalently $C_{sv}\prec C_{vs}$.
This implies that $|C_{sv}\cap X|\le |C_{vs}\cap X| \le |(V-C_{sv})\cap X|$, and so $|C_{sv}\cap X|\le L$,
$C_{sv}\in \Pi_s^{(L)}$ and hence $v\in C_{sv}\subseteq\langle \Pi_s^{(L)}\rangle=\langle \Pi\rangle$.
We  have $U\subseteq \langle \Pi\rangle$ and $|U|=L$,
and so $|(V-\langle\Pi\rangle)\cap X| \le |(V-U)\cap X| \le L$.
\end{sloppypar}
\end{proof}

By using the same reasoning as in Section~\ref{sec:select-s1}
and noticing that one call to ${\tt CertifiedOrderedCuts}$
involves one call to {\tt OrderedCuts} and $O(\log n)$ calls to a minimum cut procedure (by Theorem~\ref{th:IsolatingCuts}),
we obtain
\begin{theorem}
The expected complexity of Algorithm~\ref{alg:GH'} with line 4 implemented as in Algorithm~\ref{alg:select-s} is
$O((t_{\tt OC}(n,m)+t_{\tt MC}(n,m)\log n)\cdot \log^4 n)$.
\end{theorem}


\section{Proofs for Sections~\ref{sec:implementation} and~\ref{sec:impl}}\label{sec:implementation:proofs}

\subsection{Proof of Lemma~\ref{lemma:OC-add-node}}

We will need the following well-known fact about minimum cuts.
It follows, for example, from the properties of the parametric maxflow problem~\cite{Gallo:SICOMP89}.
\begin{lemma}\label{lemma:parametric} 
Consider graph $G$ and pairs of disjoint subsets $(S,T)$ and $(S',T')$ with $S\subseteq S'$ and $T\supseteq T'$.
If $U$ is a minimum $S$-$T$ cut then there exists a minimum $S'$-$T'$ cut $U'$ with $U'\subseteq U$.
\end{lemma}

We proceed with the proof of Lemma~\ref{lemma:OC-add-node}.
\renewcommand{\thelemmaRESTATED}{\ref{lemma:OC-add-node}}
\begin{lemmaRESTATED}[restated]
Suppose that $(\X,\calE)$ is an {\tt OC} tree
for $\varphi u$ and $v\in\varphi$ is the parent of $u$ in $(\varphi u,\calE)$.
Suppose further that $(\X^{-u},\calE^{-u})$ is an {\tt OC} tree for $(\varphi,G)$.
Then $(\X,\calE)$ is valid for~$G$ if and only if set $[u]$ is a minimum $v$-$u$ cut in $G[[v]^{-u};v]$.
\end{lemmaRESTATED}

\begin{proof}
Clearly, we have ${[w]}^\downarrow={[w]^{-u}}^\downarrow$ for all $w\in\varphi$ and ${[w]}^\downarrow=[u]$ for $w=u$.
Therefore, $(\X,\calE)$ is valid for~$G$ if and only if $[u]$ is a minimum $\varphi$-$u$ cut in $G$.
Let $T$ be the minimal minimum $\varphi$-$u$ cut in $G$;
since   $\varphi \cap [v]^{-u}=\{v\}$,
it suffices to show that $T\subseteq[v]^{-u}$.

Let us write $\varphi=\alpha v \ldots$. We claim that $T\subseteq [v]^\downarrow={[v]^{-u}}^\downarrow$.
Indeed, assume that $|\alpha|\ge 1$ (otherwise $[v]^\downarrow=V$ and the claim is trivial).
Since $[v]^\downarrow$ is a minimum $\alpha$-$v$ cut
and $u\in [v]^\downarrow$, set $[v]^\downarrow$ is also a minimum $\alpha$-$\{v,u\}$ cut.
Applying Lemma~\ref{lemma:parametric} gives that $T\subseteq [v]^\downarrow$, as desired.

Now consider node $w\in\varphi$ with $w\prec v$. Let us write $\varphi = \alpha v \beta w \ldots$.
Let $U$ be a minimum $\alpha v \beta w$-$u$ cut.
We know that $[w]^\downarrow={[w]^{-u}}^\downarrow$ is a minimum $\alpha v \beta$-$w$ cut.
Define $A=[w]^\downarrow-U$ and $B=U-[w]^\downarrow$.
By symmetry of cuts and submodularity, we have
\begin{eqnarray*}
{\tt cost}([w]^\downarrow)+{\tt cost}(U) &=&{\tt cost}([w]^\downarrow)+{\tt cost}(V-U) \\
&\ge & {\tt cost}([w]^\downarrow \cap (V-U))+{\tt cost}([w]^\downarrow\cup (V-U)) \\
&=& {\tt cost}(A) + {\tt cost}(V-B) \;\;=\;\; {\tt cost}(A) + {\tt cost}(B)
\end{eqnarray*}
Clearly, we have $\alpha v\beta\cap A=\alpha v\beta w\cap B=\varnothing$, $w\in A$ and $u\in B$.
Set $A$ is an $\alpha v\beta$-$w$ cut, and thus ${\tt cost}(A)\ge {\tt cost}([w]^\downarrow)$.
This implies that ${\tt cost}(B)\le {\tt cost}(U)$, and thus  $B$ is a minimum $\alpha v\beta w$-$u$ cut.
Lemma~\ref{lemma:parametric} gives that  $T\subseteq B$.
We have $B\cap[w]^\downarrow=\varnothing$ and thus $T\cap [w]^\downarrow=\varnothing$.
Since this holds for all $w\in\varphi$ with $w\prec v$ (and since $T\subseteq[v]^\downarrow$),
we conclude that $T\subseteq [v]^{-u}$, as desired.

\end{proof}

\subsection{Proof of Lemma~\ref{lemma:pistar}}

\renewcommand{\thelemmaRESTATED}{\ref{lemma:pistar}}
Part {\rm (b)} follows directly from part {\rm (a)} and Lemma~\ref{lemma:Goldberg}, so we focus on proving part {\rm (a)}.
Consider {\tt OC} tree $(\Omega,\calE)$ for sequence  $\varphi$ with $|\varphi|\ge 2$,
and let $t\in \varphi$ be a leaf node in $(\varphi,\calE)$.
We say that $t$ is a {\em free leaf in $(\varphi,\calE)$}
if the following holds: if $\varphi=\ldots t \ldots u \ldots$ then $t$ and $u$ have different parents.
In particular, the last node of $\varphi$ is a free leaf if $|\varphi|\ge 2$.
We will show the following result.
\begin{lemma}\label{lemma:free-leaf}
Suppose that $(\Omega,\calE)$ is an {\tt OC} tree for $(\varphi,G)$ and $t\in\varphi$ is a free leaf in $(\varphi,\calE)$.
Then $(\Omega^{-t},\calE^{-t})$ is an {\tt OC} tree for $(\varphi^{-t},G)$.
\end{lemma}
\begin{proof}
We use induction on $|\varphi|$.
If $|\varphi|=2$ then the claim is trivial; suppose that $|\varphi|>2$.
Clearly, for any $w\in\varphi^{-t}$ we have ${[w]^{-t}}^\downarrow=[w]^\downarrow$.

Let $u$ be the last node of $\varphi$.
It can be seen that $(\Omega^{-u},\calE^{-u})$ is an {\tt OC} tree for $\varphi^{-u}$
(if $\varphi^{-u}=\alpha w\ldots$ then set ${[w]^{-u}}^\downarrow=[w]^\downarrow$ is indeed a minimum $\alpha$-$w$ cut).
Thus, if $t=u$ then the claim holds. Suppose that $\varphi=\ldots t \ldots u$.
Note that $t,u$ are leaf nodes with distinct parents.
Denote $\varphi^{-tu}=(\varphi^{-u})^{-t}=(\varphi^{-t})^{-u}$
and $(\Omega^{-tu},\calE^{-tu})=((\Omega^{-u})^{-t},(\calE^{-u})^{-t})=((\Omega^{-t})^{-u},(\calE^{-t})^{-u})$.
Clearly, $t$ is a free leaf in $(\varphi^{-u},\calE^{-u})$,
and thus $(\Omega^{-tu},\calE^{-tu})$ is an {\tt OC} tree for $(\varphi^{-tu},G)$ by the induction hypothesis.
Let $v$ be the parent of $u$ in $(\varphi,\calE)$.
We have $[u]=[u]^{-t}$ and $[v]^{-u}=[v]^{-tu}$ since $t$ is a leaf node which is not a child of~$v$.

 $(\Omega^{-u},\calE^{-u})$ is an {\tt OC} tree for $(\varphi^{-u},G)$
and  $(\Omega,\calE)$ is an {\tt OC} tree for $(\varphi,G)$.
By Lemma~\ref{lemma:OC-add-node}, $[u]$ is a minimum $v$-$u$ cut in $G[[v]^{-u};v]$.

 $(\Omega^{-tu},\calE^{-tu})$ is an {\tt OC} tree for $(\varphi^{-tu},G)$
 and $[u]^{-t}$ is a minimum $v$-$u$ cut $G[[v]^{-tu};v]$.
 By Lemma~\ref{lemma:OC-add-node}, $(\Omega^{-t},\calE^{-t})$ is an {\tt OC} tree for $(\varphi^{-t},G)$.

\end{proof}

We are now ready to prove Lemma~\ref{lemma:pistar}.
Let $(\Omega,\calE)$ be an {\tt OC} tree for $(\varphi,G)$, and consider non-root node $u\in\varphi$.
Let us repeatedly modify $(\varphi,\Omega,\calE)$ using the following
algorithm: while $(\varphi,\calE)$ has a free leaf $t$ with $t\notin \pi^\ast(u)\cup\{u\}$,
pick arbitrary such $t$ and update $(\varphi,\Omega,\calE)\leftarrow(\varphi^{-t},\Omega^{-t},\calE^{-t})$.
By construction, $u$ is always in $\varphi$. We claim that the sequence $\pi^\ast(u)$ never changes.
Indeed, it suffices to show that each step preserves $\pi(u)$;
applying this fact to $\pi(\pi(u)), \pi(\pi(\pi(u))),\ldots$ will then yield the claim.
It suffices to consider the case when $u$ comes after $t$, i.e. $\varphi=s \ldots t \ldots u \ldots$.
Let $v$ be the parent of $u$.
We say that node $w\in\varphi$ is {\em $u$-admissible in $(\varphi,\calE)$} if it satisfies
the following conditions:
(i)~$w\sqsubset u$; 
(ii)~$w\in\{v\}\cup\{w\::\:wv\in\calE\}$.
Recall that $\pi(u)$ is the maximal $u$-admissible node  $w$ in $(\varphi,\calE)$ w.r.t. $\sqsubseteq$.
We cannot have $w=t$, since $t$ is a free leaf. 
It can be checked that a node $w\ne t$ is $u$-admissible in $(\varphi,\calE)$
if and only if it is $u$-admissible in $(\varphi^{-t},\calE^{-t})$; this implies the claim.

Denote $\psi=\pi^\ast(u)u$.
We will show next that upon termination we have $\varphi=\psi$; clearly, this will imply
Lemma~\ref{lemma:pistar}(a).
Suppose the claim is false, then $\varphi=s \ldots t \alpha$ where $t\notin \psi$ and $\alpha\subseteq \psi$.
If node $t$ has a child $w$ then $w\in\psi$ and thus $t\in\pi^\ast(w)\subseteq \psi$ - a contradiction.
Thus, $t$ is a leaf in $(\varphi,\calE)$. It is not a free leaf since the algorithm has terminated,
thus there exists $w\in\alpha$ such that $t$ and $w$ have the same parent. But then $t\in \pi^\ast(w)\subseteq\psi$ - a contradiction.

\subsection{Proof of Lemma~\ref{lemma:divide-and-conquer:one}}

\renewcommand{\thelemmaRESTATED}{\ref{lemma:divide-and-conquer:one}}
\begin{lemmaRESTATED}[restated]
Consider sequence $\varphi=\alpha\ldots$ in graph $G$ with $|\alpha|\ge 1$.
Suppose that $(\X^\circ,\calE^\circ)$ is an {\tt OC} tree for $(\alpha,G)$,
and for each $v\in\alpha$ pair 
$(\X^v,\calE^v)$ is an {\tt OC} tree for $(\varphi\cap [v]^\circ,G[[v]^\circ;v])$.
Then $(\X,\calE)$ is an {\tt OC} tree for $(\varphi,G)$
where 
$\X=\bigcup_{v\in\alpha}\Omega^v$
and $\calE=\calE^\circ\cup\bigcup_{v\in\alpha}\calE^v$.
\end{lemmaRESTATED}

We use induction on $|\varphi|$. If $\varphi=\alpha$ then the claim is trivial.
Suppose that $\varphi=\alpha\beta u$. Let $v\in\alpha$ be the unique node with $u\in[v]^\circ$,
and denote $\gamma=\beta\cap [v]^\circ$.
We will work with pairs $(\Omega^{-u},\calE^{-u})$ and $((\X^v)^{-u},(\calE^v)^{-u})$,
which are {\tt OC} trees for sequences $\alpha\beta$ and $v\gamma$ respectively.
Clearly, we have $[w]=[w]^v$ for $w\in\gamma u$
and $[w]^{-u}=([w]^v)^{-u}$ for $w\in\gamma$
(where $([w]^v)^{-u}$ is the component of $(\Omega^v)^{-u}$ to which $w$ belongs).
Let $p\in v\gamma$ be the parent of $u$ in $(\varphi,\calE)$
(equivalently, in $(v\gamma u,\calE^v)$). Denote $H=G[[v]^\circ;v]$.
By Lemma~\ref{lemma:free-leaf}, $((\X^v)^{-u},(\calE^v)^{-u})$ is an {\tt OC} tree for $(v\gamma,H)$.
Thus, the induction hypothesis gives that $({\Omega}^{-u},{\calE}^{-u})$
is an {\tt OC} tree for $(\alpha\beta,G)$.
By Lemma~\ref{lemma:OC-add-node},  $[u]^v=[u]$ is a minimum $p$-$u$ cut in $H[{([p]^v)}^{-u};p]=G[[p]^{-u};p]$.
By Lemma~\ref{lemma:OC-add-node}, $(\Omega,\calE)$ is an {\tt OC} tree for $(\varphi,G)$.

\subsection{Proof of Lemma~\ref{lemma:divide-and-conquer:two}}

\renewcommand{\thelemmaRESTATED}{\ref{lemma:divide-and-conquer:two}}
\begin{lemmaRESTATED}[restated]
Consider sequence $\varphi=s\alpha\ldots$ in graph $G$ with $|\alpha|\ge 1$.
Let $(S,T)$ be a minimum $s$-$\alpha$ cut in $G$, and let $T_s=T\cup\{s\}$.
Suppose that $(\Omega',\calE')$ is an {\tt OC} tree for $(\varphi\cap S,G[S;s])$
and $(\Omega'',\calE'')$ is an {\tt OC} tree for $(\varphi\cap T_s,G[T_s;s])$.
Then $(\X,\calE)$ is an {\tt OC} tree for $(\varphi,G)$
where $\Omega=\{[s]'\cup[s]''\}\cup (\Omega'-\{[s]'\}) \cup (\Omega''-\{[s]''\})$
and $\calE=\calE'\cup\calE''$.
\end{lemmaRESTATED}

First, let us consider node $u\in\varphi\cap S$ with $u\ne s$. We can write $\varphi=s\alpha\beta u\ldots$, since $\alpha\cap S=\varnothing$.
We need to show that $[u]^\downarrow$ is a minimum $s\alpha\beta$-$u$ cut.
If $|\alpha|=1$ then the claim follows from Lemma~\ref{lemma:divide-and-conquer:one}
(since $(\Omega^\circ,\calE^\circ)=(\{S,T\},\{\alpha s\})$ is an {\tt OC} tree for $(s\alpha,G)$).
The general case can be reduced to the case above by contracting nodes in $\alpha$ to a single node.
Such transformation preserves set $S$ in a minimum $s$-$\alpha$ cut $(S,T)$;
thus, {\tt OC} tree $(\Omega',\calE')$ for $(\varphi\cap S,G[S;s])$ is not affected,
and therefore set $[u]^\downarrow$ is also not affected.

Now consider node $u\in\varphi\cap T$. Let us write $\varphi= \beta u \ldots$, and denote $\beta\cap S=A$ and $\beta\cap T=B$.
For brevity, we denote $X_1\ldots X_k=X_1\cup \ldots \cup X_k$ for disjoint subsets $X_1,\ldots,X_k$;
if $X_i=\{x_i\}$ then we write $x_i$ instead of $X_i$.
We make two claims:
\begin{itemize}
\item There exists a minimum $AB$-$u$ cut $(S',T')$ with $S\subseteq S'$.
Indeed, $(S,T)$ is a minimum $s$-$(\alpha\cup\{u\})$ cut (since $u\in T$).
The claim now holds by Lemma~\ref{lemma:parametric}.
\item $[u]^\downarrow$ is a minimum $SB$-$u$ cut in $G$.
Indeed, $(\Omega'',\calE'')$ is an {\tt OC} tree for $(\varphi\cap T_s,G[T_s;s])$,  thus set $[u]^\downarrow=[u]''^\downarrow$ is a minimum $sB$-$u$ cut in $G[T_s;s]$.
This is equivalent to the claim above.
\end{itemize}
These claims imply that $[u]^\downarrow$ is a minimum $AB$-$u$ cut, or equivalently a minimum $\beta$-$u$ cut.

\subsection{Proof of Theorem~\ref{th:random-permutation}: costructing {\tt OC} tree for random permutations}\label{sec:permutations}
We will use the following algorithm.

\begin{algorithm}[H]
\setcounter{AlgoLine}{0}

  \DontPrintSemicolon
    if $|\varphi|\le 2$ then compute {\tt OC} tree $(\X,\calE)$ for $(\varphi,G)$ non-recursively and return $(\X,\calE)$\!\!\!\!\!\!\!\!\! \\
    split $\varphi=\alpha\beta$ where $|\alpha|=\left\lceil\tfrac 12(|\varphi|+1)\right\rceil$ \\
    call $(\Omega^\circ,\calE^\circ)\leftarrow{\tt OrderedCuts}(\alpha;G)$, set $(\X,\calE)=(\varnothing,\calE^\circ)$ \\
    \For{$v\in \alpha$}
    {
    	compute minimum $v$-$(\beta\cap [v]^\circ)$ cut $(S,T)$ in $G[[v]^\circ;v]$ such that $T$ is minimal \\
    	let $(\Omega^v,\calE^v)\leftarrow{\tt OrderedCuts}(v(\beta\cap T),G[T\cup\{v\};v])$ \\
    	update $\Omega:=\Omega\cup \{S\cup[v]^v\} \cup (\Omega^v-[v]^v)$ and $\calE:=\calE\cup \calE^v$
    }
    return $(\X,\calE)$
      \caption{${\tt OrderedCuts}(\varphi;G)$: divide-and-conquer algorithm. 
      }\label{alg:divide-and-conquer}
      
\end{algorithm}

The correctness follows by an induction
argument and by the previous lemmas:
at line 7
 $ (\{S\cup[v]^v\} \cup (\Omega^v-[v]^v),\calE^v)$ is an {\tt OC} tree for $G[[v]^\circ;v]$ by the last case of Lemma~\ref{lemma:divide-and-conquer:two}
 applied to sequence $v (\beta \cap T)$,
 and the output at line 8 is an {\tt OC} tree for $(\varphi,G)$ by Lemma~\ref{lemma:divide-and-conquer:one}.

Computations performed by the algorithm can be represented by a rooted tree whose 
nodes have the form $\sigma=(\varphi;G)$. For such node we denote $\varphi_\sigma=\varphi$, and define $H_\sigma=(V_\sigma,E_\sigma)$ to be the graph obtained from $G$
by removing node $s$ and incident edges, where $s$ is the first node of sequence $\varphi$.
The leaves $\sigma$ of this tree satisfy $|\varphi_\sigma|\le 2$.
Each non-leaf node $\sigma$ has a child $\tau$ corresponding to the recursive call at line 3
(we call it the {\em left child}), and children $\{\tau^v\::\:v\in\alpha\}$ corresponding to recursive calls at line 6
(we call them {\em right children}).
Let us assign  label $\lambda(\sigma)=(\lambda_1(\sigma),\ldots,\lambda_d(\sigma))\in \{0,1\}^\ast$ according to the following rules:
(i) the label of the root node is the empty string;
(ii) if $\sigma$ is a non-leaf node then its left child is assigned label $(\lambda(\sigma),0)$
and its right children are assigned label $(\lambda(\sigma),1)$.
Let $d_\sigma$ be the depth of node $\sigma$, then $d_\sigma=|\lambda(\sigma)|$.
The set of labels that appears during the algorithm will be denoted as $\Lambda\subseteq\{0,1\}^\ast$,
and for a label $\mu\in \Lambda$ let $\Sigma_\mu$ be the set of nodes $\sigma$ with $\lambda(\sigma)=\mu$.

\begin{lemma}\label{lemma:GALSG} The execution of ${\tt OrderedCuts}(\varphi;G)$ satisfies the following properties. \\
{\rm (a)} The maximum depth satisfies $d_{\max}\le \log_2|\varphi|+O(1)$,
and $|\Lambda|\le O(|\varphi|)$. \\
{\rm (b)} Consider non-leaf node $\sigma$, and let $\tau$ be its left child and $\tau_1,\ldots,\tau_k$ be its right children.
Then $(V_\tau,E_\tau)=(V_\sigma,E_\sigma)$, $V_{\tau_1},\ldots,V_{\tau_k}$ are disjoint subsets of $V_\sigma$,
and $E_{\tau_1},\ldots,E_{\tau_k}$ are disjoint subsets of $E_\sigma$. \\
{\rm (c)} For each  $\mu\in\Lambda$ sets in $\{V_\sigma\::\:\sigma\in\Sigma_\mu\}$ and in $\{E_\sigma\::\:\sigma\in\Sigma_\mu\}$ are disjoint. 
\end{lemma}
\begin{proof}
By construction, if $\tau$ is a child of $\sigma$ then $|\varphi_\tau|\le |\varphi_\sigma|/2+1$ if $|\varphi_\sigma|$ is even,
and $|\varphi_\tau|\le (|\varphi_\sigma|+1)/2$ if $|\varphi_\sigma|$ is odd.
Let $d_{\max}$ be the maximum depth. For each non-leaf node $\sigma$ we have $|\varphi_\sigma|\ge 2^{d_{\max}-d_\sigma-1}+2$,
therefore $d_{\max}\le \log_2|\varphi|+O(1)$ where $\varphi$ is the input string.
This implies that $|\Lambda|\le 1+2^1+\ldots+2^{d_{\max}}\le O(|\varphi|)$, and proves part (a).

Part (b) follows from the algorithm's construction. Part (c) follows from part (b) combined with a straightforward induction argument.
\end{proof}
Lemma~\ref{lemma:GALSG} gives the first part of Theorem~\ref{th:random-permutation}, i.e.\ that the  complexity of Algorithm~\ref{alg:divide-and-conquer}
is $O(|\varphi|\cdot t_{\tt MC}(n,m))$. We now focus on the second part.
 Let us define a {\em generalized sequence} $\varphi$ as a set of nodes 
together with a partial order~$\sqsubseteq$ on this set. With some abuse
of terminology the set of nodes in $\varphi$ will be denoted simply as $\varphi$.
We say that $\varphi$ is a {\em $0$-sequence} if $\sqsubseteq$ does not impose any relations on $\varphi$,
and $\varphi$ is {\em $1$-sequence} if there exists $s\in\varphi$ such that $s\sqsubseteq v$ for all $v\in\varphi-\{s\}$
and $\sqsubseteq$ does not impose any relations on $\varphi-\{s\}$. 
If $\varphi$ is a $1$-sequence then we will write $\varphi=s\varphi'$ where $\varphi'$ is a $0$-sequence.

Let us define a randomized algorithm ${\tt OrderedCuts}'(\varphi;G)$ that takes as an input $1$-sequence $\varphi=s\varphi'$, graph $G$
and returns pair $(\Omega,\calE)$ which is an {\tt OC} tree for some sequence $sv_1,\ldots,v_\ell$
such that $\{v_1,\ldots,v_\ell\}=\varphi'$ (see Algorithm~\ref{alg:divide-and-conquer'}).
It differs from ${\tt OrderedCuts}(\varphi;G)$ only in line 2:
in Algorithm~\ref{alg:divide-and-conquer} the order of nodes in $\varphi$ is fixed,
while Algorithm~\ref{alg:divide-and-conquer'} first chooses a partial order on $\varphi'$
and then performs the same steps as Algorithm~\ref{alg:divide-and-conquer}.


\begin{algorithm}[H]
\setcounter{AlgoLine}{0}

  \DontPrintSemicolon
    if $|\varphi|\le 2$ then compute {\tt OC} tree $(\X,\calE)$ for $(\varphi,G)$ non-recursively and return $(\X,\calE)$\!\!\!\!\!\!\!\!\! \\
    sample a random subset $A\subseteq \varphi'$ of size $|A|=\left\lceil\tfrac 12(|\varphi|+1)\right\rceil-1$,
let $\alpha$ be the $1$-sequence with $\alpha=sA$, and let $\beta$ be a $0$-sequence with $\beta=\varphi-\alpha$
\\
    call $(\Omega^\circ,\calE^\circ,\sqsubseteq^\circ)\leftarrow{\tt OrderedCuts}(\alpha;G)$, set $(\X,\calE)=(\varnothing,\calE^\circ)$ \\
    \For{$v\in \alpha$}
    {
    	compute minimum $v$-$(\beta\cap [v]^\circ)$ cut $(S,T)$ in $G[[v]^\circ;v]$ such that $T$ is minimal \\
    	let $(\Omega^v,\calE^v,\sqsubseteq^v)\leftarrow{\tt OrderedCuts}'(v(\beta\cap T),G[T\cup\{v\};v])$ \\
    	update $\Omega:=\Omega\cup \{S\cup[v]^v\} \cup (\Omega^v-[v]^v)$ and $\calE:=\calE\cup \calE^v$
    }
    return $(\X,\calE)$
      \caption{${\tt OrderedCuts}'(\varphi;G)$: divide-and-conquer algorithm for 1-sequence $\varphi=s\varphi'$.\!\!\!\!\!\!\!\!\! 
      }\label{alg:divide-and-conquer'}
      
\end{algorithm}

We claim that we can prove the second part of Theorem~\ref{th:random-permutation}
by analyzing the expected complexity of ${\tt OrderedCuts}'(\varphi;G)$ for a given $1$-sequence $\varphi$.
Indeed, let us augment procedure ${\tt OrderedCuts}'(s\varphi';G)$ so
that it refines a partial order on elements of $\varphi'$ as it runs.
In the beginning all elements of $\varphi'$ are incomparable (as $\varphi'$ is a 0-sequence).
Whenever line 2 chooses  decomposition $\tilde \varphi'=AB$ of the current 0-sequence $\tilde\varphi'$
(possibly at a deeper level of recursion),
we refine the partial order on $\varphi'$ by setting
$a\sqsubseteq b$ for all $(a,b)\in A\times B$.
(It can be checked by induction that before that moment
elements in $A\cup B$ were pairwise-incomparable).
When the algorithm terminates, let us select a total order on $\varphi'$
by uniformly sampling from all total orders consistent with the current partial order.
Clearly, the augmented procedure generates a total order on $\varphi'$
which is a random permutation of the original 0-sequence $\varphi'$.
Furthermore, the output of ${\tt OrderedCuts}'(s\varphi';G)$ is identical
to the output of ${\tt OrderedCuts}(s\varphi';G)$ with this permutation.
The claim follows.

Next, we analyze the expected complexity of ${\tt OrderedCuts}'(\varphi;G)$ for a given $1$-sequence $\varphi$.
We use the same notation and terminology for analyzing ${\tt OrderedCuts}'(\varphi;G)$ that we defined in the beginning of this section for
${\tt OrderedCuts}(\varphi;G)$. We make the following claim about the execution of  ${\tt OrderedCuts}'(\varphi;G)$.

\begin{lemma}\label{lemma:PAOISHGASFADASGAKGHAKSHDBAJAKJDFJADFBAJDF}
\begin{sloppypar}
{\rm (a)} Consider non-leaf node $\sigma$, and let $\tau_1,\ldots,\tau_k$ be its right children.
Then  $\mathbb E[\sum_{i\in[k]} |V_{\tau_i}|]\le \tfrac 12 |V_\sigma|$ and
$\mathbb E[\sum_{i\in[k]} |E_{\tau_i}|]\le \tfrac 12 |E_\sigma|$.
\end{sloppypar}
\noindent{\rm (b)} For each $\mu\in\Lambda$ we have 
$\mathbb E[|V_\mu|]\le n\cdot \left(\tfrac 12\right)^{||\mu||}$ and
$\mathbb E[|E_\mu|]\le m\cdot \left(\tfrac 12\right)^{||\mu||}$ 
where $V_\mu=\bigcup_{\sigma\::\:\lambda(\sigma)=\mu} V_\sigma$, $E_\mu=\bigcup_{\sigma\::\:\lambda(\sigma)=\mu} E_\sigma$,
$n=|V_G|$, $m=|E_G|$, and $||\mu||$ is the number of $1$'s in $\mu$.
\end{lemma}
\begin{proof}
Part (b) is a straightforward consequence of part (a) (and Lemma~\ref{lemma:GALSG}(b)), so we focus on proving part (a).
Consider the call ${\tt OrderedCuts}'(s\varphi';G)$ with graph $G=(V,E,w)$ corresponding to node $\sigma$.
Define $\ell=|\varphi'|$, $p=\left\lceil\tfrac 12(|\varphi|+1)\right\rceil-1$ and $q=\ell-p$, then $p\ge q$.
Define sets $A=\alpha-\{s\}$ and $B=\beta$.
Recall that $A$ is a random subset of $\varphi'$ of size $p$, and $B=\varphi'-A$.

Let $A'$ be a random subset of $A$ of size $q$,
and let $T^\ast$ be the minimal minimum $A'$-$B$ cut in $G$.
We claim that 
$\mathbb E[|V_{G[T^\ast]}|]\le \tfrac 12 |V|$
and $\mathbb E[|E_{G[T^\ast]}|]\le \tfrac 12 |E|$.
Indeed, for disjoint
 subsets $X,Y\subseteq \varphi'$ of size $q$ let $T^\ast_{X,Y}$ be the minimal minimum $X$-$Y$ cut in $G$,
Clearly, sets $T^\ast_{X,Y}$ and $T^\ast_{Y,X}$ are disjoint.
Conditioned on the event $\{A',B\}=\{X,Y\}$, we have $(A',B)=(X,Y)$ with prob. $\tfrac 12$
and $(A',B)=(Y,X)$ with prob. $\tfrac 12$.
Therefore,
$$
\mathbb E[|V_{G[T^\ast]}|\:\:|\:\:\{A',B\}=\{X,Y\}]
= \tfrac 12 |V_{G[T^\ast_{X,Y}]}| + \tfrac 12 |V_{G[T^\ast_{Y,X}]}| \le \tfrac 12 |V|
$$
Summing this over pairs $(X,Y)$ with appropriate probabilities gives
the desired bound on $\mathbb E[|V_{G[T^\ast]}|]$.
The bound on $\mathbb E[|E_{G[T^\ast]}|]$ can be obtained in a similar way.

Now consider node $v\in\alpha$, and let $(S,T)$ be the cut computed at line 5 for node $v$.
$(S,T)$ can be equivalently defined as the minimum $A_v$-$B_v$ cut in $G$ with minimal $T$
where $A_v=(V-[v]^\circ)\cup\{v\}$ and $B_v=B\cap[v]^\circ$.
We have $A'\subseteq A_v$ and $B_v\subseteq B$, so by Lemma~\ref{lemma:parametric} we have $T\subseteq T^\ast$.
This implies $H_{\tau_v}$ is a subgraph of $G[T^\ast]$
where $\tau_v$ is the right child of $\sigma$ corresponding to node $v$,
and thus proves the claim.

\end{proof}

Using Lemma~\ref{lemma:PAOISHGASFADASGAKGHAKSHDBAJAKJDFJADFBAJDF}(b),
we can bound the expected total number of nodes in edges as follows:
$$
\sum_{\mu\in\Lambda} \mathbb E[|V_\mu|] \le n\cdot \sum_{\mu\in\Lambda} \left(\tfrac 12\right)^{||\mu||}
\qquad\qquad
\sum_{\mu\in\Lambda} \mathbb E[|E_\mu|] \le m\cdot \sum_{\mu\in\Lambda} \left(\tfrac 12\right)^{||\mu||}
$$
Theorem~\ref{th:random-permutation} now follows from the following calculation:
\begin{align*}
\sum_{\mu\in\Lambda} \left(\tfrac 12\right)^{||\mu||}
&\le\sum_{d=0}^{d_{\max}}\sum_{\mu\in\{0,1\}^d} \left(\tfrac 12\right)^{||\mu||}
=\sum_{d=0}^{d_{\max}}\sum_{k=0}^d \binom{d}{k} \left(\tfrac 12\right)^{k} 
=\sum_{d=0}^{d_{\max}}(1+\tfrac 12)^d \\
&=2((\tfrac 32)^{d_{\max}+1}-1)
\le (\tfrac 32)^{\log_2 n + O(1)}=O(n^{\log_2 (3/2)})
\end{align*}

\subsection{Analysis of the practical implementation} \label{sec:practical:proof:complexity}\label{sec:alg:proofs}

In this section we give a (naive) bound on the complexity of 
the OC algorithm described in Section~\ref{sec:impl}.

\begin{theorem}
Alg.~\ref{alg:GH'} with line 4 implemented as in Alg.~\ref{alg:practical:SSPQ}
for graph $G$ with ${\tt size}(G)=(n,m)$ satisfies ${\tt size}(OC)=(O(n^2), O(nm))$.
\end{theorem}
\begin{proof}
Recall that lines 6-9 of Alg.~\ref{alg:GH'} process sets in $\Pi$ starting from the minimal sets in $\Pi$.
Let us consider a modification that at line 6 picks a maximal set $S\in\Pi$ instead of a minimal set.
Processing such set will create supernodes $A=X-S$ and $B=X\cap S$. Consider family of sets $\Pi'=\{S'\in \Pi\::\:S'\subsetneq S\}$.
It can be seen that $\Pi'$ is a laminar family of sets with the property that each $S'\in\Pi'$ is a minimum $v$-$t$ cut where
$v,t$ are the maximal nodes w.r.t.\ $\preceq$ in $B$ and in $S'\cap X$, respectively.
Let us further modify the execution as follows: when processing supernode $B$, we
first run steps 6-9 for the laminar family $\Pi'$ defined above (if it is non-empty).
Clearly, this gives an execution that is equivalent to the original execution
in the following sense: it uses the same graphs $H$ on which the {\tt OrderedCuts} procedure is called, possibly in a different order.
Note that in the new execution some supernodes may be split without calling {\tt OrderedCuts}
(using instead the laminar family $\Pi'$ provided by the parent call). We will refer to the iterations defined above as {\em OC iterations} and {\em dummy iterations}, respectively.
(An iteration is one pass through the loop at lines 3-9).

By repeatedly applying such modification we obtain a valid execution of Algorithm~\ref{alg:GH'} 
such that every family $\Pi$ produced at line 4 is a set of disjoint subsets.
We can now apply a standard argument.
Each supernode $X$ that appears during the execution of the modified algorithm can be assigned a depth using natural rules:
(i) the initial supernode $V$ is at depth $0$;
(ii) if supernode $X$ at depth $d$ is split into supernodes $X_0,\ldots,X_r$ then the latter supernodes are assigned depth $d+1$.
Let $X_1,\ldots,X_{r(d)}$ be the supernodes at depth $d$,
and let $H_1,\ldots,H_{k(d)}$ be the corresponding auxiliary graphs.
By the results in Section~\ref{sec:select-s1}, the total size of graphs $H_1,\ldots,H_{k(d)}$ is $(O(n),O(m))$.
Clearly, the maximum depth cannot exceed $n$ (since the size of each supernode is smaller than the size of its parent).
The claim follows.
\end{proof}

\begin{theorem}
The total size of graphs on which the maxflow problem is solved during the call
${\tt OrderedCuts}(\varphi;G)$ (excluding terminals 
and their incident edges) is $(O(n^2 \log n),O(nm\log n))$.
\end{theorem}
\begin{proof}
Given a pair $(\varphi,G)$ with $\varphi=s\ldots $ and $G=(V,E,w)$, we denote $\dot G$ to be the graph obtained from $G$ by removing $s$ and all incident edges.
Each call to ${\tt OrderedCuts}(\varphi;G)$ with $|\varphi|>2$ makes two recursive calls
${\tt OrderedCuts}(\varphi';G')$ (at line 3) and ${\tt OrderedCuts}(\varphi'';G'')$ (at lines 5 or 8).
We say this call performs a {\em trivial split} (or it is a {\em trivial call}) if $S=\{s\}$ and $k>1$,
in which case $G'$ has a single vertex and $G''=G$.
If we remove all trivial calls, we obtain a sequence of computations that still fits the structure of procedure {\tt OrderedCuts} as defined in Section~\ref{sec:intro}
(and each call is non-trivial).
Let us analyze the size of maxflow graphs for such sequence.
Clearly, we $|\dot V'|\le |\dot V|-1$ and $|\dot V''|\le |\dot V|-1$.
Thus, the maximum depth will be at most $n$. A straightforward induction argument shows
that graphs $\dot G$ at a fixed depth are disjoint. Therefore, the total size of graphs $\dot G$ is at most $(n^2,nm)$.

Now let us consider trivial calls.
A consecutive sequence of such calls may have a length of at most $O(\log n)$
(since after each such call the value of $\bar k$ is halved).
Thus, to each trivial call we can assign a non-trivial call on
the same graph such that at most $O(\log n)$ trivial calls are assigned to each non-trivial call.
This proves the theorem.
\end{proof}

\end{document}